\documentclass[a4paper]{article}
\input{packages.sty}
\input{commands.sty}

\begin{document}

\author{\begin{minipage}{\columnwidth}\centering\normalsize
	Bjoern Andres$^{a,b,1}$, Silvia Di Gregorio$^{a,c}$, Jannik Irmai$^a$,\\
	Lucas Fabian Naumann$^a$ and Shengxian Zhao$^{a,b}$
	\end{minipage}}
\date{\begin{minipage}{\columnwidth}\centering\normalsize
	$^a$Faculty of Computer Science, TU Dresden, Germany\\
	$^b$Center for Scalable Data Analytics and AI (ScaDS.AI) Dresden/Leipzig, Germany\\
	$^c$Université Sorbonne Paris Nord, CNRS, Laboratoire d'Informatique de Paris Nord, LIPN, F-93430 Villetaneuse, France\end{minipage}}
\title{\Large\bf The canonical facets of multi-separator polytopes}
\maketitle

\footnotetext[1]{Correspondence: \texttt{bjoern.andres@tu-dresden.de}}
\setcounter{footnote}{1}

\section*{Abstract}

We initiate a polyhedral study of the graph multi-separator problem proposed by \citet{irmai-2023-separator} as an alternative to the lifted multicut problem for application to the task of image segmentation.
Starting with an integer linear program (ILP) formulation and the multi-separator polytope spanned by its feasible solutions, we characterize in terms of efficiently-decidable, graph-theoretic conditions all facets induced by inequalities of the ILP.
We proceed by strengthening these inequalities and describing additional facets of some multi-separator polytopes induced by the stronger inequalities.
Specifically, we obtain a totally dual integral description of the multi-separator polytope for paths in the case where separation is considered for all vertex pairs.
Finally, we relate the multi-separator polytope to the boolean quadric polytope, showing that facets induced by odd-cycle inequalities do not transfer generally, and to the lifted multicut polytope, showing that either polytope is a projection of a face of the other.

\setlength{\cftbeforesecskip}{0ex}
\renewcommand{\cftsecfont}{\normalfont}
\renewcommand{\cftsecpagefont}{\normalfont}
\tableofcontents

\section{Introduction}

The fundamental mathematical objects in this article are \emph{connectors} and \emph{separators} of a graph\footnote{Unless otherwise noted, all graphs $G = (V, E)$ in this article are finite, simple and undirected, i.e., $E \subseteq \tbinom{V}{2}$.}:
For any two vertices $a$,~$b$ of a graph $G$, we say that $a$ and $b$ are \emph{connected} in $G$ if and only if there exists a path in $G$ from $a$ to $b$.
A \emph{connected component} of $G$ is a maximal subgraph of $G$ in which any two vertices are connected.
We call a subset $U \subseteq V$ an \emph{$\{a,b\}$-connector} of $G$, and say that $a$ and $b$ are \emph{connected} in $G$ by $U$, if and only if $a$ and $b$ are contained and connected in the subgraph $G[U]$ of $G$ induced by $U$.
Dually, we call a subset $S \subseteq V$ an \emph{$\{a,b\}$-separator} of $G$, and say that $a$ and $b$ are \emph{separated} in $G$ by $S$, if and only if $V \setminus S$ is not an $\{a,b\}$-connector of $G$, i.e., either $\{a, b\} \cap S \neq \emptyset$ or $a$,~$b$ are in different connected components of $G[V \setminus S]$.
This definition adopted from \citet{goering-2000} is slightly uncommon in that an $\{a,b\}$-separator may contain $a$ or $b$.

The multi-separator problem is a combinatorial optimization problem whose feasible solutions are the vertex subsets $S$ of a graph $G$.
In addition to the edges of \(G\), the model includes long-range interactions between prescribed vertex pairs, which reward or penalize whether these vertices remain connected after the removal of \(S\).

Given costs for vertices and vertex pairs, the objective is to minimize the sum of the costs of the vertices in $S$ and those vertex pairs that are separated by $S$ in $G$:

\begin{definition}[\citet{irmai-2023-separator}]
\label{definition:multi-separator-problem}
For any graph $G = (V, E)$, 
any set $F \subseteq \tbinom{V}{2}$ of vertex pairs called \emph{interactions} on $G$, and any $c \in \mathbb{R}^{V \cup F}$, the instance of the \emph{(minimum cost) multi-separator problem} has the form \eqref{eq:ms} written below in which $F(S)$ denotes the set of those $f \in F$ that are separated by $S$ in $G$.
\begin{align}
\label{eq:ms}
\min_{S \subseteq V} \;\;
\sum_{v \in S} c_v
+
\sum_{\mathclap{f \in F(S)}} c_f
\tag{\ms}
\end{align}
\end{definition}

Interactions can be edges, i.e.~vertex pairs in $F \cap E$, or non-edges, i.e.~vertex pairs in $F \setminus E$.
Costs assigned to edges (respectively, non-edges) penalize or reward the separation of neighboring (respectively, non-neighboring) vertices.
Notice that $E$ and $F$ play different roles here: 
While the separation of any vertex pair is determined by $E$ and $S$ alone, costs are assigned to the separation of vertex pairs in $F$.
See also \Cref{fig:ms}.

\begin{figure}
    \centering
    \input{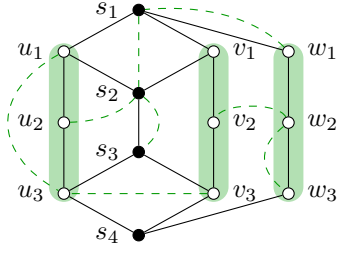}
    \caption{
        Depicted above are a graph $G$ (solid black lines) and a set $F$ of vertex pairs called interactions (dashed green lines) that includes the edges $\{s_1, w_1\}$,~$\{s_2, s_3\}$, $\{w_2, w_3\}$ and the non-edges $\{u_1, u_3\}$,~$\{u_2, s_2\}$,~$\{s_1, s_2\}$,~$\{u_3, v_3\}$, $\{v_2, w_2\}$.
        Depicted in addition are a vertex subset $S = \{s_1, s_2, s_3, s_4\}$ (solid black vertices) and the connected components of the graph $G[V \setminus S]$ induced by the complement of $S$ (green areas).
        All interactions except $\{u_1, u_3\}$ and $\{w_2, w_3\}$ are separated by $S$ in $G$, i.e., they belong to $F(S)$.
		In the \emph{multi-separator problem}, vertex separation is considered for all interactions, whether these are edges or non-edges.
    }
    \label{fig:ms}
\end{figure}

The multi-separator problem generalizes the quadratic unconstrained binary optimization problem \citep{padberg1989boolean}, the node-weighted Steiner tree problem \citep{moss-2007} and the multi-terminal vertex separator problem \citep{cornaz2019multi}, according to \citet[Thm.~3--5]{irmai-2023-separator}.
\citet{irmai-2023-separator} introduce it as an alternative to the lifted multicut problem \citep{andres-2023-a-polyhedral} for application to the tasks of image segmentation 
\citep{keuper-2015a, beier-2017, keuper-2017, keuper-2020}.
While the two problems are reducible to one another in linear time \citep[Thm.~1 and 2]{irmai-2023-separator}, the multi-separator problem is less complex than the lifted multicut problem for certain cost patterns \citep[Thm.~8 and 9]{irmai-2023-separator}.

In this article, we state the multi-separator problem in the form of an integer linear program (ILP) and characterize in terms of efficiently decidable, graph-theoretic conditions the canonical facets of the associated multi-separator polytopes, i.e., those facets of multi-separator polytopes that are induced by inequalities of the ILP.
Next, we strengthen inequalities of the ILP and describe additional facets of some multi-separator polytopes induced by the stronger inequalities.
In particular, we obtain a totally dual integral description of multi-separator polytopes of paths in the case where separation is considered for all vertex pairs.
Finally, we relate the multi-separator polytope to the boolean quadric polytope, showing that facets induced by odd-cycle inequalities do not transfer generally, and to the lifted multicut polytope, showing that either polytope is a projection of a face of the other.

The article is organized as follows.
\Cref{section:related-work} reviews related work on vertex separators, with emphasis on polyhedral aspects.
\Cref{section:preliminaries} recalls some basic facts about connectors and separators and their generalizations to pairs of vertex subsets, which are used in our polyhedral analysis of the multi-separator polytope.
\Cref{section:ilp} presents an ILP formulation of the multi-separator problem, from which the multi-separator polytope is defined.
In \Cref{section:dimension}, we begin our study of the facial structure of multi-separator polytopes by  determining their dimension.
\Cref{section:facets-box} is devoted to the study of box inequalities.
\Cref{section:facets-connector} and \Cref{section:facets-separator} give complete characterizations of the facets of multi-separator polytopes defined by connector inequalities and separator inequalities, respectively.
\Cref{section:facets-path,section:facets-intersection} present two additional classes of facets of multi-separator polytopes, which are induced by inequalities that are not part of the ILP formulation.
\Cref{section:related-polytopes} discusses the relationships of the multi-separator polytope with two noticeable relatives: the boolean quadric polytope and the lifted multicut polytope.
\Cref{section:conclusion} concludes the article.
Efficient separation algorithms for inequalities of the ILP formulation are given in \Cref{section:separation}.
A technical lemma that we use for constructing basis vectors of faces of multi-separator polytopes is put into \Cref{section:lemma-ie}, together with its proof.
A technical lemma that we use to establish the sufficiency of conditions for an inequality to be facet-defining is discussed in \Cref{section:pseudoforest}, together with its proof.
All other proofs can be found in \Cref{section:proofs}.


\section{Related work}
\label{section:related-work}

Many combinatorial optimization problems related to vertex separators have been proposed in the literature, capturing different aspects of tasks we encounter in the real world, and many researchers have contributed to the analysis of the associated polytopes:
Notably, \citet{balas2005vertex} initiate a polyhedral study of the problem of finding a partition of the vertex set of a connected graph $G$ into three disjoint subsets $A$,~$B$ and~$C$, with $A$ and $B$ non-empty and bounded by an integer depending only on the order of $G$, such that $C$ is an $\{A, B\}$-separator and its weight is minimized.
After stating the problem in the form of a mixed integer linear program (MILP), they define the vertex separator polytope and investigate its dimension and facial structure.
Based on the MILP and the valid inequalities described there, a branch-and-cut algorithm for this problem is developed in a companion paper \cite{souza2005vertex}.
Another class of valid inequalities for the vertex separator polytope is proposed by \citet{didi2011exact}.

A generalization of this problem has also been studied:
Given positive constants $\beta$ and $\kappa$, the generalized vertex separator problem asks for a vertex subset such that the graph obtained by deleting it has at most $\beta$ parts, each of cardinality at most $\kappa$.
A polyhedral investigation of this problem is conducted by \citet{borndorfer1998decomposing} from the perspective of matrix decomposition.

Given a vertex weighted graph $G$ and a positive integer, the $k$-separator problem consists in finding a minimum-weight vertex subset whose removal leads to a graph where the size of each connected component is less than or equal to $k$.
This problem is studied by \citet{oosten2007disconnecting} and \citet{ben2015k} from a polyhedral point of view, based on different ILP formulations.
The ILP proposed by \citet{oosten2007disconnecting} is similar to ours but does not include what we call separator inequalities (see \Cref{lemma:ilp} below), which are an essential part of the defining inequalities of the multi-separator problem.
Related is the vertex $k$-cut problem that asks for a vertex $k$-cut of minimum weight.
Here, a vertex $k$-cut is a vertex subset whose deletion disconnects the given graph into at least $k$ connected components.
Several ILP formulations are introduced by \citet{cornaz2019vertex} and \citet{furini2020integer}, along with families of valid inequalities.

Another variant is the so-called multi-terminal vertex separator problem, in which a set of terminal vertices are prescribed, the task is to determine a minimum weight vertex subset whose deletion leads to a disconnected graph whose each connected component contains precisely one terminal vertex.
The polyhedral structure related to an ILP formulation is studied by \citet{cornaz2019multi} who develop from this a branch-and-cut algorithm.
The multi-separator problem we study here assigns costs not only to vertices but also to vertex pairs, which incorporates additional non-local information.
Moreover, we do not impose any constraints on the number of connected components, their size or the size of separators; instead, these are properties of feasible solutions.

The multi-separator problem is closely related also to other combinatorial optimization problems regarding polyhedral aspects:
The boolean quadric polytope studied by \citet{padberg1989boolean} is a special case of the multi-separator polytope where only adjacent vertex pairs are taken into account.
We discuss this relation in \Cref{section:relatives-bqp}.
The lifted multicut polytope studied by \citet{hornakova-2017} and \citet{andres-2023-a-polyhedral} is associated with the lifted multicut problem, introduced by \citet{keuper-2015a} as a generalization of the multicut problem \cite{chopra1993partition}, that can be seen as an edge version of the multi-separator problem.
We discuss this relation in \Cref{section:relatives-lmp}.
The cut inequalities that occur in an ILP formulation of the lifted multicut problem \citep{andres-2023-a-polyhedral} resemble the separator inequalities we introduce for the multi-separator problem.
Interestingly, these two classes of seemingly similar inequalities have different properties: 
The problem of deciding whether a cut inequality defines a facet of the lifted multicut polytope is NP-hard \cite{naumannbox}.
In contrast, we show in \Cref{section:facets-separator} that the problem of deciding whether a separator inequality defines a facet of the multi-separator polytope can be solved in polynomial time.

The Steiner cut inequalities provide an ILP formulation of the Steiner tree problem \cite{aneja1980integer,chopra1994steiner} and have a similar form as the separator inequalities.
Despite this similarity, the conditions for separator inequalities to be facet-defining are completely different from those for the Steiner cut inequalities, which are always facet-defining \cite{chopra1994steiner}.

There are some works \citep{alvarez2013maximum, wang2017imposing, rehfeldt2022optimal, fischetti2017thinning, carvajal2013imposing} on separator inequalities in the context of studying connectivity of graphs.
Typically, such a separator inequality has the form $x_a + x_b \leq 1 + \sum_{s \in S} x_{s}$, where $a$,~$b$ are two vertices in a graph $G$ and $S$ is an $ab$-separator of $G$ with $S \cap \{a, b\} = \emptyset$.
Note that though this class of separator inequalities looks similar to our separator inequalities introduced in \Cref{lemma:ilp}, their corresponding facial structures differ substantially.
In particular, \citet{wang2017imposing} studied the connected subgraph polytope of a graph and showed that the above separator inequality is facet-defining if and only if $S$ is a minimal $ab$-separator.
This is in contrast to the facet-defining separator inequalities for the multi-separator problem, whose structure is much more complicated (see \Cref{section:facets-separator}).

Toward structural properties of vertex separators, Menger's theorem \cite{menger1927allgemeinen} determines the minimum size of vertex separators in terms of the maximum number of disjoint paths between any pair of vertex subsets.
Moreover, the celebrated planar separator theorem \cite{lipton1979separator} establishes an asymptotic bound on the size of a separator whose removal splits the given graph into smaller pieces, and has been generalized to other classes of graphs \cite{gilbert1984separator,kawarabayashi2010separator}.
Furthermore, Escalante has discovered a natural partial order on the set of all minimal separators with respect to any two fixed non-adjacent vertices, with which this set becomes a lattice.
See \citet{escalante1972schnittverbande,escalante1974note} for the original papers and \citet{halin1993lattices} for a survey.


\section{Preliminaries}
\label{section:preliminaries}

In this section, we fix elementary terms and notation.
Let $G = (V, E)$ be a graph.
For any $U \subseteq V$, we define the \emph{open neighborhood} $N_{G}(U)=\{v \in V\setminus U \mid \exists u \in U: \{u,v\} \in E\}$ of $U$ in $G$ as the set of all vertices not in $U$ but adjacent to some vertex of $U$.
We define the \emph{closed neighborhood} $N_{G}[U]$ of $U$ in $G$ as the union of $U$ and $N_G(U)$.
For any distinct $a, b \in V$, we use the shorthand notation $ab$ for the unordered pair $\{a,b\}$.
Since we interpret an edge or interaction as a two-element set of vertices, for any $f \in \tbinom{V}{2}$ with endvertices $a$ and $b$, i.e., $f = \{a, b\}$, an $f$-connector (respectively, $f$-separator) of $G$ means an $\{a, b\}$-connector (respectively, $\{a, b\}$-separator) of $G$.

We adopt a generalization of connectors and separators from \citet{goering-2000} who uses it for a short proof of Menger's theorem \cite{menger1927allgemeinen}.
Let $A, B \subseteq V$.
We call a subset $U \subseteq V$ an \emph{$\{A, B\}$-connector} of $G$ if and only if there exist $a \in A$ and $b \in B$ such that $U$ is an $\{a, b\}$-connector of $G$.
Dually, we call a subset $S \subseteq V$ an \emph{$\{A, B\}$-separator} of $G$ if and only if for all $a \in A$ and all $b \in B$, $S$ is an $\{a, b\}$-separator of $G$.
Clearly, $S$ is an $\{A, B\}$-separator of $G$ if and only if $V \setminus S$ is not an $\{A, B\}$-connector of $G$, i.e., every $\{A, B\}$-connector of $G$ intersects $S$.
For any pair $a$,~$b$ of vertices of $G$, $\{\{a\}, \{b\}\}$-connectors/separators are exactly $\{a, b\}$-connectors/separators.
So, by abuse of notation, we may safely identify them whenever no confusion arises.

\begin{figure}
    \centering
    \input{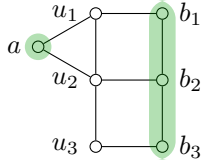}
    \caption{
        Depicted above is a graph $G$, along with two vertex subsets, $A = \{a\}$ and $B = \{b_1, b_2, b_3\}$ (in green).
        There are three minimal $\{A, B\}$-connectors of $G$: $\{a, u_1, b_1\}$, $\{a, u_2, b_2\}$ and~$\{a, u_2, u_3, b_3\}$.
        All the minimal $\{A, B\}$-separators of $G$ are given by $\{a\}$, $\{u_1, u_2\}$, $\{b_1, b_2, b_3\}$, $\{u_2, b_1\}$, $\{u_1, u_3, b_2\}$, $\{b_1, b_2, u_3\}$ and $\{u_1, b_2, b_3\}$.
        Notice that $\{u_1, u_2\}$ is the only minimal $\{A, B\}$-separator of $G$ that is disjoint from $A$ and $B$.
    }
    \label{fig:connector-separator}
\end{figure}

Of particular interest are connectors and separators that are
minimal.
A \emph{minimal $\{A, B\}$-connector} of $G$ is an $\{A, B\}$-connector $U$ of $G$ such that $U$ is minimal among all $\{A, B\}$-connectors of $G$ with respect to set-theoretic inclusion.
Similarly, a \emph{minimal $\{A, B\}$-separator} of $G$ is an $\{A, B\}$-separator $S$ of $G$ that is minimal with respect to set-theoretic inclusion.
See \Cref{fig:connector-separator} for an illustration.
The following proposition generalizes characterizations of minimal connectors and minimal separators that are well-known for vertex pairs.
\begin{proposition}\label{prop:connector-separator-minimal}
    Let $A$,~$B$ and~$U$ be vertex subsets of a graph $G$.
    \begin{enumerate}[label=\textnormal{(\roman*)}]
    \item \label{item:prop-connector-minimal}
    $U$ is a minimal $\{A, B\}$-connector of $G$ if and only if $G[U]$ is an induced path $u_0, u_1, \dots, u_n$ for some $n \geq 0$, such that $U \cap A = \{u_0\}$ and $U \cap B = \{u_n\}$.
    \item \label{item:prop-separator-minimal}
    $U$ is a minimal $\{A, B\}$-separator of $G$ if and only if every minimal $\{A, B\}$-connector of $G$ intersects $U$, and for every \(u\in U\), there exists a minimal \(\{A,B\}\)-connector \(C\) of \(G\) such that $C\cap U=\{u\}$.
    \end{enumerate}
\end{proposition}
\begin{proof}
    See \Cref{section:proof:prop:connector-separator}.
\end{proof}
\begin{remark}
    Let $A$ and $B$ as in \Cref{prop:connector-separator-minimal}.
    Then, any singleton contained in $A \cap B$ is a minimal $\{A, B\}$-connector of $G$.
    Consequently, $A \cap B$ is a subset of any minimal $\{A, B\}$-separator of $G$.
\end{remark}
\begin{corollary}\label{cor:connector-separator-minimal}
    Let $a$ and $b$ be two vertices of a graph $G$ and let $U$ be a vertex subset of $G$.
    \begin{enumerate}[label=\textnormal{(\roman*)}]
        \item \label{item:cor-connector-minimal}
        $U$ is a minimal $\{a, b\}$-connector of $G$ if and only if $G[U]$ is an $\{a, b\}$-path.
        \item \label{item:cor-separator-minimal}
        $U$ is a minimal $\{a, b\}$-separator of $G$ if and only if every $\{a, b\}$-path $P$ in $G$ intersects $U$ and for each $u \in U$ there exists an $\{a, b\}$-path in $G$ intersecting $U$ at precisely $u$.
    \end{enumerate}
\end{corollary}
Given two vertices $a$ and~$b$ in a graph $G$, a minimal $\{a, b\}$-separator of $G$ is either a singleton whose unique element is one of the two endvertices $a$ and~$b$, or a vertex subset (possibly empty) whose complement in $G$ induces a disconnected subgraph with $a$ and~$b$ in different connected components.


\section{Multi-separator polytopes}
\label{section:polytope}

\subsection{Definition and outer approximation}
\label{section:ilp}

In this section, we reformulate the multi-separator problem in the form of an ILP and define the associated multi-separator polytope.

\begin{definition}
Let $G = (V, E)$ be a graph and let $F \subseteq \tbinom{V}{2}$.
Given a set $U$ of vertices, let $x^U$ be the characteristic vector of our objects of interest, vertices and interactions, that are separated by $V \setminus U$ in $G$.
Formally, $x^U$ is the vector in $\{0, 1\}^{V \cup F}$ defined such that the following conditions hold:
\begin{align*}
\forall v \in V \colon \quad
	& x^U_v = 0 \;\Leftrightarrow\; v \in U \,, \\
\forall f \in F \colon \quad
	& x^U_f = 0 \;\Leftrightarrow\; U \text{ is an $f$-connector of } G \,.
\end{align*}
We call
$
\feasible \coloneqq 
\left\{
	x^U \in \{0, 1\}^{V \cup F} \mid
	U \subseteq V
\right\}
$
the set of \emph{feasible vectors}.
Moreover, we call the convex hull
$
\polytope \coloneqq \conv \feasible
$
of $\feasible$ in the affine space $\mathbb{R}^{V \cup F}$ the \emph{multi-separator polytope} with respect to $G$ and $F$.
\end{definition}


The map $2^V \to \feasible \colon S \mapsto x^{V \setminus S}$ is a bijection from the feasible set of the multi-separator problem \eqref{eq:ms} to the set $\feasible$ of feasible vectors.
For any feasible solution $S \subseteq V$ to the multi-separator problem, its objective value is
\begin{align*}
\sum_{v \in S} c_v + \sum_{f \in F(S)} c_f
= 
\sum_{v \in V} c_v \, x^{V \setminus S}_v + \sum_{f \in F} c_f \, x^{V \setminus S}_f
= \langle c, x^{V \setminus S} \rangle
\enspace .
\end{align*}
Thus, $S \subseteq V$ is a solution to the multi-separator problem if and only if $x^{V \setminus S}$ is a solution to the problem $\min \; \{ \langle c, x \rangle \mid x \in \feasible\}$.

\begin{lemma}
\label{lemma:ilp}
For any graph $G = (V, E)$ and any $F \subseteq \tbinom{V}{2}$, the set $\feasible$ contains precisely those $x \in \mathbb{Z}^{V \cup F}$ that satisfy the linear system
\begin{align}
\forall f \in F \ 
\forall C \in f\textnormal{-connectors}(G) \colon \hspace{6.5ex}
x_{f} & \leq \sum_{v \in C} x_v 
\,,
\label{eq:connector}
\\
\forall f \in F \ 
\forall S \in f\textnormal{-separators}(G) \colon \quad
1 - x_{f} & \leq \sum_{v \in S} (1 - x_v)
\,,
\label{eq:separator}
\\
\forall g \in V \cup F \colon \hspace{8ex}
0 & \leq x_g \leq 1
\,,
\label{eq:box}
\end{align}
where $f\textnormal{-connectors}(G)$ and $f\textnormal{-separators}(G)$ denote the set of all $f$-connectors and all $f$-separators of $G$, respectively.
This still holds if we enforce \eqref{eq:connector} only for minimal connectors and enforce \eqref{eq:separator} only for minimal separators.
\end{lemma}

\begin{proof}
See \Cref{section:proof:lemma:ilp}.
\end{proof}

We call an inequality of the form~\eqref{eq:connector} a \emph{connector inequality},
one of the form~\eqref{eq:separator} a \emph{separator inequality}, and 
one of the form~\eqref{eq:box} a \emph{box inequality}.
Precisely these inequalities are said to be \emph{canonical}.
Intuitively, the connector (respectively, separator) inequalities encode the fact that if all the vertices of an $f$-connector (respectively, an $f$-separator) of $G$ are labeled zero (respectively, one) by a feasible vector, then so is the interaction $f$ itself.

\begin{corollary}
\label{corollary:ilp}
The polytope $\Xi'_{GF} = \{x \in \mathbb{R}^{V \cup F} \mid \eqref{eq:connector} \text{ and } \eqref{eq:separator} \text{ and } \eqref{eq:box}\}$ is an outer polytope of the multi-separator polytope, i.e., $\polytope \subseteq \Xi'_{GF}$, with the same integer points, i.e., $\feasible = \polytope \cap \{0,1\}^{V \cup F} = \Xi'_{GF} \cap \{0,1\}^{V \cup F}$.
\end{corollary}

\Cref{corollary:ilp} asserts that $\min \{\langle c, x \rangle \mid x \in \mathbb{R}^{V \cup F} \text{ and } \eqref{eq:connector} \text{ and } \eqref{eq:separator} \text{ and } \eqref{eq:box}\}$ is a linear relaxation of the multi-separator problem \eqref{eq:ms} with the same integer feasible solutions.
The separation problems for the canonical inequalities can be solved in polynomial time, and hence the linear relaxation.
For the box inequalities, efficient separation is trivial.
For the connector and separator inequalities, efficient separation algorithms are given in \Cref{section:separation}.


\subsection{Dimension}
\label{section:dimension}

In this section, we show that multi-separator polytopes of connected graphs with arbitrary interaction sets are full-dimensional.

\begin{proposition}\label{prop:dimension}
For any connected graph $G = (V, E)$ and any $F \subseteq \tbinom{V}{2}$, the multi-separator polytope $\polytope$ is full-dimensional, i.e., $\dim \polytope = \lvert V \cup F \rvert$.
\end{proposition}
\begin{proof}
	See \Cref{section:proof:prop:dimension}.
\end{proof}

\Cref{prop:dimension} implies that each facet of $\polytope$ is defined by a unique linear inequality (up to scaling by a positive constant).
In addition, \Cref{prop:dimension} implies the following characterization of facet-defining canonical inequalities:

\begin{proposition}\label{prop:canonical-facets}
	Let $G = (V, E)$ be a connected graph, let $F \subseteq \tbinom{V}{2}$, and let $ax \leq \alpha$ be a facet-defining inequality for $\polytope$ with $a \in \mathbb{Z}^{V \cup F}$ and $\alpha \in \mathbb{Z}$.
	Then $ax \leq \alpha$ is a canonical inequality up to positive scaling if and only if there is at most one $f \in F$ such that $a_f \neq 0$.
\end{proposition}
\begin{proof}
	See \Cref{section:proof:prop:canonical-facets}.
\end{proof}

We conclude this section by discussing how the multi-separator polytope on a graph can be constructed from the multi-separator polytopes on the connected components of that graph.

\begin{proposition}\label{prop:ms-disconnected}
	Let $G = (V, E)$ be any graph, let $F$ be any set of interactions on $G$, and let $\{G_i\}_{i \in I}$ be the collection of all connected components of $G$.
	For any $i \in I$, let $F_i \subseteq F$ contain those interactions whose both endvertices belong to $G_i$.
	Then, up to translation, $\polytope$ is the product of the multi-separator polytopes on its connected components, i.e.:
	\[
		\polytope \cong
		\mathbbm{1}_{F(\emptyset)}
		+
		\prod_{i \in I} \Xi_{G_iF_i}
		\,.
	\]
	In particular, suppose an inequality $ax \leq \alpha$ with $a \in \mathbb{R}^{V \cup F}$ and $\alpha \in \mathbb{R}$ is facet-defining for $\polytope$, then there is a unique $i \in I$ such that $a|_{V_i \cup F_i} x \leq \alpha - \sum_{f \in F(\emptyset)} a_f$ is facet-defining for $\Xi_{G_iF_i}$.
	Conversely, any facet-defining inequality of $\Xi_{G_iF_i}$ for some $i \in I$ is also facet-defining for $\polytope$.
\end{proposition}

\begin{proof}
	See \Cref{section:proof:prop:ms-disconnected}.
\end{proof}

In view of \Cref{prop:ms-disconnected}, there is no loss of generality in considering only connected graphs for a polyhedral study of the multi-separator problem.
This justifies our restriction to connected graphs in some of the previous and forthcoming statements.


\subsection{Facets induced by box inequalities}
\label{section:facets-box}

\begin{figure}
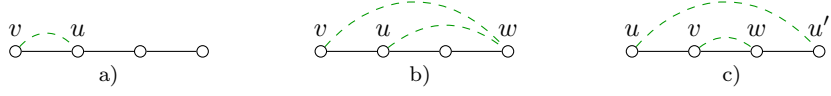

	\centering
	\begin{subfigure}[b]{0.23\textwidth}
		\centering
		\input{figures/box/v=0_1.tex}
		\caption{}
		\label{fig:box-v=0_1}
	\end{subfigure}
	\hspace{3ex}
	\begin{subfigure}[b]{0.23\textwidth}
		\centering
		\input{figures/box/v=0_2.tex}
		\caption{}
		\label{fig:box-v=0_2}
	\end{subfigure}
	\hspace{3ex}
	\begin{subfigure}[b]{0.23\textwidth}
		\centering
		\input{figures/box/f=1.tex}
		\caption{}
		\label{fig:box-f=1}
	\end{subfigure}
	\caption{Depicted above are graphs (black) together with interactions (green, dashed) where some condition of \Cref{thm:v=0-fd} or \Cref{thm:f=1-fd} is violated for a vertex or interaction.
	\subref{fig:box-v=0_1} Since $u$ is a neighbor of $v$ and $\{u, v\} \in F$, Condition~\ref{item:v=0-fd-1} of \Cref{thm:v=0-fd} is violated, and thus, the lower box inequality $x_v \geq 0$ is not facet-defining.
	\subref{fig:box-v=0_2} Since $u$ is a $\{v, w\}$-cut-vertex, Condition~\ref{item:v=0-fd-2} of \Cref{thm:v=0-fd} is violated, and thus, the lower box inequality $x_v \geq 0$ is not facet-defining.
	\subref{fig:box-f=1} Since both $v$ and $w$ are $\{u, u^{\prime}\}$-cut-vertices, the condition in \Cref{thm:f=1-fd} is violated, and thus, the upper box inequality $x_{\{v, w\}} \leq 1$ is not facet-defining.}
	\label{fig:box-not-facet}
\end{figure}

In this section, we characterize those box inequalities~\eqref{eq:box} that define a facet of a multi-separator polytope.
We begin by noting that some box inequalities may be implied (dominated) by other, stronger, valid linear inequalities, as demonstrated by the proposition below.

\begin{proposition}\label{prop:box}
	For any graph $G = (V, E)$, any $F \subseteq \tbinom{V}{2}$ and any $x \in \mathbb{R}^{V \cup F}$ satisfying inequalities~\eqref{eq:connector},~\eqref{eq:separator} and $x_f \leq 1$ for all $f \in F$, the following two statements hold:
	\begin{enumerate}[label=\textnormal{(\roman*)}]
		\item If $x_v \geq 0$ for all $v \in V$, then $x_f \geq 0$ for all $f \in F$.
		\item Suppose that $G$ does not contain isolated vertices and $E \subseteq F$, then $0 \leq x_v \leq 1$ for all $v \in V$.
	\end{enumerate}
\end{proposition}

\begin{proof}
	See \Cref{section:proof:prop:box}.
\end{proof}

As a corollary of \Cref{prop:box}, the lower box inequalities associated with interactions never define facets of the multi-separator polytope:

\begin{corollary}\label{cor:f=0-fd}
	For any connected graph $G=(V,E)$, any $F \subseteq \tbinom{V}{2}$ and any $f \in F$, the lower box inequality $x_f \geq 0$ associated with $f$ does not define a facet of $\polytope$.
\end{corollary}

In the following theorems, we state under which conditions the remaining box inequalities are facet-defining.
Examples depicted in \Cref{fig:box-not-facet} illustrate how the conditions of \Cref{thm:v=0-fd} or \Cref{thm:f=1-fd} can be violated for a vertex or interaction, and thus, the corresponding inequality does not define a facet.

\begin{theorem}\label{thm:v=0-fd}
	For any connected graph $G = (V, E)$, any $F \subseteq \tbinom{V}{2}$ and any $v \in V$,
    the lower box inequality $x_v \geq 0$ defines a facet of $\polytope$ if and only if the following two conditions hold:
	\begin{enumerate}[label=\textnormal{(\roman*)}]
		\item \label{item:v=0-fd-1}
		For any $u \in N_{G}(v)$, $\{u, v\} \not\in F$.		
		\item \label{item:v=0-fd-2}
		For any $u \in N_G(v)$ and any $w \in V \setminus \{u,v\}$ with $\{u,w\}, \{v,w\} \in F$, $u$ is not a $\{v,w\}$-cut-vertex\footnote{For any vertices $u, v, w$ of a graph $G$, we call $u$ a $\{v, w\}$-\emph{cut-vertex} of $G$ if and only if $\{u\}$ is a minimal $\{v, w\}$-separator of $G$.} of $G$.
	\end{enumerate}
\end{theorem}

\begin{proof}
	See \Cref{section:proof:thm:box:v=0}.
\end{proof}


\begin{theorem}\label{thm:f=1-fd}
    For any connected graph $G = (V, E)$, any $F \subseteq \tbinom{V}{2}$ and any $f \in F$,
    the upper box inequality $x_f \leq 1$ defines a facet of $\polytope$ if and only if there is no interaction \(f'\in F\setminus \{f\}\) such that $f$ is a pair of $f^\prime$-cut-vertices of $G$.
\end{theorem}

\begin{proof}
	See \Cref{section:proof:thm:box:f=1}.
\end{proof}

\begin{theorem}\label{thm:v=1-fd}
    For any connected graph $G = (V, E)$, any $F \subseteq \tbinom{V}{2}$ and any $v \in V$,
    the upper box inequality $x_v \leq 1$ defines a facet of $\polytope$ if and only if there is no interaction $f \in F$ such that $v$ is an $f$-cut-vertex of $G$.
\end{theorem}

\begin{proof}
	See \Cref{section:proof:thm:box:v=1}.
\end{proof}


\subsection{Facets induced by connector inequalities}
\label{section:facets-connector}

\begin{figure}
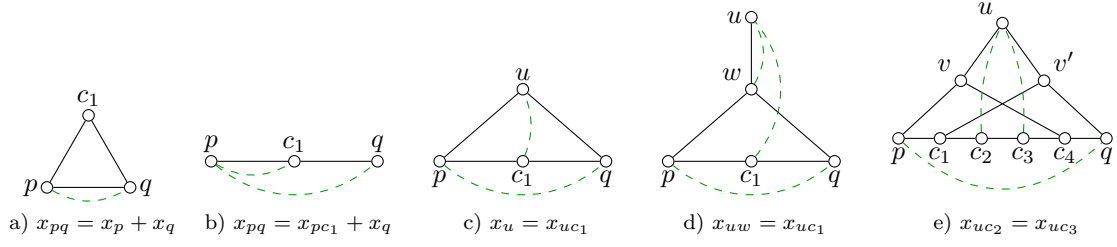

	\centering
	\begin{subfigure}[b]{0.15\textwidth}
		\centering
		\input{figures/connector/fd-a.tex}
		\caption{$x_{pq} = x_p + x_q$}
		\label{fig:connector-fd-a}
	\end{subfigure}
	\hfill
	\begin{subfigure}[b]{0.19\textwidth}
		\centering
		\input{figures/connector/fd-b.tex}
		\caption{$x_{pq} = x_{pc_1} + x_q$}
		\label{fig:connector-fd-b}
	\end{subfigure}
	\hfill
	\begin{subfigure}[b]{0.19\textwidth}
		\centering
		\input{figures/connector/fd-c.tex}
		\caption{$x_u = x_{uc_1}$}
		\label{fig:connector-fd-c}
	\end{subfigure}
	\hfill
	\begin{subfigure}[b]{0.19\textwidth}
		\centering
		\input{figures/connector/fd-d.tex}
		\caption{$x_{uw} = x_{uc_1}$}
		\label{fig:connector-fd-d}
	\end{subfigure}
	\hfill
	\begin{subfigure}[b]{0.23\textwidth}
		\centering
		\input{figures/connector/fd-e.tex}
		\caption{$x_{uc_2} = x_{uc_3}$}
		\label{fig:connector-fd-e}
	\end{subfigure}
	\caption{
		Depicted above are graphs $G = (V,E)$ (black, solid) and interactions $F$ (green, dashed) with $\set{p,q} \in F$.
		In each of these graphs, there exists a $\set{p,q}$-connector $C$ that violates one condition of \Cref{thm:connector-fd}.
		Thus, the connector inequality with respect to $\{p,q\}$ and $C$ is not facet-defining for $\polytope$.
		Written below each of the depicted graphs is an inequality satisfied by all feasible vectors for this graph.
		These inequalities are examples of the inequalities constructed in the proof of \Cref{thm:connector-fd} to establish necessity of the respective conditions.
		\subref{fig:connector-fd-a} Condition \ref{item:connector-fd-1} is violated for $C = \{p, c_1, q\}$ as $\{p,q\} \subset C$ is a $\set{p,q}$-connector.
		\subref{fig:connector-fd-b} Condition \ref{item:connector-fd-2} is violated for $C = \{p, c_1, q\}$ as $\{p, c_1\} \in F$.
		\subref{fig:connector-fd-c} Condition \ref{item:connector-fd-3} is violated for $C = \{p, c_1, q\}$ as $u$ is a $c_1$-bypass and $\{u, c_1\} \in F$.
		\subref{fig:connector-fd-d} Condition \ref{item:connector-fd-4} is violated for $C = \{p, c_1, q\}$ as $\{u, w\}, \{u, c_1\} \in F$ and $w$ is both a $c_1$-bypass and a $\{\{u\}, C\}$-cut-vertex.
		\subref{fig:connector-fd-e} Condition \ref{item:connector-fd-5} is violated for $C = \{p, c_1, c_2, c_3, c_4, q\}$ as $\{u, c_2\}, \{u, c_3\} \in F$ and the set $\{v, v'\}$ of $c_2$- and $c_3$-bypasses is a $\{\{u\}, C\}$-separator.
	}
	\label{fig:connector-conditions}
\end{figure}

In this section, we establish the exact condition under which a connector inequality~\eqref{eq:connector} defines a facet of a multi-separator polytope.
For this purpose, we introduce the notion of a \emph{bypass} vertex of a connector.
As the connector inequality associated with a non-minimal connector cannot be facet-defining (by \Cref{lemma:ilp}), we introduce bypass vertices only for minimal connectors, avoiding complications.

\begin{definition}\label{def:bypass}
	Let $G = (V, E)$ be a connected graph, let $F \subseteq \tbinom{V}{2}$, let $f \in F$ and let $C \subseteq V$ be a minimal $f$-connector of $G$.
	For any $v \in V$ and any $c \in C$, we call $v$ a $c$-\emph{bypass} of $C$ if and only if $v \notin C$ and $(C \setminus \{c\}) \cup \{v\}$ is an $f$-connector of $G$.
	For any $v \in V$, we call $v$ a \emph{bypass} of $C$ if and only if $v$ is a $c$-bypass of $C$ for some $c \in C$.
\end{definition}

Examples of bypasses are depicted in \Cref{fig:connector-conditions}~\subref{fig:connector-fd-c}~--~\subref{fig:connector-fd-e}.
Notice that the neighborhood of any bypass of \(C\) contains at least two vertices of \(C\).
With the definition of a bypass, we are now ready to state the conditions that characterize facet-defining connector inequalities:

\begin{theorem}\label{thm:connector-fd}
	For any connected graph $G = (V, E)$, 
	any $F \subseteq \tbinom{V}{2}$,
    any $f \in F$ and 
	any $f$-connector $C$ of $G$,
	the associated connector inequality
	\begin{equation}\label{eq:connector-fd}
		x_f \leq \sum_{v \in C} x_v
	\end{equation}
	defines a facet of the multi-separator polytope $\polytope$ if and only if all the following conditions hold:
	\begin{enumerate}[label=\textnormal{(\roman*)}]
		\item $C$ is a minimal $f$-connector of $G$.
		\label{item:connector-fd-1}
		\item There is no interaction $f^\prime \in F\setminus \{f\}$ such that $f^\prime \subseteq C$.
		\label{item:connector-fd-2}
		\item There is no interaction $uw \in F$ such that $u$ is a $w$-bypass of $C$.
		\label{item:connector-fd-3}
		\item There are no two distinct interactions $uw, uw^\prime \in F$ such that $w$ is both a $w^\prime$-bypass of $C$ and a $\{\{u\}, C\}$-cut-vertex\footnote{For any vertex $v$ and vertex subsets $A, B$ of a graph $G$, we call $v$ an $\{A,B\}$-cut-vertex of $G$ if and only if $\{v\}$ is a minimal $\{A,B\}$-separator of $G$.} of $G$.
		\label{item:connector-fd-4}
		\item There are no two distinct interactions $\{u, w\}, \{u, w^\prime\} \in F$ with $u \not\in N_G[C]$ and $w, w^\prime \in C$ such that
		$
			\left\{
                v \in N_G(C) \;\middle|\;
				v \ \textnormal{is both a} \ w \text{-bypass and a} \ w^\prime \textnormal{-bypass of} \ C
            \right\}
		$
		is a $\{\{u\}, C\}$-separator of $G$.
		\label{item:connector-fd-5}
	\end{enumerate}
\end{theorem}

\begin{proof}
	See \Cref{section:proof:thm:connector}.
\end{proof}

Non-trivial examples of connectors whose associated inequalities fail to define facets of $\polytope$ are depicted in \Cref{fig:connector-conditions}.
Moreover, Condition~\ref{item:connector-fd-2} of \Cref{thm:connector-fd} implies the following corollary.

\begin{corollary}
Consider the special case in which every edge of $G$ is an interaction in $F$, i.e.~$E \subseteq F$.
Now, for any interaction $f \in F$ and any minimal $f$-connector $C$ of $G$ containing more than two vertices, the corresponding connector inequality is not facet-defining. \end{corollary}

\subsection{Facets induced by separator inequalities}
\label{section:facets-separator}

\begin{figure}
	\centering
	\input{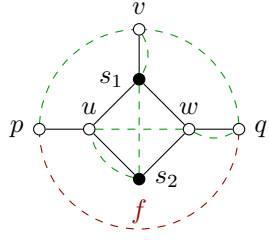}
	\caption{Depicted above are a graph $G$ (solid black), a designated interaction $f=\{p,q\}$ (dashed red), and other interactions $vs_1, vp, vq, us_2, uw, s_1s_2, qw$ (dashed green).
	The set $S=\{s_1, s_2\}$ is a minimal $f$-separator of $G$.
	All interactions except $qw$ are in $F(S)$.
	The subgraph $G(S, s_1)$ is induced by $\{p, u, s_1, v, w, q\}$.
	The subgraph $G(S, s_2)$ is induced by $\{p, u, s_2, w, q\}$.
	Moreover, $S(f) = S(uw) = S$, $S(vs_1) = S(vp) = S(vq) = \{s_1\}$, $S(us_2) = \{s_2\}$, and $S(s_1s_2) = \emptyset$.
	}
	\label{fig:separator-def}
\end{figure}

In this section, we characterize by graph properties those separator inequalities~\eqref{eq:separator} that define a facet of a multi-separator polytope.
In order to capture important features of the set of those feasible vectors that satisfy any given separator inequality at equality, we first introduce some additional concepts and notation.
Since the separator inequality associated with a non-minimal separator cannot be facet-defining (by \Cref{lemma:ilp}), we present the following definitions only for minimal separators, once again avoiding complications.

\begin{definition}\label{def:separator}
	Let $G = (V, E)$ be a connected graph, let $F \subseteq \tbinom{V}{2}$, let $f \in F$ and let $S \subseteq V$ be a minimal $f$-separator of $G$.
	\begin{enumerate}[label=\textnormal{(\roman*)}]
		\item For any vertex $s \in S$, let $G(S, s)$ denote the connected component of the subgraph of $G$ induced by $V \setminus (S\setminus \{s\})$ that contains $s$.
		\item For any interaction $f^{\prime} \in F(S)$, let $S(f^{\prime})$ denote the set of those vertices $s$ in $S$ for which $f^\prime$ is connected in $G(S, s)$.
		\item Call a vertex $v \in V$ \emph{degenerate} (with respect to $G$,~$F$,~$f$ and~$S$) if $v \in S$ or each $s \in S$ has at least one of the following three properties: 1.~$v \not\in G(S, s)$; 2.~$v \in G(S, s)$ and $\{v, s\} \in F$, 3.~$v \in G(S, s)$ and $v$ is an $f$-cut-vertex of $G(S, s)$.
		Otherwise, call $v$ \emph{non-degenerate} (with respect to $G$,~$F$,~$f$ and~$S$).
		\item We define $\HH(f, S)$ as the graph $(V \setminus S, F^{\prime})$, where $F'$ is the set of those interactions $f' \in F(S) \setminus \{f\}$ such that $f^{\prime} \cap S = \emptyset$ and for each $s \in S(f^{\prime})$, the interaction $f^{\prime}$ contains at least one $f$-cut-vertex of $G(S, s)$.
	\end{enumerate}
\end{definition}

\begin{remark}\label{remark:seprator-def-P}
	For a fixed minimal $f$-separator $S$ of $G$ as in \Cref{def:separator}, we can list all the elements of $G(S, s)$, $S(f^\prime)$ and $\HH(f, S)$ in polynomial time.
	Indeed, such a problem boils down to checking whether two vertices are connected in some induced subgraph of $G$.
	As a consequence, deciding whether a vertex is degenerate can also be done in polynomial time.
\end{remark}

Examples of $G(S, s)$ and $S(f^\prime)$ are shown in \Cref{fig:separator-def}.
Examples illustrating degenerate vertices and $\HH(f, S)$ are given in \Cref{fig:separator-not-facet}.
Note that an interaction that contains an $f$-cut-vertex of $G(S, s)$ for each $s \in S$ is not necessarily a minimal $f$-separator of $G$.

\begin{remark}\label{remark:separator-feasible-set}
	The main motivation for introducing $G(S, s)$ and $S(f^\prime)$ comes from the following fact: 
	For any vertex subset $U$ of a graph $G$, the feasible vector $x^{U}$ is in the face $\Sigma$ defined by the separator inequality associated with an interaction $f$ and an $f$-separator $S$ of $G$ if and only if either $U \cap S$ is empty, or $U \cap S$ is a singleton and $U$ is an $f$-connector of $G$.
	For an informal discussion, we call a vertex whose variable assumes the value 0 (respectively 1) \emph{unblocked} (respectively \emph{blocked}):
	For every feasible solution in $\Sigma$, the $f$-separator $S$ is either blocked, or unblocked at exactly one vertex such that $f$ is connected by an unblocked path.
	Thus, in order to stay in $\Sigma$, when a vertex $s$ in $S$ is unblocked, the vertices of $G(S, s)$ cannot be blocked or unblocked freely. 
	Instead, we need to maintain the existence of an unblocked path connecting $f$.
	On the one hand, this fact prevents us from constructing a basis of the associated linear space of $\Sigma$ in a straightforward way, like in the proof of \Cref{prop:dimension}.
	On the other hand, interactions connected in components of $G[V \setminus (S \setminus \{s\})]$ other than $G(S, s)$ are necessarily not in $F(S)$ and can thus be handled analogously to the proof of \Cref{prop:dimension}, as will be shown in detail below.
\end{remark}

\begin{remark}\label{remark:separator-inequality-minimal}
	As the $f$-separator $S$ is assumed to be minimal in \Cref{def:separator}, every $s \in S$ is such that $f$ is connected in $G(S, s)$, by \ref{item:cor-separator-minimal} of \Cref{cor:connector-separator-minimal}.
\end{remark}

\begin{remark}
	With the notation of \Cref{def:separator}, every $f$-cut-vertex of $G$ is degenerate.
\end{remark}

\begin{theorem}\label{thm:separator-fd}
	Let $G = (V, E)$ be a connected graph, let $F \subseteq \tbinom{V}{2}$, let $f \in F$ and let $S$ be an $f$-separator of $G$.
	The associated separator inequality
	\begin{equation}\label{eq:separator-fd}
		1 - x_{f} \leq \sum_{s \in S} (1 - x_{s})
	\end{equation}
	is facet-defining for the multi-separator polytope $\polytope$ if and only if all the following conditions are satisfied:
	\begin{enumerate}[label=\textnormal{(\roman*)}]
		\item \label{item:separator-fd-1}
		$S$ is a minimal $f$-separator of $G$.
		\item \label{item:separator-fd-2}
		For every interaction $f^{\prime} \in F(S) \setminus \{f\}$, there exists an $s \in S(f^{\prime})$ such that $f^{\prime}$ is not a pair of $f$-cut-vertices of $G(S, s)$.
		\item \label{item:separator-fd-3}
		If $v$ is a non-isolated vertex of $\HH(f, S)$ such that $S$ is an $\{f, \{v\}\}$-separator of $G$, then $v$ is a non-degenerate leaf\;\footnote{A \emph{leaf} of a graph is a vertex of degree $1$ in this graph.} of $\HH(f, S)$.
		\item \label{item:separator-fd-4}
		$\HH(f, S)$ is a forest in which every connected component contains at most one degenerate vertex.
	\end{enumerate}
\end{theorem}

\begin{proof}
	See \Cref{section:proof:thm:separator-fd}.
\end{proof}

The following result says that one can efficiently certify in polynomial time that a given separator inequality is facet-defining.
In contrast, for a cut inequality of the closely related lifted multicut problem, such a decision problem is NP-hard \cite{naumannbox}.

\begin{lemma}\label{lemma:separator-polynomial-time}
	Conditions~\ref{item:separator-fd-1}~--~\ref{item:separator-fd-4} in \Cref{thm:separator-fd} can be checked in polynomial time.
	In particular, we can efficiently decide whether a separator inequality defines a facet.
\end{lemma}

\begin{proof}
	See \Cref{section:proof:lemma:separator-polynomial-time}.
\end{proof}

As $f$-cut-vertices are a specific type of minimal $f$-separators, we can obtain the following result from the proof of \Cref{thm:separator-fd}.

\begin{corollary}\label{cor:ineq-cut-vertex}
	Let $G = (V, E)$ be a connected graph, let $F \subseteq \tbinom{V}{2}$, let $f \in F$, and let $g \in V \cup F$ be an $f$-cut-vertex or a pair of $f$-cut-vertices of $G$.
	Then the inequality $x_{g} \leq x_f$ is valid.

	Moreover, it is facet-defining for $\polytope$ if and only if the following two conditions are satisfied:
	\begin{enumerate}[label=\textnormal{(\roman*)}, leftmargin=2.5em]
		\item \label{item:cor-ineq-cut-vertex-i}
		There is no interaction $f^\prime \in F \setminus \{f, g\}$ such that $f'$ is a pair of $f$-cut-vertices of $G$, and $g$ is an $f^\prime$-cut-vertex or a pair of $f^\prime$-cut-vertices of $G$.
		\item \label{item:cor-ineq-cut-vertex-ii}
		There do not exist two adjacent interactions $f^\prime, f^{\prime\prime} \in F$ such that all of the following conditions are satisfied:
		\begin{enumerate}[label=\textnormal{(\alph*)}, leftmargin=2.2em]
			\item The vertex in $f^{\prime} \cap f^{\prime\prime}$ is not an $f$-cut-vertex of $G$;
			\item $f^{\prime} \triangle f^{\prime\prime}$ is a pair of $f$-cut-vertices of $G$;
			\item $g$ is a cut-vertex or a pair of cut-vertices of both $f^{\prime}$ and $f^{\prime\prime}$ of $G$.
		\end{enumerate}
	\end{enumerate}
\end{corollary}

See \Cref{fig:separator-cor} for examples illustrating the conditions of \Cref{cor:ineq-cut-vertex}.

\begin{figure}
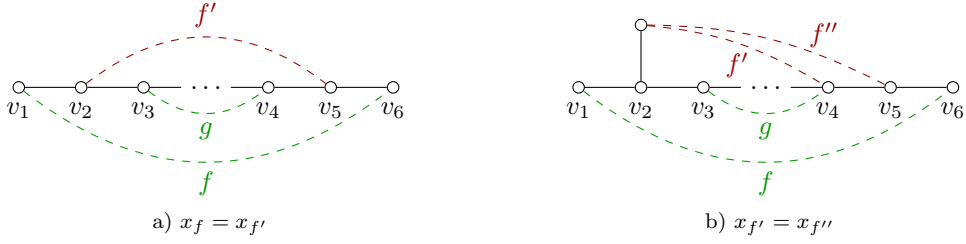

	\centering
	\begin{subfigure}[b]{0.45\textwidth}
		\centering
		\input{figures/separator/separator-cor-1.tex}
		\caption{$x_f = x_{f^{\prime}}$}
		\label{fig:separator-cor-1}
	\end{subfigure}
	\hspace{3ex}
	\begin{subfigure}[b]{0.45\textwidth}
		\centering
		\input{figures/separator/separator-cor-2.tex}
		\caption{$x_{f^{\prime}} = x_{f^{\prime\prime}}$}
		\label{fig:separator-cor-2}
	\end{subfigure}
	\caption{Examples illustrating the conditions of \Cref{cor:ineq-cut-vertex}.
	Each subfigure depicts a graph $G$ (solid black) together with the interactions (dashed) defining a valid inequality $x_g \leq x_f$. 
	Neither inequality is facet-defining: In \subref{fig:separator-cor-1}, condition~\ref{item:cor-ineq-cut-vertex-i} is violated by $f'$ (highlighted in red), while in \subref{fig:separator-cor-2}, condition~\ref{item:cor-ineq-cut-vertex-ii} is violated by $f'$ and $f''$ (highlighted in red).  
	For each subfigure, we additionally list an equality satisfied by all feasible vectors in the face induced by the corresponding inequality.
	}
	\label{fig:separator-cor}
\end{figure}

For a better understanding of the conditions of \Cref{thm:separator-fd}, and complementary to the proof, we make some remarks and give several examples (see \Cref{fig:separator-not-facet}) illustrating how they can be violated.

\begin{remark}
	In \Cref{thm:separator-fd}, Conditions~\ref{item:separator-fd-3} and \ref{item:separator-fd-4} are not independent:
	If there is an edge $\{u, v\}$ of $\HH(f, S)$ whose two endvertices are degenerate and $v$ is separated from $f$ by $S$ in $G$, then both conditions are violated.
	See \Cref{fig:separator-not-facet}~\subref{fig:sep-a} for an example.
\end{remark}

\begin{remark}
	In the context of \Cref{thm:separator-fd}, assume Condition~\ref{item:separator-fd-1}.
	If Condition~\ref{item:separator-fd-2} does not hold for some $f^\prime \in F(S) \setminus \{f\}$, then $S(f^\prime)$ can only be
	\begin{enumerate}[label=\textnormal{(\Alph*)}]
		\item \label{item:separator-condition-ii-1}
		the empty set, in which case $S$ contains two disjoint $f^\prime$-separators of $G$;
		\item \label{item:separator-condition-ii-2}
		a singleton $\{s\}$ with $s \in f^\prime \cap S$, in which case the endvertex of $f^\prime$ that is not in $S$ must be an $f$-cut-vertex of $G(S, s)$;
		\item \label{item:separator-condition-ii-3}
		the whole separator $S$, in which case both endvertices of $f^\prime$ are $f$-cut-vertices of $G$.
	\end{enumerate}
	E.g., in \Cref{fig:separator-def}, the interactions $s_1s_2$,~$us_2$ and~$uw$ violate Condition~\ref{item:separator-fd-2}, since they satisfy the cases \ref{item:separator-condition-ii-1},~\ref{item:separator-condition-ii-2} and~\ref{item:separator-condition-ii-3} above, respectively.
	To see why only these three cases are possible, suppose $S(f^\prime)$ is neither a singleton consisting of an endvertex of $f$ nor the empty set.
	In other words, $f^\prime$ is disjoint from $S$ and it is a pair of $f$-cut-vertices of $G(S, s)$ for any vertex $s$ in the non-empty set $S(f^\prime)$.
	Because of \Cref{remark:separator-inequality-minimal}, this implies that $f^\prime$ is connected in $G(S, s)$ for every $s \in S$. 
	Hence, $S(f^\prime) = S$.
\end{remark}

We also give an alternative characterization of facet-defining separator inequalities for an interesting special case.

\begin{lemma}\label{lemma:separator-three-disjoint-paths}
	Under conditions~\ref{item:separator-fd-1},~\ref{item:separator-fd-2} and~\ref{item:separator-fd-3} of \Cref{thm:separator-fd}, if there exist at least three internally vertex-disjoint $f$-paths in $G$ such that each of them intersects the separator $S$ in one vertex, then condition~\ref{item:separator-fd-4} is satisfied if and only if there is no interaction $\{p, v\} \in F(S)$ with $p \in f$ such that the following two properties hold: 1.~$v$ is a vertex not separated from $f$ by $S$ in $G$ and 2.~$\{v, s\} \in F$ for every $s \in S$.
\end{lemma}

\begin{proof}
	See \Cref{section:proof:lemma:separator-three-disjoint-paths}.
\end{proof}

\begin{remark}
	As an application of the facet-defining conditions we obtained, we can construct valid inequalities stronger than the ordinary separator inequalities~\eqref{eq:separator}.
	%
	%
	Here, as a simple illustration, we observe that \Cref{lemma:separator-three-disjoint-paths} suggests to add the term $\sum_{s \in S} (1-x_{vs}) - (1-x_{vp})$ to the left hand side of the separator inequality $(1-x_f) \leq \sum_{s \in S} (1-x_s)$, where $v$ is an arbitrary vertex separated from $q \in f = \{p, q\}$ by $S$ in $G$ and $\{v, p\}, \{v, s\} \in F$ for all $s \in S$.
	This procedure results in the following inequality
	\begin{equation}\label{eq:seprator-strengthened}
		(1 - x_f)
		\leq
		\sum_{s \in S} (x_{vs} - x_s) + (1 - x_{vp})
		\,,
	\end{equation}
	which can indeed be verified to be valid for $\polytope$.
	Of course, since $\sum_{s \in S} (1-x_{vs}) - (1-x_{vp}) \geq 0$, \eqref{eq:seprator-strengthened} is stronger than the separator inequality from which it is derived.
\end{remark}

\begin{figure}
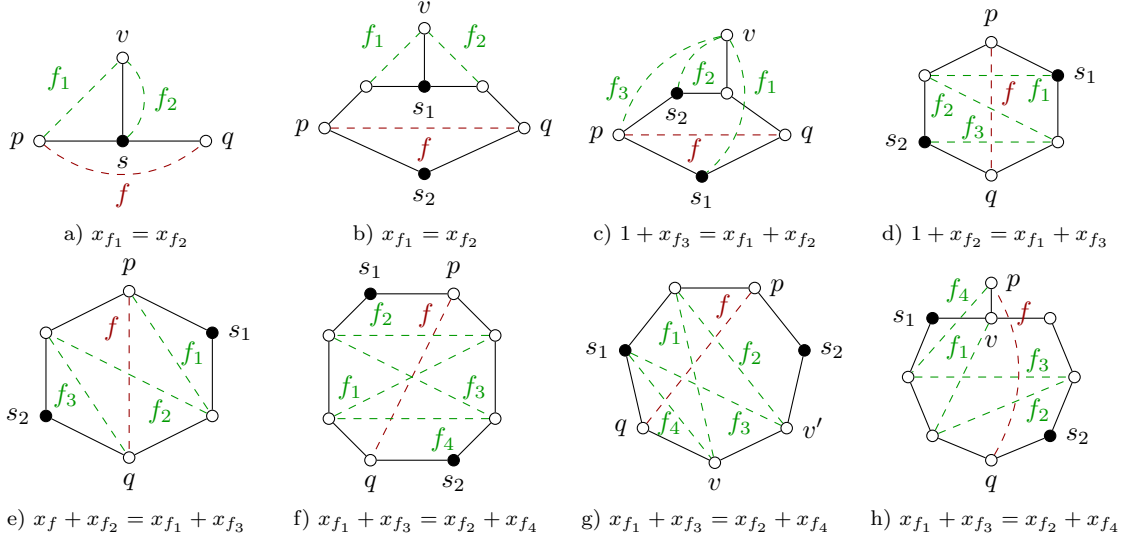

	\centering	
	\begin{subfigure}[b]{0.23\textwidth}
		\centering
		\input{figures/separator/a.tex}
		\caption{$x_{f_1} = x_{f_2}$}
		\label{fig:sep-a}
	\end{subfigure}
	\hfill
	\begin{subfigure}[b]{0.23\textwidth}
		\centering
		\input{figures/separator/b.tex}
		\caption{$x_{f_1} = x_{f_2}$}
		\label{fig:sep-b}
	\end{subfigure}
	\hfill
	\begin{subfigure}[b]{0.23\textwidth}
		\centering
		\input{figures/separator/c.tex}
		\caption{$1 + x_{f_3} = x_{f_1} + x_{f_2}$}
		\label{fig:sep-c}
	\end{subfigure}
	\hfill
	\begin{subfigure}[b]{0.23\textwidth}
		\centering
		\input{figures/separator/d.tex}
		\caption{$1 + x_{f_2} = x_{f_1} + x_{f_3}$}
		\label{fig:sep-d}
	\end{subfigure}

	\begin{subfigure}[b]{0.23\textwidth}
		\centering
		\input{figures/separator/e.tex}
		\caption{$x_{f} + x_{f_2} = x_{f_1} + x_{f_3}$}
		\label{fig:sep-e}
	\end{subfigure}
	\hfill
	\begin{subfigure}[b]{0.23\textwidth}
		\centering
		\input{figures/separator/f.tex}
		\caption{$x_{f_1} + x_{f_3} = x_{f_2} + x_{f_4}$}
		\label{fig:sep-f}
	\end{subfigure}
	\hfill
	\begin{subfigure}[b]{0.23\textwidth}
		\centering
		\input{figures/separator/g.tex}
		\caption{$x_{f_1} + x_{f_3} = x_{f_2} + x_{f_4}$}
		\label{fig:sep-g}
	\end{subfigure}
	\hfill
	\begin{subfigure}[b]{0.23\textwidth}
		\centering
		\input{figures/separator/h.tex}
		\caption{$x_{f_1} + x_{f_3} = x_{f_2} + x_{f_4}$}
		\label{fig:sep-h}
	\end{subfigure}
	\caption{Each subfigure depicts a graph $G$ (solid black), interactions (dashed edges), a designated interaction $f = pq$ (dashed, red), and a minimal $f$-separator $S$ (filled black vertices).
	In each case, the separator inequality associated with $f$ and $S$ does not define a facet because either Condition~\ref{item:separator-fd-3} or~\ref{item:separator-fd-4} is violated, though Conditions~\ref{item:separator-fd-1} and~\ref{item:separator-fd-2} are satisfied.
	This fact can also be confirmed by the additional equalities listed below the subfigures.
	%
	\subref{fig:sep-a}
		All the vertices are degenerate and $f_1$ is the only edge of $\HH(f, \{s\})$.
		Condition~\ref{item:separator-fd-3} is violated by the degenerate vertex $v$, though it is a leaf of $\HH(f, \{s\})$.
		As both endvertices of $f_1$ are degenerate, Condition~\ref{item:separator-fd-4} also does not hold.
	\subref{fig:sep-b}
		The only degenerate vertices are $s_1$,~$s_2$,~$p$,~$q$ and the edge set of $\HH(f, S)$ consists of $f_1$ and $f_2$.
		The unique non-isolated vertex $v$ of $\HH(f, S)$ separated from $f$ by $S = \{s_1, s_2\}$ in $G$ is $v$, it is non-degenerate but not a leaf of $\HH(f, S)$, so Condition~\ref{item:separator-fd-3} of \Cref{thm:separator-fd} is violated by $v$.
		Thus the corresponding separator inequality is not facet-defining.
		However, $\HH(f, S)$ satisfies Condition~\ref{item:separator-fd-4}.
	\subref{fig:sep-c}
		The only degenerate vertices are $v$,~$s_1$,~$s_2$,~$p$,~$q$ and the only edge of $\HH(f, \{s_1, s_2\})$ is $f_3$.
		It is direct to see that Condition~\ref{item:separator-fd-3} is satisfied as no vertices are separated from $f$ by $S = \{s_1, s_2\}$ in $G$.
		The odd path $\HH(f, \{s_1, s_2\})$ has two degenerate endvertices and hence Condition~\ref{item:separator-fd-4} is false.
	\subref{fig:sep-d}
		All the vertices are degenerate and $f_2$ is the only edge of $\HH(f, \{s_1, s_2\})$.
		As no vertices of $\HH(f, \{s_1, s_2\})$ are separated from $f$ by $S$ in $G$, Condition~\ref{item:separator-fd-3} is vacuously true.
		However, Condition~\ref{item:separator-fd-4} is not satisfied as the two leaves, i.e. the endvertices of $f_2$, of the odd path $\HH(f, \{s_1, s_2\})$ are degenerate.
	\subref{fig:sep-e}
		In this example, $s_1$,~$s_2$,~$p$,~$q$ are all the degenerate vertices and the edge set of $\HH(f, \{s_1, s_2\})$ is $\{f_1, f_2, f_3\}$.
		Condition~\ref{item:separator-fd-3} is vacuously true due to the same reason as in the previous case.
		$\HH(f, \{s_1, s_2\})$ is an odd path whose two endvertices are degenerate and hence violates Condition~\ref{item:separator-fd-4}.
	\subref{fig:sep-f}
		The only degenerate vertices are $s_1$,~$s_2$,~$p$,~$q$ and the edge set of $\HH(f, \{s_1, s_2\})$ consists of $f_1$,~$f_2$,~$f_3$ and~$f_4$.
		As the previous case, Condition~\ref{item:separator-fd-3} is vacuously true.
		On the other hand, $\HH(f, \{s_1, s_2\})$ is an even cycle and hence Condition~\ref{item:separator-fd-4} is violated.
	\subref{fig:sep-g}
		The degenerate vertices in this case include $v$ and~$v^\prime$ besides $f$-cut-vertices $p$,~$q$ and the vertices $s_1$,~$s_2$ in the separator $S$.
		The edge set of $\HH(f, S)$ consists of $f_1$ and $f_2$.
		As no vertices are separated from $f$ by $S$ in $G$, Condition~\ref{item:separator-fd-3} is trivially satisfied.
		Condition~\ref{item:separator-fd-4} is violated because $\HH(f, S)$ is an even path with two degenerate endvertices $v$,~$v^\prime$.
	\subref{fig:sep-h}
		The only degenerate vertices are $f$-cut-vertices $p$,~$q$ and~$v$, besides those in $S$.
		The edges of $\HH(f, S)$ are given by $f_1$,~$f_2$,~$f_3$ and~$f_4$.
		Again, as no vertices are separated from $f$ by $S$ in $G$, Condition~\ref{item:separator-fd-3} is trivially satisfied.
		$\HH(f, S)$ is an even path with two endvertices $p$ and~$v$ being $f$-cut-vertices of $G$, so Condition~\ref{item:separator-fd-4} is violated.
	}
	\label{fig:separator-not-facet}
\end{figure}


\subsection{Facets induced by path inequalities}
\label{section:facets-path}

In this section, we introduce a class of valid inequalities, called \emph{path inequalities}, which strengthen the connector inequalities.
We then characterize when these inequalities define facets of the multi-separator polytope.
Two illustrative examples are shown in \Cref{fig:path-ineq-example}.

\begin{definition}
	Let $G = (V, E)$ be a connected graph, let $F \subseteq \tbinom{V}{2}$, let $f \in F$, and let $P = (V_P, E_P)$ be an $f$-path in the graph $(V, F)$.
	We call the inequality written below the \emph{path inequality} associated with $f$ and $P$.
	\begin{equation}\label{eq:path}
		x_f \;\leq\; \sum_{e \in E_P} x_e \;-\; \sum_{v \in {V_P} \setminus f} x_v
	\end{equation}
\end{definition}

The inequality~\eqref{eq:path} can be interpreted as follows: Separating the pair $f$ requires separating at least one edge along the path $P$, while accounting for overlaps at internal vertices.
In this sense, path inequalities refine connector inequalities by exploiting the structure of a specific path between the endpoints of $f$.

\begin{lemma}\label{lemma:path-valid}
	Let $G = (V, E)$ be a connected graph and let $F \subseteq \tbinom{V}{2}$.
	For any $f \in F$ and any $f$-path $P$ in $(V, F)$, the corresponding path inequality~\eqref{eq:path} is valid for $\polytope$.
\end{lemma}

\begin{proof}
	See \Cref{section:proof:lemma:path-valid}.
\end{proof}

The class of path inequalities is large and structurally rich. 
In the remainder of this section, we restrict attention to path inequalities where $P$ is a path in $G$.
For this subclass, we obtain a complete characterization of when such inequalities are facet-defining.

\begin{theorem}\label{thm:path-facet}
	Let $G = (V, E)$ be a connected graph, let $F \subseteq \tbinom{V}{2}$, let $f \in F$, and let $P = (V_P, E_P)$ be an $f$-path in $G$.
	Suppose that $E_P \subseteq F \setminus \{f\}$.
	Then the path inequality~\eqref{eq:path} defines a facet of $\polytope$ if and only if $P$ is chordless in the graph $(V, F \setminus \{f\})$.
\end{theorem}
\begin{proof}
	See \Cref{section:proof:thm:path-facet}.
\end{proof}

As a consequence of \Cref{thm:path-facet}, we obtain the following characterization regarding a class of inequalities defined on cycles in $G$:

\begin{corollary}
	Let $G = (V, E)$ be a connected graph, let $F \subseteq \tbinom{V}{2}$, let $C = (V_C, E_C)$ be a cycle in $G$ with $E_C \subseteq F$, and let $e$ be an edge of $C$.
	Then the inequality written below is facet-defining for $\polytope$ if and only if $C$ is chordless in $(V, F)$.
	\[
	\sum_{v \in V_C \setminus e} x_v + x_e - \sum_{e^\prime \in E_C \setminus \{e\}} x_{e^\prime} \leq 0
	\]
\end{corollary}

\begin{figure}
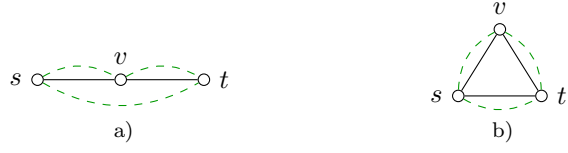

	\centering
    \begin{subfigure}[b]{0.3\textwidth}
        \centering
        \input{figures/path-ineq-example-1.tex}
        \caption{}
        \label{fig:path-ineq-1}
    \end{subfigure}
    \hspace{2ex}
    \begin{subfigure}[b]{0.3\textwidth}
        \centering
        \input{figures/path-ineq-example-2.tex}
        \caption{}
        \label{fig:path-ineq-2}
    \end{subfigure}
	\caption{
		Each subfigure shows a graph \(G\) (solid black) with interactions (dashed green) and a path \(P\) on vertices \(\{s, v, t\}\).
		In \subref{fig:path-ineq-1}, the $st$-connector inequality $x_{st} \leq x_s + x_v + x_t$ is not facet-defining by \Cref{thm:connector-fd}~\ref{item:connector-fd-2}, whereas the path inequality $x_{st} \leq x_{sv} + x_{vt} - x_v$ is facet-defining by \Cref{thm:path-facet}.
		In \subref{fig:path-ineq-2}, the connector inequality $x_{st} \leq x_s + x_v + x_t$ is not facet-defining by \Cref{thm:connector-fd}~\ref{item:connector-fd-1}, whereas the path inequality $x_{st} \leq x_{sv} - x_v + x_{vt}$ is facet-defining by \Cref{thm:path-facet}.}
	\label{fig:path-ineq-example}
\end{figure}


\subsection{Facets induced by intersection inequalities}
\label{section:facets-intersection}

In this section, we introduce a class of inequalities, called \emph{intersection inequalities}, arising from interactions between vertex pairs, and characterize precisely when such inequalities define facets of the multi-separator polytope.

\begin{definition}\label{def:ineq-interaction}
	Let $G = (V, E)$ be a connected graph and let $F \subseteq \tbinom{V}{2}$.
    Suppose that for some pairwise distinct vertices $v_1, v_2, w_1, w_2 \in V$, the four pairs
    \[
        \{v_1, w_1\},\ \{v_2, w_2\},\ \{v_1, w_2\},\ \{v_2, w_1\}
    \]
    are all contained in $F$.  
    If, in addition, $\{v_1, v_2\}, \{w_1, w_2\} \in E$, then the inequality
    \begin{equation}\label{eq:ineq-intersection}
        x_{v_1w_2} + x_{v_2w_1} \;\leq\; x_{v_1w_1} + x_{v_2w_2}
    \end{equation}
    is called an \emph{intersection inequality}.
\end{definition}

The inequality~\eqref{eq:ineq-intersection} compares the two \emph{parallel pairs} $\{v_1,w_2\}$ and $\{v_2,w_1\}$ with the two \emph{crossing pairs} $\{v_1,w_1\}$ and $\{v_2,w_2\}$.
Intuitively, it asserts that separating the parallel pairs is at least as restrictive as separating the crossing pairs, provided the endpoints are linked by the edges $\{v_1,v_2\}$ and $\{w_1,w_2\}$.
For an illustration, see \Cref{fig:intersection-ineq-example}.

\begin{figure}
	\centering
	\input{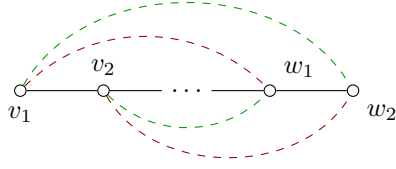}
	\caption{An illustration of the intersection inequality~\eqref{eq:ineq-intersection} with $G$ depicted in solid black: the crossing pairs $(v_1,w_1)$ and $(v_2,w_2)$ (in red) versus the parallel pairs $(v_1,w_2)$ and $(v_2,w_1)$ (in green).}
	\label{fig:intersection-ineq-example}
\end{figure}

\begin{lemma}\label{lemma:intersection-valid}
    Under the assumptions of \Cref{def:ineq-interaction}, the intersection inequality~\eqref{eq:ineq-intersection} is valid for the multi-separator polytope $\polytope$ if and only if $\{v_1, w_1\}$ is a $\{v_2, w_2\}$-separator of $G$ and $\{v_2, w_2\}$ is a $\{v_1, w_1\}$-separator of $G$.
\end{lemma}
\begin{proof}
    See \Cref{section:proof:lemma:intersection-valid}.
\end{proof}

The next result shows that no additional conditions are required for these inequalities to be facet-defining: validity already captures the full strength of the construction.

\begin{theorem}\label{theorem:intersection-facet}
    Under the assumptions of \Cref{def:ineq-interaction}, the intersection inequality~\eqref{eq:ineq-intersection} is facet-defining for the multi-separator polytope $\polytope$ if and only if $\{v_1, w_1\}$ is a $\{v_2, w_2\}$-separator of $G$ and $\{v_2, w_2\}$ is a $\{v_1, w_1\}$-separator of $G$.
\end{theorem}
\begin{proof}
    See \Cref{section:proof:theorem:intersection-facet}.
\end{proof}

As we will see in \Cref{section:relatives-lmp}, intersection inequalities, together with those discussed before, suffice to completely describe the multi-separator polytope for path graphs with complete pairwise vertex interactions.


\section{Related polytopes}\label{section:related-polytopes}

In this section, we discuss relations between the multi-separator polytope and two closely related polytopes, namely the boolean quadric polytope and the lifted multicut polytope.

\subsection{Boolean quadric polytope}
\label{section:relatives-bqp}

The boolean quadric polytope $\bqp_G$ associated with a graph $G = (V, E)$ as defined by \citet{padberg1989boolean} is isomorphic to the multi-separator polytope associated with the graph $G$ and interactions $F = E$, more specifically, it is the image of $\Xi_{GE}$ under the automorphism of $\mathbb{R}^{V \cup E}$ given by $x \mapsto \mathbbm{1}_{V \cup E} - x$.
In this section, we consider any interaction set $F$ on $G$ and show how to transfer valid inequalities from the boolean quadric polytope associated with the graph $(V, F)$ to the multi-separator polytope associated with $G$ and $F$. Afterward, we give some examples illustrating that the conditions under which an inequality that is valid for $\bqp_G$ becomes facet-defining for $\bqp_G$ are generally quite different from the conditions under which the corresponding derived inequality that is valid for $\polytope$ becomes facet-defining for $\polytope$.

\begin{lemma}\label{lemma:bqp-msp}
    Let $G = (V, E)$ be a connected graph, let $F$ be an interaction set on $G$, and let $ax \leq \alpha$ be an inequality in $\mathbb{R}^{V \cup F}$ with $a \in \mathbb{R}^{V \cup F}$ and $\alpha \in \mathbb{R}$.
    \begin{enumerate}[label=\textnormal{(\roman*)}]
        \item \label{item:bqp-msp-1}
        If $ax \leq \alpha$ is a valid inequality for the boolean quadric polytope $\bqp_{(V, F)}$ such that for every $f \in F$ we have $f \in E$ or $a_f \geq 0$, then the flipped inequality $a(1-x) \leq \alpha$ is valid for the multi-separator polytope $\polytope$.
        \item \label{item:bqp-msp-2}
        If $a(1-x) \leq \alpha$ is valid for $\polytope$ such that for every $f \in F$ we have $f \in E$ or $a_f \leq 0$, then $ax \leq \alpha$ is valid for $\bqp_{(V, F)}$.
        \item \label{item:bqp-msp-3}
        Under the assumption that $\{f \in F \mid a_f \neq 0\} \subseteq E \subseteq F$, if $a(1-x) \leq \alpha$ is facet-defining for $\polytope$, then so is $a|_{V \cup E}x \leq \alpha$ for $\bqp_{G}$.
    \end{enumerate}
\end{lemma}
\begin{proof}
    See \Cref{section:proof:lemma:bqp-msp}.
\end{proof}

The last statement~\ref{item:bqp-msp-3} of \Cref{lemma:bqp-msp} tells us that, under the given assumption, if a condition is necessary for $ax \leq \alpha$ to be facet-defining for $\bqp_G$, then so it is for $a(1-x) \leq \alpha$ to be facet-defining for $\polytope$.
We remark that the validity of an inequality $ax \leq \alpha$ for the boolean quadric polytope $\bqp_{(V, F)}$ does not necessarily imply that $a(1-x) \leq \alpha$ is valid for the multi-separator polytope $\polytope$; even if this is the case, the property of being facet-defining usually does not transfer from $\bqp_{(V, F)}$ to $\polytope$.
Illustrative examples can be given by the odd-cycle inequalities described by \citet{padberg1989boolean}:

\begin{example}
    Let $G = (V, E)$ be a connected graph, let $F$ be an interaction set $F$ on $G$, let $C = (V_C, E_C)$ be a cycle in $G$ with $E_C \subseteq F$ and let $M \subseteq E_C$ be an edge subset of odd cardinality.
	We denote by
	\begin{align*}
	    S_0 & =
        \left\{
            v \in V_C \mid \exists \, e \neq e^{\prime} \in M \colon \quad 
            	e \cap e^{\prime} = \{v\}
        \right\}
	\\
	    S_2 & =
        \left\{
            v \in V_C \mid \exists \, e \neq e^{\prime} \in E_C \setminus M \colon \quad 
            	e \cap e^{\prime} = \{v\}
        \right\}
	\end{align*}
	the set of all vertices on the cycle $C$ with both incident edges in $M$ and the set of all vertices on the cycle $C$ with no incident edges in $M$, respectively, and
	\begin{equation*}
	    S_1 = V_C \setminus (S_0 \cup S_2)\,.
	\end{equation*}
	Then the \emph{odd-cycle inequality}
    \begin{equation}\label{eq:odd-cycle-bqp}
	    \sum_{v \in S_0} x_{v} - \sum_{v^\prime \in S_2} x_{v^\prime} + \sum_{e^\prime \in E_C \setminus M} x_{e^\prime} - \sum_{e \in M} x_{e} \leq \left \lfloor{\lvert M \rvert /2}\right \rfloor
	\end{equation}
	is valid for $\bqp_{G}$ \cite{padberg1989boolean}.
    Since $\lvert E_C \rvert = 2 \lvert M \rvert - \lvert S_0 \rvert + \lvert S_2 \rvert$, the inequality obtained from \eqref{eq:odd-cycle-bqp} by flipping variables is
    \begin{equation}\label{eq:odd-cycle-msp}
	    \sum_{v^\prime \in S_2} x_{v^\prime} - \sum_{v \in S_0} x_{v} + \sum_{e \in M} x_{e} - \sum_{e^\prime \in E_C \setminus M} x_{e^\prime} \leq \left \lfloor{\lvert M \rvert /2}\right \rfloor\,,
	\end{equation}
    which, by \Cref{lemma:bqp-msp}, is still valid for $\polytope$.
    However, if the cycle $C$ is in the graph $(V, F)$ instead of $G$, the associated inequality~\eqref{eq:odd-cycle-msp} may not be a valid inequality for $\polytope$ anymore, unless $M$ is an edge subset of $G$, by \Cref{lemma:bqp-msp}.
    Such an example is given in \Cref{fig:odd-cycle-example} \subref{fig:odd-cycle-a}.

    By Theorem~$9$ of \citet{padberg1989boolean}, an odd-cycle inequality~\eqref{eq:odd-cycle-bqp} is facet-defining for $\bqp_{G}$ if and only if $C$ is a chordless cycle of $G$, which is also a necessary condition for \eqref{eq:odd-cycle-msp} to be facet-defining for $\polytope$ with $E \subseteq F$, by \Cref{lemma:bqp-msp}.
    However, it can be checked that the flipped odd-cycle inequality~\eqref{eq:odd-cycle-msp} defines a facet of $\Xi_{CF_C}$ if and only if $F_C = E_C$, i.e., the cycle $C$ is chordless in the graph $(V_C, F_C)$, where $F_C = F \cap \tbinom{V_C}{2}$.
    Therefore, if $E_C$ is a strict subset of $F_C$, then none of the facet-defining odd-cycle inequalities for the boolean quadric polytope $\bqp_C$ defines a facet of the multi-separator polytope $\Xi_{CF_{C}}$.
    See \Cref{fig:odd-cycle-example}~\subref{fig:odd-cycle-b},~\subref{fig:odd-cycle-c} and~\subref{fig:odd-cycle-d} for an illustration.
    These examples tell us that the converse of the last statement of \Cref{lemma:bqp-msp} is not true.
    \Cref{fig:odd-cycle-example}~\subref{fig:odd-cycle-e} gives a cycle in $G$ that is chordless in $(V, F)$, whose associated odd-cycle inequality is valid for both $\bqp_{(V, F)}$ and $\polytope$ but is facet-defining only for $\bqp_{(V, F)}$.
    In \Cref{fig:odd-cycle-example}~\subref{fig:odd-cycle-f}, the cycle is not chordless in $(V, F)$ and hence the associated odd-cycle inequality is not facet-defining for $\bqp_{(V, F)}$; however, it is facet-defining for $\polytope$.
    
    While odd-cycle inequalities~\eqref{eq:odd-cycle-bqp} that are facet-defining for the boolean quadric polytope admit a simple characterization, the examples given in \Cref{fig:odd-cycle-example} suggest that the conditions under which the odd-cycle inequalities~\eqref{eq:odd-cycle-msp} define facets of the multi-separator polytope are more complicated.
\end{example}

\begin{figure}[t]
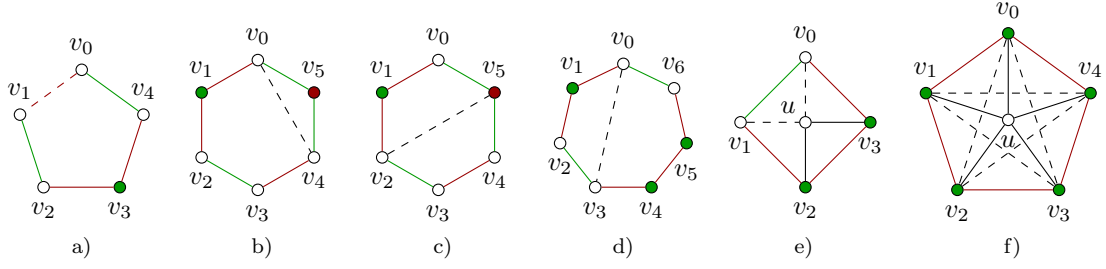

    \centering
    \begin{subfigure}[b]{0.15\textwidth}
		\centering
		\input{figures/relatives/odd-cycle-a.tex}
		\caption{}
		\label{fig:odd-cycle-a}
	\end{subfigure}
	\hfill
    \begin{subfigure}[b]{0.15\textwidth}
		\centering
		\input{figures/relatives/odd-cycle-b.tex}
		\caption{}
		\label{fig:odd-cycle-b}
	\end{subfigure}
	\hfill
    \begin{subfigure}[b]{0.15\textwidth}
		\centering
		\input{figures/relatives/odd-cycle-c.tex}
		\caption{}
		\label{fig:odd-cycle-c}
	\end{subfigure}	
	\hfill
    \begin{subfigure}[b]{0.15\textwidth}
		\centering
		\input{figures/relatives/odd-cycle-d.tex}
		\caption{}
		\label{fig:odd-cycle-d}
	\end{subfigure}
	\hfill
    \begin{subfigure}[b]{0.15\textwidth}
		\centering
		\input{figures/relatives/odd-cycle-e.tex}
		\caption{}
		\label{fig:odd-cycle-e}
	\end{subfigure}	
	\hfill
    \begin{subfigure}[b]{0.2\textwidth}
		\centering
		\input{figures/relatives/odd-cycle-f.tex}
		\caption{}
		\label{fig:odd-cycle-f}
	\end{subfigure}
    \caption{Depicted above are graphs $G$ (solid) and interactions $F$ (both solid and dashed).
    Vertices and interactions with coefficients $+1$ and $-1$ in the odd-cycle inequalities~\eqref{eq:odd-cycle-msp} are depicted in red and green, respectively.
    The odd-cycle inequality associated with the cycle in \subref{fig:odd-cycle-a} is not valid for the corresponding multi-separator polytope $\polytope$, which can be seen by looking at the feasible vector $x^{\{v_0, v_1, v_3\}}$.
    Each of the odd-cycle inequalities associated with the colored cycles in \subref{fig:odd-cycle-b},~\subref{fig:odd-cycle-c},~\subref{fig:odd-cycle-d} and~\subref{fig:odd-cycle-e} is valid but not facet-defining for $\polytope$ due to the existence of extra interactions.
    More specifically, in \subref{fig:odd-cycle-b} we have the additional equality $x_{v_5} + x_{v_0v_4} = x_{v_0v_5} + x_{v_4v_5}$ containing the face defined by the corresponding odd-cycle inequality.
    Similarly, the additional inequalities $x_{v_2v_5} = x_{v_3v_4}$, $x_{v_0v_3} = 1$, and $x_{uv_1} = x_{uv_2}$ are satisfied in \subref{fig:odd-cycle-c},~\subref{fig:odd-cycle-d}, and~\subref{fig:odd-cycle-e}, respectively.
    The odd-cycle inequality associated with the colored cycle in \subref{fig:odd-cycle-f} is valid and facet-defining for $\polytope$.
    }
    \label{fig:odd-cycle-example}
\end{figure}

\subsection{Lifted multicut polytope}
\label{section:relatives-lmp}

In this section, we show that every lifted multicut polytope (see \Cref{def:lifted-multicut-problem} below) is the projection of a face of some multi-separator polytope, and vice versa.
Our proofs are inspired by the linear-time reductions between the multi-separator and lifted multicut problem by \citet{irmai-2023-separator}.

\begin{definition}\cite{hornakova-2017}
	\label{def:lifted-multicut-problem}
    For any connected graph $G = (V, E)$ and any graph $G' = (V, F)$ with $E \subseteq F$, let $\lmcfeasible_{GG'}$ denote the set of all $y \in \{0, 1\}^{F}$ that satisfy the following inequalities:
    \begin{align}
        \forall (V_C, E_C) \in \textnormal{cycles}(G) \ 
        \forall e \in C \colon \qquad
        \hspace{4ex} y_{e} & \leq \sum_{e' \in E_C \setminus \{e\}} y_{e'}
        \,,
        \label{eq:lmc-cycle}
        \\
        \forall f \in F \setminus E \ 
        \forall (V_P, E_P) \in f\textnormal{-paths}(G) \colon \qquad
        \hspace{4ex} y_{f} & \leq \sum_{e \in E_P} y_e
        \,,
        \label{eq:lmc-path}
        \\
        \forall f \in F \setminus E \
        \forall C \in f\textnormal{-cuts}(G) \colon \qquad
        1 - y_f & \leq \sum_{e \in C} (1 - y_e)
        \,.
        \label{eq:lmc-cut}
    \end{align}
    The convex hull $\lmc_{GG'} \coloneqq \conv \lmcfeasible_{GG'}$ is called the \emph{lifted multicut polytope} with respect to $G$ and $G'$.
\end{definition}

\begin{remark}\label{remark:lmc-ilp}
    As for the multi-separator problem, replacing the cycles in \eqref{eq:lmc-cycle} by chordless cycles, the paths in \eqref{eq:lmc-path} by chordless paths, and the cuts in \eqref{eq:lmc-cut} by minimal cuts, results in equivalent definitions of the same $\lmcfeasible_{GG'}$ and $\lmc_{GG'}$.
\end{remark}

Recall that the subdivision $\widehat{G}$ of a graph $G = (V, E)$ is the graph obtained from subdividing all edges of $G$, i.e., every edge $e = \{u, w\} \in E$ is replaced by two edges $\{u, v_e\}$ and $\{v_e, w\}$ via inserting a new vertex $v_e$.
Therefore, $\widehat{G} = (\widehat{V}, \widehat{E})$, where $\widehat{V} = V \cup \{v_e \mid e \in E\}$ and $\widehat{E} = \{\{v, v_e\} \mid v \in e \in E\}$.

\begin{proposition}\label{prop:lmc-to-ms}
    For any connected graph $G = (V, E)$, any graph $G^{\prime} = (V, F)$ with $E \subseteq F$, the lifted multicut polytope $\lmc_{GG^{\prime}}$ is the projection of a face of the multi-separator polytope $\Xi_{\widehat{G}F}$ onto $\mathbb{R}^{F}$, where $\widehat{G}$ denotes the graph resulting from the subdivision of all edges of $G$.
\end{proposition}
\begin{proof}
    See \Cref{section:proof:proposition:lmc-to-ms}.
\end{proof}

\begin{proposition}\label{prop:ms-to-lmc}
    For any connected graph $G = (V, E)$ with $F \subseteq \tbinom{V}{2}$, the multi-separator polytope $\Xi_{GF}$ is the projection of a face of the lifted multicut polytope $\lmc_{\bar{G}\bar{G}^{\prime}}$, where $\bar{G}$ is the graph obtained from the subdivision $\widehat{G}$ of $G$ by adding a new vertex $\bar{v}$ and a new edge $\{v, \bar{v}\}$ for each $v \in V$, and $\bar{G}^{\prime}$ is obtained from $\bar{G}$ by adding new edges $\{\bar{v}, v_e\}$ for all $v \in e \in E$ and all edges in $F$.
\end{proposition}
\begin{proof}
    See \Cref{section:proof:proposition:ms-to-lmc}.
\end{proof}

Although these two polytopes are closely related, we do not see a direct way to transform facets from one to the other in general.
A special case where such a relation can be identified is when the underlying graph is a path.
In the following, we show that the multi-separator polytope for paths is equivalent to the lifted multicut polytope for paths.
Thanks to this equivalence and a result from \citet{lange2021lifted}, we are able to obtain a complete totally dual integral description of the multi-separator polytope in the case where we consider as interactions all the edges that are incident on the vertices of the underlying path.

In this part, we stick to the following notation for simplicity:
For any $n \geq 1$, let $P_n = (V_n, E_n)$ be the path of length $n$ with $V_n = \{0,\dots,n\}$, $E_n = \{\{i,i+1\}\mid i\in \{0,\dots,n-1\}\}$, and let $F \subseteq F_n = \binom{V_n}{2}$ be a superset of $E_n$.
Note that for $P_n$ the connector inequalities~\eqref{eq:connector} and separator inequalities~\eqref{eq:separator} reduce to
\begin{align}
	\forall ij \in F \colon i < j \colon \quad 
	x_{ij} & \leq \sum_{k=i}^j x_k 
	\label{eq:path-connector} 
	\\
	\forall ij \in F \colon i < j \quad
	\forall k \in V \colon i \leq k \leq j \colon \quad
	x_k & \leq x_{ij} 
	\label{eq:path-separator}
\end{align}
Let $P^\prime_{n+1} = (V_{n+1}, E^\prime_{n+1})$ be another path, where 
\[
    E^\prime_{n+1} = \{e^\prime = \{i, j\} \subseteq V_{n+1} \mid i < j, \{i, j-1\} \in V_n \cup F \}
	\,.
\]
The following proposition establishes a correspondence between the multi-separator polytope $\Xi_{P_n F}$ and the lifted multicut polytope $\lmc_{P_{n+1} P^\prime_{n+1}}$.

\begin{proposition}\label{prop:ms-lmc-path}
	For any $n \in \mathbb{N}$ and any superset $F \subseteq \tbinom{V_n}{2}$ of $E_n$, we have $\Xi_{P_n F} = \lmc_{P_{n+1} P^\prime_{n+1}}$.
\end{proposition}

\begin{proof}
    See \Cref{section:proof:proposition:ms-lmc-path}.
\end{proof}

\citet{lange2021lifted} have established a complete description of the lifted multicut polytope for paths which, by the above observation, yields a complete description of the lifted multi-separator polytope for paths.
More specifically, let
\begin{align}
	x_{0n} & \leq 1
	\label{eq:path-box} 
	\\
	\forall i \in \{0,\dots,n-1\} \colon \quad \hspace{11.2ex}  
	x_{0i} &\leq x_{0,i+1} 
	\label{eq:path-separator-left} 
	\\
	\forall i \in \{1,\dots,n\} \colon \quad \hspace{11ex} 
	x_{in} &\leq x_{i-1,n} 
	\label{eq:path-separator-right}
	\\
	\forall i \in \{0,\dots,n-1\} \colon \quad \hspace{8.5ex}
	x_{i,i+1} & \leq x_{ii} + x_{i+1,i+1} 
	\label{eq:path-triangle} 
	\\
	\forall i,j \in \{1,\dots,n-1\} \ \forall i \leq j \colon \quad
	x_{ij} + x_{i-1,j+1} & \leq x_{i-1,j} + x_{i,j+1} 
	\label{eq:path-intersection}
\end{align}
where $x_{ii}$ is the variable corresponding to vertex $i$ for $i \in \{0,\dots,n\}$.
Notice that \eqref{eq:path-box} is a box inequality, \eqref{eq:path-triangle} are connector inequalities, \eqref{eq:path-separator-left} and \eqref{eq:path-separator-right} are generalized separator inequalities (See \Cref{cor:ineq-cut-vertex}), and \eqref{eq:path-intersection} are intersection inequalities introduced in \Cref{section:facets-intersection}.
System \eqref{eq:path-box}~--~\eqref{eq:path-intersection} corresponds to system (11)~--~(15) of \citet{lange2021lifted}.
By Theorem~1 of their work, this latter system provides a complete totally dual integral description of the lifted multicut polytope for paths.
As a consequence, we obtain the following result immediately.

\begin{proposition}\label{prop:ms-lmc-path-tdi}
	For any $n \in \mathbb{N}$, we have
	\begin{equation*}
		\Xi_{P_n F_n} =
        \left\{
            x \in \mathbb{R}^{V_n \cup F_n} \mid x \textnormal{ satisfies \eqref{eq:path-box} -- \eqref{eq:path-intersection}}
        \right\}
        \,.
	\end{equation*}
	Furthermore, the system \eqref{eq:path-box} -- \eqref{eq:path-intersection} is totally dual integral.
\end{proposition}

\begin{proof}
    See \Cref{section:proof:proposition:ms-lmc-path-tdi}.
\end{proof}

\section{Conclusion}
\label{section:conclusion}

We describe in terms of efficiently-decidable, graph-theoretic conditions the connector inequalities, separator inequalities and box inequalities that define a facet of the multi-separator polytope. 
To obtain a tighter polyhedral relaxation, we introduce path and intersection inequalities.
For path graphs where every vertex pair forms an interaction, this yields a totally dual integral description of the multi-separator polytope.
We relate the multi-separator polytope to the boolean quadric polytope, showing that facets induced by odd-cycle inequalities do not transfer generally, and to the lifted multicut polytope, showing that either polytope is a projection of a face of the other.
Our finding that it can be decided efficiently whether a separator inequality defines a facet of the multi-separator polytope is in contrast to the result of \citet{naumannbox} that deciding whether a cut inequality defines a facet of the lifted multicut-polytope is NP-hard.

\bibliography{references}
\bibliographystyle{plainnat}

\addtocontents{toc}{\protect\setcounter{tocdepth}{1}}


\appendix
\section{Separation algorithms}

In this section we describe how to efficiently separate connector and separator inequalities.
In particular, we show that the connector inequalities~\eqref{eq:connector} can be separated by searching for shortest paths and that the separator inequalities~\eqref{eq:separator} can be separated by solving a max-flow problem in an auxiliary graph.

\begin{proposition}\label{prop:connector-separation}
    Let $G = (V, E)$ be a connected graph, let $F \subseteq \binom{V}{2}$ be a set of interactions and let $x \in [0,1]^{V \cup F}$.
    Then one can decide in time $\mathcal{O}(|E||V|+|V|^2 \text{log}|V|)$ whether there exists a violated connector inequality and return such an inequality.
\end{proposition}

\begin{proof}
    We define the non-negative edge weights $w_{uv} = \frac{1}{2}(x_u+x_v)$ for $\{u, v\} \in E$.
    By this definition, for every $\{s, t\}$-path $P=(V_P, E_P)$ in $G$ it holds that
    \begin{align}\label{eq:path-weight-connector-weight}
        w(P) \coloneqq \sum_{e \in E_P} w_e = \sum_{u \in V_P} x_u - \frac{1}{2}(x_s + x_t)\,.
    \end{align}
    Now let $f=\{s,t\} \in F$ and let $P=(V_P,E_P)$ be a minimum weighted $f$-path in $G$ with weights $w$.
    Then $x$ violates a connector inequality associated with $f$ if and only if $w(P) + \tfrac{1}{2}(x_s + x_t) < x_{f}$ by the following arguments. 

    Suppose it holds that $w(P) + \tfrac{1}{2}(x_s + x_t) < x_{f}$. 
    Clearly, $V_P$ is an $f$-connector of $G$ and \eqref{eq:path-weight-connector-weight} yields that $x$ violates the connector inequality associated with the interaction $f$ and the $f$-connector $V_P$.

    Conversely, suppose that there exists an $f$-connector $C$ such that $x$ violates the connector inequality associated with $f$ and $C$.
    By \Cref{lemma:ilp}, every connector inequality associated with a non-minimal connector is dominated by a connector inequality associated with a minimal connector. 
    Let $C'$ be a minimal $f$-connector of $G$ with $C' \subseteq C$.
    Since $C'$ is an $f$-minimal it induces an $f$-path $P'=(V_{P'},E_{P'})$ with $V_{P'} = C'$.
    By the assumption that $P$ is a minimum-weight $f$-path, together with \eqref{eq:path-weight-connector-weight} and the violated connector inequality, we obtain
    \[
        w(P) + \tfrac{1}{2}(x_s + x_t) \leq w(P') + \tfrac{1}{2}(x_s + x_t) = \sum_{u \in C'} x_u \leq \sum_{u \in C} x_u < x_{f} \enspace.
    \]

    Using Dijkstra's algorithm, shortest $f$-paths in $G$ for all $f \in F$ can be computed in time $\mathcal{O}(|E||V|+|V|^2 \text{log}|V|)$.
\end{proof}

\begin{proposition}
    Let $G = (V, E)$ be a connected graph, let $F \subseteq \binom{V}{2}$ be a set of interactions and let $x \in [0,1]^{V \cup F}$.
    Then for any $f \in F$, one can find an $f$-separator $S$ of $G$ such that $x$ violates the separator inequality associated with $f$ and $S$ or decide that such a separator does not exist in time $\mathcal{O}(|V|^2|E|)$.
\end{proposition}

\begin{proof}
    Suppose $f = \{u,v\}$. Firstly, we define the auxiliary directed graph $G'=(V', E')$ with 
    \begin{align*}
        V' = {}& \{s',s''\mid s \in V\}\,, \\
        E' = {}& \{(s',s'') \mid s \in V \} \\
        & \cup \{(s'',t'), (t'', s') \mid st \in E,\, s,t \notin \{u,v\}\} \\
        & \cup \{(u'', s') \mid s \in V \setminus \{u\} \text{ with } us \in E\} \\
        & \cup \{(s'', v') \mid s \in V \setminus \{v\} \text{ with } vs \in E\} \,.
    \end{align*}
    Further, we define the non-negative edge weights $w_{(s',s'')} = 1-x_s$ for $s \in V$ and $w_a = 2$ for all other arcs $a \in E' \setminus \{(s',s'') \mid s \in V\}$.
    We show that there exists an $f$-separator $S$ such that $x$ violates the separator inequality associated with $f$ and $S$ if and only if the value of a minimum $(u',v'')$-cut of $G'$ is strictly smaller than $1-x_{f}$.

    Since $(u',u'')$ is the only arc in $E'$ that is adjacent to $u'$, we observe that the set $\{(u',u'')\}$ is a $(u',v'')$-cut of value $w_{(u',u'')} = (1-x_u) \leq 1$.
    Let $S'$ be a minimum $(u',v'')$-cut of $G'$ with weight $w(S') \coloneqq \sum_{a \in S'} w_a$.
    By the above observation and the minimality of $S'$, it holds that $w(S') \leq 1$.
    It follows from the definition of $w$ that $S' \subseteq \{(s',s'') \mid s \in V\}$.
    Let $S \subseteq V$ such that $S' = \{(s',s'')\mid s \in S\}$.
    By our definition of $G'$, the vertex set $S$ is an $f$-separator of $G$. 
    Conversely, every $f$-separator $T \subseteq V$ induces a $(u',v'')$-cut $T'=\{(s',s'') \mid s \in T\}$ of $G'$.
    By the definition of the weights, it holds that $\sum_{s \in S} (1 - x_s) = w(S')$.
    Therefore, there exists a violated separator inequality associated with the interaction $f$ if and only if the minimum $(u',v'')$-cut $S'$ has value $w(S') < (1 - x_{f})$.

    Using Dinic's algorithm a minimum $(u',v'')$-cut of $G'$ can be computed in time $\mathcal{O}(|V'|^2|E'|)$.
    The result follows from the observation $|V'| = 2 |V|$ and $|E'| = |V| + 2|E|$.
\end{proof}

\label{section:separation}

\section{Construction of feasible vectors}
\label{section:lemma-ie}

In this appendix, we introduce a basic technique for constructing from a set $\sigmab$ of feasible vectors of the multi-separator problem vectors in the the linear space
\begin{align*}
	\lin{\sigmab}
    =
	\left\{ 
		\sum_{x \in \sigmab} \alpha_x \, x
		\ \middle|\ 
		\alpha \in \mathbb{R}^{\sigmab}, \
		\sum_{x \in \sigmab} \alpha_x = 0
	\right\}
	\bjoerni{\enspace .}
\end{align*}

\begin{lemma}
    \label{lemma:ie}
    Let $G$ be a graph, let $F$ be a set of interactions on $G$, and let $\sigmab \subseteq \feasible$.
    Suppose that $A$ and $B$ are two vertex subsets of $G$ such that the feasible vectors $x^A$,~$x^B$,~$x^{A \cap B}$ and~$x^{A \cup B}$ are all in $\sigmab$, then
    \begin{align}
        \mathbbm{1}_{Q_{AB}} & \in \lin{\sigmab}\,,
       	\label{eq:ie-1}
        \\
        \mathbbm{1}_{F_{AB}} - \mathbbm{1}_{H_{AB}} & \in \lin{\sigmab}\,.
        \label{eq:ie-2}
    \end{align}
    where\footnote{$A \triangle B$ denotes the symmetric difference of $A$ and $B$.}
    \begin{align*}
        Q_{AB} & \coloneqq
        \left\{
            f \in F \;\middle|\;
            A \cup B \ \textnormal{is an $f$-connector of $G$, but $A \cap B$ is not}
        \right\}
        \;\cup\; (A \triangle B) 
        \,,
		\\
        F_{AB} & \coloneqq
        \left\{
            f \in F \;\middle|\;
            A \cup B \ \textnormal{is an $f$-connector of $G$, but $A$ and $B$ are not}
        \right\}
        \,,
		\\
        H_{AB} & \coloneqq
        \left\{
            f \in F \;\middle|\;
            \textnormal{$A$ and $B$ are $f$-connectors of $G$ but $A \cap B$ is not}
        \right\}
		\,.
    \end{align*}
    \end{lemma}
    
    \begin{proof}
        For the given two vertex subsets $A$ and $B$, we define
        \begin{align*}
	        y^{A, B} & \coloneqq x^{A \cap B} - x^{A \cup B}\,, \\
	        z^{A, B} & \coloneqq x^A + x^B - x^{A \cap B} - x^{A \cup B} 
	        \,.
        \end{align*}
        Both $y^{A, B}$ and $z^{A, B}$ are in the linear space $\lin{\sigmab}$ because $x^A$,~$x^B$,~$x^{A \cap B}$ and~$x^{A \cup B}$ are in $\sigmab$, by the assumption, and the sum of coefficients in the forms that define $y^{A, B}$ and $z^{A, B}$ equals zero.
        Thus, it is sufficient to show 
        \begin{align}
            y^{A, B} & = \mathbbm{1}_{Q_{AB}} \label{eq:ie-1-proof} \,,\\            
            z^{A, B} & = \mathbbm{1}_{F_{AB}} - \mathbbm{1}_{H_{AB}} \label{eq:ie-2-proof}
            \enspace . 
        \end{align}
        Equation~\eqref{eq:ie-1-proof} is a straightforward consequence of the definition of $y^{A, B}$.
        To establish \eqref{eq:ie-2-proof}, we define
        \begin{align*}
            A^\prime & := A \cup \{f \in F \mid f\ \textnormal{is connected in}\ G[A]\} \,,\\
            B^\prime & := B \cup \{f \in F \mid f\ \textnormal{is connected in}\ G[B]\} \,,
        \end{align*}
        and observe that the following equations hold:
        \begin{align*}
            \mathbbm{1}_{A^\prime} & = \mathbbm{1}_{V \cup F} - x^A\,,
            \\
            \mathbbm{1}_{B^\prime} & = \mathbbm{1}_{V \cup F} - x^B\,,
			\\
            \mathbbm{1}_{A^\prime \cup B^\prime} & = \mathbbm{1}_{V \cup F} - x^{A \cup B} - \mathbbm{1}_{F_{AB}}\,,
            \\
            \mathbbm{1}_{A^\prime \cap B^\prime} & = \mathbbm{1}_{V \cup F} - x^{A \cap B} + \mathbbm{1}_{H_{AB}}\,.            
        \end{align*}
        By the inclusion-exclusion principle, the characteristic vectors defined by $A^{\prime}$ and $B^{\prime}$ are related by the identity
        \begin{equation*}
	        \mathbbm{1}_{A^\prime} 
	        + \mathbbm{1}_{B^\prime}
            - \mathbbm{1}_{A^\prime \cap B^\prime} 
            - \mathbbm{1}_{A^\prime \cup B^\prime} = 0\,,
        \end{equation*}
        which, together with the above equalities, implies \eqref{eq:ie-2-proof}.
    \end{proof}
    
    \begin{figure}
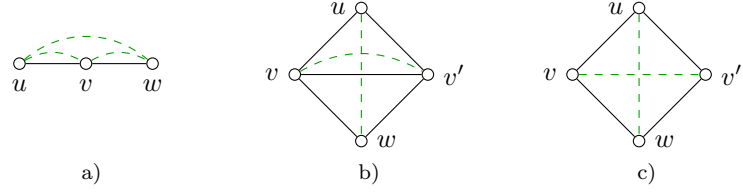

        \centering
        \begin{subfigure}[b]{0.2\textwidth}
            \centering
            \input{figures/ie1.tex}
            \caption{}
            \label{fig:ie1}
        \end{subfigure}
        \hspace{3ex}
        \begin{subfigure}[b]{0.2\textwidth}
            \centering
            \input{figures/ie2.tex}
            \caption{}
            \label{fig:ie2}
        \end{subfigure}
        \hspace{3ex}
        \begin{subfigure}[b]{0.2\textwidth}
            \centering
            \input{figures/ie3.tex}
            \caption{}
            \label{fig:ie3}
        \end{subfigure}
        \caption{Examples illustrating the constructions in \Cref{lemma:ie}.
        Depicted are graphs $G$ (solid, black) and interactions $F$ (dashed, green).}
        \label{fig:ie}
    \end{figure}
    
    The example below illustrates how the constructions $Q_{AB}$,~$F_{AB}$ and~$H_{AB}$ introduced in \Cref{lemma:ie} are computed in a special, but typical, case, i.e., when $A \setminus B$ and $B \setminus A$ contain at most one vertex.
    
    \begin{example}\label{example:ie}
        Let $G$ and $F$ be as in \Cref{lemma:ie}.
        If $A \subseteq B$ and $B \setminus A = \{b\}$, then
        \[
            Q_{AB} =
            \{b\} \cup
            \left\{
                f \in F \mid
                b \ \textnormal{is an $f$-cut-vertex of}\ G[B]
            \right\}
            \,,
        \]
        and $F_{AB} = H_{AB} = \emptyset$.
        If $A$ and~$B$ satisfy $A \setminus B = \{a\}$ and $B \setminus A = \{b\}$, then
        \begin{align*}
            F_{AB} & =
            \left\{
                f \in F \mid
                f \subseteq A \cup B\,,\ \textnormal{both $a$ and $b$ are $f$-cut-vertices of} \ G[A \cup B]
            \right\}
            \,,
            \\
            H_{AB} & =
            \left\{
                f \in F \mid
                f \subseteq A \cap B\,,\ \textnormal{$\{a, b\}$ is a minimal $f$-separator of}\ G[A \cup B]
            \right\}
            \,.
        \end{align*}
        We illustrate these cases using \Cref{fig:ie}.
        \begin{enumerate}[label=\textnormal{(\Alph*)}]
            \item In \Cref{fig:ie}~\subref{fig:ie1},~\subref{fig:ie2} and~\subref{fig:ie3}, let $B$ be the vertex set of $G$ and let $A = B \setminus \{u\}$.
            Then $Q_{AB} = \{u, uv, uw\}$ in \Cref{fig:ie}~\subref{fig:ie1} and $Q_{AB} = \{u, uw\}$ in \Cref{fig:ie}~\subref{fig:ie2},~\subref{fig:ie3}.
            \item In \Cref{fig:ie}~\subref{fig:ie1}, let $A = \{u, v\}$ and $B = \{v, w\}$.
            Clearly, $A \setminus B = \{u\}$ and $B \setminus A = \{w\}$.
            Thus, $F_{AB} = \{uw\}$ and $H_{AB} = \emptyset$.
            \item In \Cref{fig:ie}~\subref{fig:ie2} and~\subref{fig:ie3}, let $A = \{u, v, v'\}$ and $B = \{v, v', w\}$.
            Clearly, $A \setminus B = \{u\}$ and $B \setminus A = \{w\}$.
            In \Cref{fig:ie}~\subref{fig:ie2}, we have $F_{AB} = \{uw\}$ and $H_{AB} = \emptyset$, since $\{u, w\}$ is not a minimal $\{v,v'\}$-separator of $G$.
            In contrast, $\{u, w\}$ is a minimal $\{v, v'\}$-separator of $G$ in \Cref{fig:ie}~\subref{fig:ie3}, and thus $H_{AB} = \{vv'\}$, though $F_{AB} = \{uw\}$ as before.
        \end{enumerate}
    \end{example}

\section{Dimension and pseudoforests}
\label{section:pseudoforest}
In this section, we give an algebraic lemma that we apply in order to establish the sufficiency of conditions for an inequality to induce a facet of the multi-separator polytope.

To begin with, we recall additional, basic notions from graph theory:
A \emph{loop graph} is an ordered pair of disjoint sets $(V, E)$ with $E \subseteq \tbinom{V}{1} \cup \tbinom{V}{2}$.
The incidence matrix of a loop graph $G = (V, E)$ is defined as the binary vertex-edge matrix $M \in \{0,1\}^{V \times E}$ such that $M(v, e) = 1$ if and only if $v \in e$.
A \emph{pseudoforest} is a loop graph such that each of its connected components contains at most one cycle (including loops).
It is immediate from the definition that a loop graph $G = (V, E)$ is a pseudoforest if and only if for any vertex subset $U \subseteq V$, the number of edges in the induced subgraph $G[U]$ is at most the number of vertices in $U$.

%
%

The result below relates pseudoforests to loop graphs whose incidence matrices have full column rank.

\begin{lemma}\label{lemma:pseudoforest}
    The incidence matrix of a loop graph has full column rank if and only if the loop graph is a pseudoforest without even cycles.
\end{lemma}
\begin{proof}
    See \citet{shapiro2020algebras}.
\end{proof}

\section{Proofs}
\label{section:proofs}

\subsection{Proof of Proposition~\ref{prop:connector-separator-minimal}}
\label{section:proof:prop:connector-separator}
For necessity of Part~\ref{item:prop-connector-minimal}:
Since \(U\) is an \({A,B}\)-connector, \(G[U]\) contains an \(ab\)-path for some \(a\in A\), \(b\in B\).
A minimal such path gives a connector contained in \(U\), hence by minimality it uses all vertices of \(U\).
Thus \(G[U]\) is an induced path.
Moreover, if \(U\cap A\) contained a vertex different from the first endpoint, or \(U\cap B\) contained a vertex different from the last endpoint, then a proper subpath would still be an \({A,B}\)-connector, contradicting minimality.
Hence \(U\cap A=\{u_0\}\) and \(U\cap B=\{u_n\}\).

For sufficiency of Part~\ref{item:prop-connector-minimal}:
If \(G[U]\) is such a path, then any \({A,B}\)-connector contained in \(U\) must contain the unique \(A\)-vertex \(u_0\) and the unique \(B\)-vertex \(u_k\).
Since \(G[U]\) is a path, the only \(u_0u_k\)-path in \(G[U]\) uses all vertices of \(U\). Therefore no proper subset of \(U\) is an \({A,B}\)-connector, so \(U\) is minimal.

For sufficiency of Part~\ref{item:prop-separator-minimal}:
We first observe that $U$ is indeed an $\{A, B\}$-separator of $G$, since every $\{A, B\}$-connector of $G$ contains a minimal $\{A, B\}$-connector of $G$.
If $U$ is a non-minimal $\{A, B\}$-separator of $G$, then we can find a vertex $u \in U$ such that $U \setminus \{u\}$ is also an $\{A, B\}$-separator of $G$.
This means that $C \cap U \neq \{u\}$ holds for any minimal $\{A, B\}$-connector $C$ of $G$, a contradiction.

For necessity of Part~\ref{item:prop-separator-minimal}:
Suppose $U$ is a minimal $\{A, B\}$-separator of $G$.
By the definition of separator, every minimal $\{A, B\}$-connector of $G$ intersects $U$.
If there exists a vertex $u \in U$ such that the intersection of $U$ with any minimal $\{A, B\}$-connector $C$ of $G$ is not $\{u\}$, then it follows from the observation in the proof of sufficiency that $U \setminus \{u\}$ is also an $\{A, B\}$-separator of $G$.
This contradicts the minimality of $U$.
\qed

\subsection{Proof of Lemma~\ref{lemma:ilp}}
\label{section:proof:lemma:ilp}

For the first part of the lemma, we verify that
\[
    \feasible =
    \{
        x \in \mathbb{Z}^{V \cup F} \mid
        x \ \textnormal{satisfies} \ \eqref{eq:connector}, \eqref{eq:separator} \ \textnormal{and} \ \eqref{eq:box}
    \}
    \,.
\]
Clearly, every feasible vector $x \in \feasible$ satisfies all the inequalities~\eqref{eq:connector},~\eqref{eq:separator} and~\eqref{eq:box}.
To see the converse, we just need to show that for any $f \in F$ and any $x \in \{0, 1\}^{V \cup F}$ satisfying \eqref{eq:connector} and~\eqref{eq:separator}, the following statements are equivalent:
\begin{enumerate}[label=\textnormal{(\Alph*)}]
    \item $x_f = 0$.
    \label{item:lemma-value-1}
    \item Every $f$-separator of $G$ contains a vertex at which $x$ takes zero value.
    \label{item:lemma-value-2}
    \item There exists an $f$-connector of $G$ whose vertices are labeled zero by $x$.
    \label{item:lemma-value-3}
\end{enumerate}

It is immediate from the separator inequalities~\eqref{eq:separator} that \ref{item:lemma-value-1} implies \ref{item:lemma-value-2}.
If \ref{item:lemma-value-2} holds but every $f$-connector $C$ of $G$ contains at least one vertex at which $x$ takes $1$, let $S$ be the set of all these vertices labeled $1$ by $x$.
Then $S$ is an $f$-separator of $G$, a contradiction.
Thus \ref{item:lemma-value-3} follows from \ref{item:lemma-value-2}.
By the connector inequalities~\eqref{eq:connector}, part~\ref{item:lemma-value-3} implies \ref{item:lemma-value-1}.

For the second part of the lemma, we show that for a given interaction $f \in F$, connector inequalities associated with non-minimal $f$-connectors of $G$ and separator inequalities associated with non-minimal $f$-separators of $G$ are implied by stronger ones and hence redundant in the linear relaxation.
If $C$ is an arbitrary $f$-connector of $G$, then for every minimal $f$-connector $C^{\prime}$ of $G$ contained in $C$, we have
\[
    x_{f}
    \leq
    \sum_{v \in C^{\prime}} x_{v}
    \leq
    \sum_{v \in C} x_{v}
    \,.
\]
That is, the connector inequality for \(C'\) implies the connector inequality for \(C\).
Similarly, let $S$ be any $f$-separator of $G$ and let $S^{\prime}$ be a minimal $f$-separator of $G$ that is contained in $S$, then
\[
    1 - x_{f}
    \leq
    \sum_{v \in S^{\prime}} (1 - x_{v})
    \leq
    \sum_{v \in S} (1 - x_{v})
    \,.
\]
That is, the separator inequality for \(S'\) implies the separator inequality for \(S\).
Therefore, it suffices to enforce connector inequalities only for minimal \(f\)-connectors and separator inequalities only for minimal \(f\)-separators.
This completes the proof.
\qed

\subsection{Proof of Proposition~\ref{prop:dimension}}
\label{section:proof:prop:dimension}
Since $\lin{\feasible}$ is the linear space associated with the affine hull of the multi-separator polytope, it suffices to show that $\mathbbm{1}_{\{g\}} \in \lin{\feasible}$ for every $g \in V \cup F$.

First let \(v\in V\).
Set $A = \{v\}$ and $B = \emptyset$.
Then \(A\triangle B=\{v\}\).
Moreover, since every interaction in \(F\subseteq \binom V2\) has two distinct endvertices, \(A\cup B=\{v\}\) is not an \(f\)-connector for any \(f\in F\).
Hence $Q_{AB} = \{v\}$.
By \eqref{eq:ie-1} of \Cref{lemma:ie}, we obtain $\mathbbm{1}_{\{v\}} \in \lin{\feasible}$.

Now let \(f=\{u,w\}\in F\).
Since \(G\) is connected, there exists a minimal \(f\)-connector \(U\) of $G$.
By \Cref{cor:connector-separator-minimal}, \(G[U]\) is an induced path.
Set $A = U \setminus \{u\}$ and $B = U \setminus \{w\}$.
Then \(A\cup B=U\) and \(A\cap B=U\setminus \{u,w\}\).
Since \(U\) is a minimal \(f\)-connector, neither \(A\) nor \(B\) is an \(f\)-connector, while \(A\cup B\) is.
Hence \(f\in F_{AB}\).
Moreover, if \(g\in F_{AB}\), then \(g\) must have one endvertex in \(U\setminus A=\{u\}\) and the other in \(U\setminus B=\{w\}\), and therefore \(g=f\).
Thus \(F_{AB}=\{f\}\).
Finally, any interaction \(g\in F\) with \(g\subseteq A\cap B\) is connected in $G[A \cap B]$, so \(H_{AB}=\emptyset\).
Therefore, by \eqref{eq:ie-2} of \Cref{lemma:ie}, we have $\mathbbm{1}_{\{f\}} \in \lin{\feasible}$, and the proof is complete.
\qed

\subsection{Proof of Proposition~\ref{prop:canonical-facets}}
\label{section:proof:prop:canonical-facets}

Necessity being trivial, we show sufficiency.
Let us denote $F(a) = \{f \in F \mid a_f \neq 0\}$ and $V(a) = V(a)^{+} \cup V(a)^{-}$, where $V(a)^{+} = \{v \in V \mid a_v > 0\}$ and $V(a)^{-} = \{v \in V \mid a_v < 0\}$.
The sufficiency part asserts that if $\lvert F(a) \rvert \leq 1$, then the given facet-defining inequality
\begin{equation}\label{eq:valid}
    ax \leq \alpha
\end{equation}
must be canonical.
In the following, we consider two possible cases $F(a) = \emptyset$ and $\lvert F(a) \rvert = 1$ separately.

Suppose $F(a) = \emptyset$.
Then inequality~\eqref{eq:valid} has the form
\begin{equation}\label{eq:valid-proof-1}
    \sum_{v \in V(a)^{+}} a_v x_v \leq \sum_{v \in V(a)^{-}} (-a_v) x_v + \alpha\,.
\end{equation}
As $x^{V(a)^{-}}$ is feasible, it holds that $\sum_{v \in V(a)^{+}} a_v \leq \alpha$.
If $\sum_{v \in V(a)^{+}} a_v < \alpha$, then inequality~\eqref{eq:valid-proof-1} has to be strict for all feasible vectors, which contradicts the assumption that \eqref{eq:valid} is facet-defining.
Thus we have $\alpha = \sum_{v \in V(a)^{+}} a_v$, in which case \eqref{eq:valid} becomes
\[
    \sum_{v \in V(a)^{+}} a_v (1 - x_v) + \sum_{v \in V(a)^{-}} (- a_v) x_v \geq 0\,.
\]
This inequality is a linear combination of the box inequalities $x_v \geq 0$ and $1 - x_v \geq 0$ for $v \in V$ with nonnegative coefficients.
Therefore, we must have either $\lvert V(a)^{+} \rvert = 1$ and $V(a)^{-} = \emptyset$, or $V(a)^{+} = \emptyset$ and $\lvert V(a)^{-} \rvert = 1$, which proves the first case $F(a) = \emptyset$.

We continue with the second case $\lvert F(a) \rvert = 1$.
Suppose $F(a) = \{f\}$ for some interaction $f \in F$.
If $V(a)$ is empty, then \eqref{eq:valid} becomes a box inequality associated with $f$.
From now on, we assume that $V(a)$ is non-empty.
For convenience of notation we rewrite inequality \eqref{eq:valid} as
\begin{equation}\label{eq:valid2}
    a_f x_f + \sum_{v \in V(a)^{+}} a_v x_v \leq \sum_{v \in V(a)^{-}} b_v x_v + \alpha\,,
\end{equation}
where $b \in \mathbb{Z}^{V}_{\geq 0}$ is defined by $b_v = - a_v$ if $v \in V(a)^{-}$, and $0$ otherwise.

Firstly, we consider the case $a_f > 0$.
As the feasible vector $x^{V \setminus V(a)^{+}}$ satisfies \eqref{eq:valid2}, it follows that $\alpha \geq \sum_{v \in V(a)^{+}} a_{v}$.
Furthermore, this inequality cannot be strict, i.e., we must have $\alpha = \sum_{v \in V(a)^{+}} a_{v}$, otherwise the equality $x_f = 1$ would hold for any feasible vector $x$ that satisfies \eqref{eq:valid2} at equality, contradicting the assumptions that \eqref{eq:valid2} is facet-defining and $V(a) \neq \emptyset$.
Now, substituting $\alpha = \sum_{v \in V(a)^{+}} a_{v}$ in \eqref{eq:valid2} gives
\begin{equation}\label{eq:valid2-1}
    a_f x_f \leq \sum_{v \in V(a)^{+}} a_v (1 - x_v) + \sum_{v \in V(a)^{-}} b_v x_v\,.
\end{equation}
If $V(a)^{-}$ is not an $f$-connector of $G$, then the feasible vector $x^{V(a)^{-}}$ satisfies $x^{V(a)^{-}}_f = 1$, which implies $a_f \leq 0$ by \eqref{eq:valid2-1}.
This is a contradiction.
Therefore, $V(a)^{-}$ must be an $f$-connector of $G$.
For every $f$-separator $S$ of $G[V(a)^{-}]$, it holds that
\begin{equation}\label{eq:valid2-sep}
    a_f \leq \sum_{v \in S} b_v\,,
\end{equation}
since $x^{V(a)^{-} \setminus S} \in \feasible$.
This implies that $V(a)^{+} = \emptyset$.
Otherwise, since $x_f=1$ for any $x=x^{V \setminus S}$, i.e., $f$ is separated by $S$ in $G$, implies that $S \cap V(a)^{-}$ is an $f$-separator of $G[V(a)^{-}]$, it follows from \eqref{eq:valid2-sep} that the equality $x_v = 1$ holds for every $v \in V(a)^{+}$ and every feasible vector $x$ that satisfies \eqref{eq:valid2-1} at equality.
Hence, \eqref{eq:valid2-1} reduces to
\begin{equation}\label{eq:valid2-1-1}
    a_f x_f \leq \sum_{v \in V(a)^{-}} b_v x_v\,.
\end{equation}

Let us give each vertex $v \in V$ the capacity $b_v$ and each edge an infinite capacity.
By \eqref{eq:valid2-sep} and the fact that inequality~\eqref{eq:valid2-1-1} is facet-defining, we see that the minimal capacity of an $f$-separator is precisely $a_f$.
Then, by the Max-Flow Min-Cut Theorem \cite{dantzig1955max, schrijver2003combinatorial}, there exist nonnegative integers $\kappa_C$ for all $f$-connectors $C$ of $G[V(a)^{-}]$ such that
\[
    \sum_{C \in f\textnormal{-connectors}(G[V(a)^{-}])} \kappa_C = a_f
\]
and
\[
    \sum_{C \in f\textnormal{-connectors}(G[V(a)^{-}])} \kappa_C \mathbbm{1}_{C} \leq b\,.
\]
Here inequalities between vectors are meant to be componentwise.
From this, we see that
\begin{align*}
    & - a_f x_f + \sum_{v \in V(a)^{-}} b_v x_v \\
    = & \sum_{C \in f\textnormal{-connectors}(G[V(a)^{-}])} \kappa_C \left(- x_f + \sum_{v \in C} x_v\right) + \left(b - \sum_{C \in f\textnormal{-connectors}(G[V(a)^{-}])} \kappa_C \mathbbm{1}_C\right) \cdot x|_{V}\,.
\end{align*}
That is, \eqref{eq:valid2-1-1} can be written as a nonnegative linear combination of connector inequalities associated with the interaction $f$ and $f$-connectors of $G$, and box inequalities $x_v \geq 0$ associated with vertices $v \in V(a)^{-}$.
Thus, we conclude that \eqref{eq:valid2-1} must be a connector inequality under the specified assumption.

Next, we consider the other case $a_f < 0$.
As before, since $x^{V(a)^{-}} \in \feasible$, we get $a_f + \sum_{v \in V(a)^{+}} a_v \leq \alpha$ from \eqref{eq:valid2}.
Again, the number $\alpha$ is determined by the vector $a$, that is, $a_f + \sum_{v \in V(a)^{+}} a_v = \alpha$, otherwise we would get the equality $x_f = 0$ that holds for all feasible vectors $x$ that satisfy \eqref{eq:valid2} at equality, contradicting the assumptions that \eqref{eq:valid2} is facet-defining and $V(a) \neq \emptyset$.
Hence, inequality \eqref{eq:valid2} can be written as
\begin{equation}\label{eq:valid2-2}
    - a_f (1 - x_f) \leq \sum_{v \in V(a)^{+}} a_v (1 - x_v) + \sum_{v \in V(a)^{-}} b_v x_v\,,
\end{equation}
If $V(a)^{+}$ is not an $f$-separator of $G$, then the feasible vector $x^{V\setminus V(a)^{+}}$ satisfies $x^{V\setminus V(a)^{+}}_f = 0$, which implies $a_f \geq 0$, a contradiction.
Therefore, $V(a)^{+}$ must be an $f$-separator of $G$.
For each $f$-connector $C$ of $G$, we have
\begin{equation}\label{eq:valid2-path}
    - a_f \leq \sum_{v \in V(a)^{+} \cap C} a_v\,,
\end{equation}
since $x^{C \cup (V\setminus V(a)^{+})} \in \feasible$.
This implies that $V(a)^{-} = \emptyset$.
Otherwise, by \eqref{eq:valid2-path}, we would get an equality $x_v = 0$ for every $v \in V(a)^{-}$ that holds for all feasible vectors $x$ that satisfy \eqref{eq:valid2-2} at equality.
Now, inequality \eqref{eq:valid2-2} takes the form
\begin{equation}\label{eq:valid2-2-1}
    - a_f (1 - x_f) \leq \sum_{v \in V(a)^{+}} a_v (1 - x_v)\,.
\end{equation}

As in the previous case, we assign a weight $a_v$ to each vertex $v \in V$.
Note that the minimum weight of an $f$-connector of $G$ is precisely $-a_f$ by \eqref{eq:valid2-path} and the fact that inequality~\eqref{eq:valid2-2-1} is facet-defining.
Then, by the Max-Potential Min-Work Theorem \cite{duffin1962extremal, schrijver2003combinatorial}, there exist nonnegative integers $\lambda_S$ for all $f$-separators $S$ of $G$ such that
\[
    \sum_{S \in f\textnormal{-separators}(G)} \lambda_S = - a_f
\]
and
\[
    \sum_{S \in f\textnormal{-separators}(G)} \lambda_S \mathbbm{1}_S \leq a|_{V}\,.
\]
Again, inequalities between vectors are componentwise.
Therefore, we have
\begin{align*}
    & a_f (1 - x_f) + \sum_{v \in V(a)^{+}} a_v (1 - x_v) \\
    = & \sum_{S \in f\textnormal{-separators}(G)} \lambda_S \left(- 1 + x_f + \sum_{v \in S} (1 - x_v)\right) + \left(a|_{V} - \sum_{S \in f\textnormal{-separators}(G)} \lambda_S \mathbbm{1}_S\right) \cdot \left(\mathbbm{1}_V - x|_{V}\right)\,.
\end{align*}
That is, \eqref{eq:valid2-2-1} is a nonnegative linear combination of separator inequalities associated with the interaction $f$ and $f$-separators of $G$, and box inequalities $x_v \leq 1$ associated with vertices $v \in V$.
Thus, we arrive at the conclusion that \eqref{eq:valid2-2} must be a separator inequality under the given assumption.
This completes the proof.
\qed

\subsection{Proof of Proposition~\ref{prop:ms-disconnected}}
\label{section:proof:prop:ms-disconnected}

The decomposition of $\polytope$ follows immediately from the fact that interactions in $F(\emptyset)$ are always labeled $1$ by feasible vectors and the label of any interaction not in $F(\emptyset)$ is completely determined by the connected component containing its endvertices.
The remaining statement is just a basic property of polytope product. 
\qed

\subsection{Proof of Proposition~\ref{prop:box}}
\label{section:proof:prop:box}
To see the first statement, let $f = \{u, v\} \in F$, then $\{v\}$ is an $f$-separator of $G$ and the associated separator inequality reads $x_v \leq x_f$, which implies $x_f \geq 0$ under the assumption that $x_v \geq 0$.
As for the second statement, let $\{u, v\}$ be an edge of $G$ incident with $v$, then $x_v \geq 0$ follows from the connector inequality $x_{uv} \leq x_u + x_v$ and the separator inequality $x_u \leq x_{uv}$, and $x_v \leq 1$ follows from the separator inequality $x_v \leq x_{uv}$ and our assumption that $x_{uv} \leq 1$.
\qed

\subsection{Proof of Theorem~\ref{thm:v=0-fd}}
\label{section:proof:thm:box:v=0}

Let $\sigmab = \{x \in \feasible \mid x_v = 0\}$ denote the set of all feasible vectors satisfying the given lower box inequality $x_v \geq 0$ at equality and let $\Sigmab = \{U \subseteq V \mid v \in U\}$ denote the set of all vertex subsets of $G$ containing $v$.
Clearly, every feasible vector $x \in \sigmab$ can be written as $x=x^U$ for some $U \in \Sigmab$ and vice versa.
	
Necessity.
In order to prove necessity, we consider the violation of condition~\ref{item:v=0-fd-1} and condition~\ref{item:v=0-fd-2} in turn and, for either case, construct a linear equality in $\mathbb{R}^{V \cup F}$ that is independent of $x_v = 0$ but satisfied by all $x \in \sigmab$.
Since $\polytope$ is full-dimensional by \Cref{prop:dimension}, this implies that $x_v \geq 0$ is not facet-defining in the case under consideration.

If condition~\ref{item:v=0-fd-1} is violated, then there exists an adjacent vertex $u$ of $v$ in $G$ such that $\{u, v\} \in F$.
In this case every $x \in \sigmab$ satisfies the equality $x_{\{u, v\}} = x_{u}$.
To see this, we first observe that the inequality $x_{u} \leq x_{\{u, v\}}$ is valid, since it is just the separator inequality associated with the interaction $\{u, v\} \in F$ and the $uv$-cut-vertex $u$ of $G$.
The other direction follows from the condition $x_v = 0$ and the connector inequality $x_{\{u, v\}} \leq x_u + x_{v}$ associated with the interaction $\{u, v\} \in F$ and the $uv$-connector $\{u, v\}$ of $G$, since $\{u, v\} \in E \cap F$.

If condition~\ref{item:v=0-fd-2} is violated, then there exist two vertices $u$ and $w$ in $V$ such that $\{u, v\} \in E$ and $\{v, w\}, \{u, w\} \in F$ and $u$ is a $vw$-cut-vertex of $G$.
In this case we claim that the equality $x_{vw} = x_{uw}$ holds for all $x \in \sigmab$ by the following argument.
Let $x$ be an arbitrary feasible vector in $\sigmab$ of the form $x=x^{U}$ for some vertex subset $U \in \Sigmab$.
If $x_{uw} = 0$, i.e., $U$ is a $uw$-connector of $G$, then $U$ is also a $vw$-connector of $G$ by the assumption that $\{u, v\} \in E$ and the fact that $v \in U$.
This means that $x_{vw} = 0$, and hence establishes the inequality $x_{vw} \leq x_{uw}$ for any $x \in \sigmab$.
For the other direction, we suppose that $x_{vw} = 0$, i.e., $U$ is a $vw$-connector of $G$.
Then, since $u$ is assumed to be a $vw$-cut-vertex of $G$, we see that $U$ is also a $uw$-connector of $G$, i.e., $x_{uw} = 0$.
This yields the inequality $x_{uw} \leq x_{vw}$ for any $x \in \sigmab$, as desired.

Sufficiency.
Assume conditions~\ref{item:v=0-fd-1} and~\ref{item:v=0-fd-2} are satisfied.
In the following, we show that for every vertex except $v$, and for every interaction, the corresponding standard unit vector is in the linear space $\lin{\sigmab}$ associated with the face under consideration.

For any vertex $v^\prime$ that is not $v$, set $A = \{v\}$ and $B = \{v, v^{\prime}\}$.
Then, by condition~\ref{item:v=0-fd-1}, we have $Q_{AB} = \{v^{\prime}\}$.
As both $\{v\}$ and $\{v, v^\prime\}$ belong to $\Sigmab$, it follows from \eqref{eq:ie-1} of \Cref{lemma:ie} that
\[
    \mathbbm{1}_{\{v^\prime\}} \in \lin{\sigmab}\,.
\]
That is, all the standard unit vectors associated with vertices different from $v$ are in $\lin{\sigmab}$.

For any interaction $f = \{u, w\} \in F$, we consider the following possible cases depending on the relation between $v$ and minimal $f$-connectors of $G$.
\begin{enumerate}[label=\textnormal{(\Alph*)}]
    \item Suppose that there exists an $f$-path $P = (V_P, E_P)$ in $G$ such that $v$ is an internal vertex of $P$.
    In this case, we set $A = V_P \setminus \{w\}$ and $B = V_P \setminus \{u\}$.
    Then, $A \setminus B = \{u\}$, $B \setminus A = \{w\}$, $A \cup B = V_P$ and $A \cap B = V_P \setminus \{u, w\}$, from which it follows easily that $F_{AB} = \{f\}$ and $H_{AB} = \emptyset$.
    Moreover, all $A$,~$B$,~$A \cap B$, and~$A \cup B$ belong to $\Sigmab$.
    Therefore, by \eqref{eq:ie-2} of \Cref{lemma:ie}, we get $\mathbbm{1}_{\{f\}} \in \lin{\sigmab}$.
    \label{item:v-a}
    \item Suppose that case~\ref{item:v-a} does not hold, i.e., every $f$-path in $G$ does not contain $v$ as an internal vertex, but we can find a minimal $f$-connector $C$ of $G$ such that $v \not\in C$ and $G[C \cup \{v\}]$ is not a path, i.e., such that $v$ belongs to neither $N_G(u) \setminus C$ nor $N_G(w) \setminus C$.
    An example is depicted in \Cref{fig:box-v=0-proof}~\subref{fig:box-v=0-proof-1}.
    In this case we observe that $G[C \cup \{v\}]$ cannot be a chordless cycle, otherwise $G[f \cup \{v\}]$ would be an $f$-path with $v$ as its internal vertex, a contradiction.
    Set $A = (C \cup \{v\}) \setminus \{w\}$ and $B = (C \cup \{v\}) \setminus \{u\}$.
    Then, $A\setminus B = \{u\}, B\setminus A = \{w\}, A \cup B = C \cup \{v\}$ and $A \cap B = (C \cup \{v\}) \setminus \{u, w\}$.
    Taking into account the preceding observation, we see that $F_{AB} = \{f\}$ and $H_{AB} = \emptyset$.
    Moreover, all $A$,~$B$,~$A \cap B$ and~$A \cup B$ belong to $\Sigmab$.
    Therefore, it follows from \eqref{eq:ie-2} of \Cref{lemma:ie} that $\mathbbm{1}_{\{f\}} \in \lin{\sigmab}$.
    \label{item:v-b}
    \item Suppose that neither case~\ref{item:v-a} nor case~\ref{item:v-b} is true.
    This implies that, for any minimal $f$-connector $C$ of $G$, the induced subgraph $G[C \cup \{v\}]$ is a path and $v \not\in C \setminus f$.
    In particular, $v$ lies in the closed neighborhood of $f$ in $G$.
    Let $u^\prime$ denote the unique vertex of $C$ that is adjacent to $v$ in $G$.
    Without loss of generality, we assume that $w$ is the endvertex of $G[C \cup \{v\}]$ that is different from $v$.
    So, the other endvertex $u$ of $f$ could be equal to either $v$ or $u^\prime$, depending on whether $v$ is in $C$ or not.
    See \Cref{fig:box-v=0-proof}~\subref{fig:box-v=0-proof-2} and~\subref{fig:box-v=0-proof-3} for illustration.
    Clearly, $u' \neq w$ when $u' = u$.
    On the other hand, if $u' \neq u$, then $v = u$ and $u'$ is adjacent to $v$ in $G[C]$, and hence by condition~\ref{item:v=0-fd-1} and the fact that $\{u, w\} \in F$, we again have $u' \neq w$.
    So, no matter which is the case, $u^\prime$ and $w$ are always distinct vertices.
    Let us define $A = (C \cup \{v\}) \setminus \{w\}$ and $B = (C \cup \{v\}) \setminus \{u^\prime\}$.
    Then $A \setminus B = \{u^\prime\}$, $B \setminus A = \{w\}$, $A \cup B = C \cup \{v\}$ and $A \cap B = (C \cup \{v\}) \setminus \{u^\prime, w\}$.
    It is direct to check that $F_{AB} = \{u^\prime w, v w\} \cap F$ and $H_{AB} = \emptyset$.
    Moreover, all the subsets $A$,~$B$,~$A \cap B$ and~$A \cup B$ belong to $\Sigmab$.
    Then, by \eqref{eq:ie-2} of \Cref{lemma:ie}, we obtain $\mathbbm{1}_{F_{AB}} \in \lin{\sigmab}$.
    Now we observe that, if $F_{AB} \setminus \{f\}$ is non-empty, then it is a singleton and the unique interaction in it falls into case~\ref{item:v-a}, from which we see that the standard unit vector corresponding to this unique interaction belongs to $\lin{\sigmab}$ already by the previous discussion. 
    In fact, let $f^\prime$ be the unique interaction in $\{u^\prime w, v w\}$ that is different from $f$.
    Since there exists a minimal $vw$-connector of $G$ that does not contain $u^\prime$ by condition~\ref{item:v=0-fd-2}, it follows from the assumption of case~\ref{item:v-c} that this can only happen when $u=v$, in which case $f^\prime$ coincides with $u^\prime w$ and it can be connected by a path in $G$ that includes $v$ as an internal vertex, which is exactly case~\ref{item:v-a}.
    So, either $F_{AB} \setminus \{f\}$ is empty or it contains precisely one interaction satisfying case~\ref{item:v-a}, as illustrated by \subref{fig:box-v=0-proof-2} and \subref{fig:box-v=0-proof-3} of \Cref{fig:box-v=0-proof}, respectively.
    In any case, it always holds that $\mathbbm{1}_{F_{AB} \setminus \{f\}} \in \lin{\sigmab}$.
    We thus get $\mathbbm{1}_{\{f\}} = \mathbbm{1}_{F_{AB}} - \mathbbm{1}_{F_{AB} \setminus \{f\}} \in \lin{\sigmab}$, as $f \in F_{AB}$.
    \label{item:v-c}
\end{enumerate}
In summary, for every interaction in $F$ we have constructed its corresponding standard unit vector in $\lin{\sigmab}$.

\begin{figure}
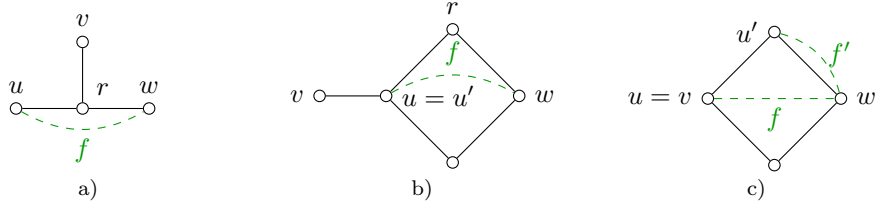

	\centering
	\begin{subfigure}[b]{0.25\textwidth}
		\centering
		\input{figures/box/v=0-proof-1.tex}
		\caption{}
		\label{fig:box-v=0-proof-1}
	\end{subfigure}
	\hspace{3ex}
	\begin{subfigure}[b]{0.25\textwidth}
		\centering
		\input{figures/box/v=0-proof-2.tex}
		\caption{}
		\label{fig:box-v=0-proof-2}
	\end{subfigure}
	\hspace{3ex}
	\begin{subfigure}[b]{0.25\textwidth}
		\centering
		\input{figures/box/v=0-proof-3.tex}
		\caption{}
		\label{fig:box-v=0-proof-3}
	\end{subfigure}

	\caption{Depicted above are graphs $G$ (black, solid), interactions (green, dashed) and a specified vertex $v$, which defines a lower box inequality $x_v \geq 0$.
    Both conditions~\ref{item:v=0-fd-1} and~\ref{item:v=0-fd-2} are satisfied.
    Case~\ref{item:v-b} in \Cref{section:proof:thm:box:v=0} applies to \subref{fig:box-v=0-proof-1}, where the required minimal $f$-connector can be chosen as $C = \{u, r, w\}$.
    Case~\ref{item:v-c} in \Cref{section:proof:thm:box:v=0} applies to \subref{fig:box-v=0-proof-2} and \subref{fig:box-v=0-proof-3}, where the required minimal $f$-connectors can be chosen as $C = \{u, r, w\}$ and $C = \{u, u', w\}$, respectively.
    Note that $u'$ is the unique vertex of $C$ that is adjacent to $v$ in $G$ in both \subref{fig:box-v=0-proof-2} and \subref{fig:box-v=0-proof-3}.
    The set $F_{AB}$ constructed in case~\ref{item:v-c} is $\{f\}$ for \subref{fig:box-v=0-proof-2} and $\{f, f'\}$ for \subref{fig:box-v=0-proof-3}.
	}
	\label{fig:box-v=0-proof}
\end{figure}

By all the results above, we see that the linear subspace $\lin{\sigmab}$ has a basis $\{\mathbbm{1}_{\{v'\}} \mid v' \in V \setminus \{v\}\} \cup \{\mathbbm{1}_{\{f\}} \mid f \in F\}$ and hence is of codimension~$1$ in $\mathbb{R}^{V \cup F}$, and so is the face defined by $\sigmab$.
Therefore, in view of \Cref{prop:dimension}, we conclude that the lower box inequality $x_v \geq 0$ defines a facet of the multi-separator polytope $\polytope$ under the given conditions, and the proof is complete.
\qed

\subsection{Proof of Theorem~\ref{thm:f=1-fd}}
\label{section:proof:thm:box:f=1}

Let $\sigmab = \{x \in \feasible \mid x_f = 1\}$ denote the set of all feasible vectors satisfying the upper box inequality $x_f \leq 1$ at equality and let $\Sigmab = \{U \subseteq V \mid V \setminus U \ \textnormal{is an $f$-separator of} \ G\}$ denote the set of all vertex subsets of $G$ that are not $f$-connectors of $G$.
Clearly, every vector $x$ in $\sigmab$ has the form $x=x^{U}$ for some vertex subset $U$ in $\Sigmab$ and vice versa.

Necessity.
Suppose to the contrary that there is an interaction $f^\prime \in F \setminus \{f\}$ such that $f$ is a pair of $f^\prime$-cut-vertices of $G$.
This implies that every $f^{\prime}$-connector of $G$ is also an $f$-connector of $G$, and thus the inequality $x_f \leq x_{f^\prime}$ is valid for $\polytope$.
As $x_{f'} \leq 1$ is also valid, we immediately see that the equality $x_{f^{\prime}} = 1$ holds for all $x \in \sigmab$.
Therefore, the face defined by $\sigmab$ is not a facet of $\polytope$ by \Cref{prop:dimension}.

Sufficiency.
Assume that $f$ is not a pair of $f^{\prime}$-cut-vertices of $G$ for all other interactions $f^{\prime}$.
We show that the given inequality defines a facet by constructing a basis of cardinality $\lvert V \rvert + \lvert F \rvert - 1$ of the associated linear space $\lin{\sigmab}$.

Firstly, for any vertex $v \in V$, it is clear that $\emptyset, \{v\} \in \Sigmab$ and $Q_{\emptyset \{v\}} = \{v\}$.
Then it follows straightforwardly from \eqref{eq:ie-1} of \Cref{lemma:ie} that $\mathbbm{1}_{\{v\}} \in \lin{\sigmab}$.

Secondly, we consider any interaction $f^\prime = \{u, w\} \in F\setminus\{f\}$.
By assumption, we can find a minimal $f^\prime$-connector $C$ of $G$ that contains at most one of the endvertices of $f$.
Set $A = C \setminus \{w\}$ and $B = C \setminus \{u\}$.
Then, $F_{AB} = \{f^\prime\}$ and $H_{AB} = \emptyset$.
It is clear that all $A$,~$B$,~$A \cap B$ and~$A \cup B$ belong to $\Sigmab$.
Thus, it follows from \eqref{eq:ie-2} of \Cref{lemma:ie} that $\mathbbm{1}_{\{f^\prime\}} \in \lin{\sigmab}$.

Therefore, $\lin{\sigmab}$ has codimension $1$ in $\mathbb{R}^{V \cup F}$ and the face defined by $\sigmab$ is a facet of $\polytope$ by \Cref{prop:dimension}.
\qed

\subsection{Proof of Theorem~\ref{thm:v=1-fd}}
\label{section:proof:thm:box:v=1}

Let $\sigmab = \{x \in \feasible \mid x_v = 1\}$ denote the set of all feasible vectors satisfying the upper box inequality $x_v \leq 1$ at equality and let $\Sigmab = \{U \subseteq V \mid v \not\in U\}$ denote the set of all vertex subsets of $G$ not containing $v$.
Clearly, the elements in $\Sigmab$ are exactly those vertex subsets $U$ of $G$ satisfying $x^{U} \in \sigmab$.
	
Necessity.
Suppose there exists an interaction \(f\in F\) such that \(v\) is an \(f\)-cut-vertex of \(G\).
Then, by the separator inequality $x_v \leq x_f$ and the box inequality $x_f \leq 1$, we see that the equality $x_f = 1$ holds for all $x \in \sigmab$.
This implies that the face defined by $\sigmab$ is not a facet of $\polytope$ by \Cref{prop:dimension}.

Sufficiency.
Assume that there is no interaction with $v$ as its cut-vertex.
To establish the dimension of the face defined by $\sigmab$, we apply the technique from \Cref{lemma:ie} to construct a basis for the associated linear space $\lin{\sigmab}$ as follows.

Firstly, for any vertex $u \in V\setminus \{v\}$, it is clear that $\emptyset, \{u\} \in \Sigmab$ and $Q_{\emptyset \{u\}} = \{u\}$.
Thus $\mathbbm{1}_{\{u\}} \in \lin{\sigmab}$ by \eqref{eq:ie-1} of \Cref{lemma:ie}.

Secondly, let us consider an arbitrary interaction $f = \{u, w\} \in F$.
By assumption, we may choose a minimal $f$-connector $C$ of $G$ such that it does not pass through $v$.
Set $A = C \setminus \{u\}$ and $B = C \setminus \{w\}$.
Then $F_{AB} = \{f\}$ and $H_{AB} = \emptyset$.
Since all $A$,~$B$,~$A \cap B$ and~$A \cup B$ belong to $\Sigmab$, it follows from \eqref{eq:ie-2} immediately that $\mathbbm{1}_{\{f\}} \in \lin{\sigmab}$.

Therefore, we have shown that $\lin{\sigmab}$ has codimension $1$ in $\mathbb{R}^{V \cup F}$.
Since $\polytope$ is full-dimensional by \Cref{prop:dimension}, the proof follows immediately.
\qed

\subsection{Proof of Theorem~\ref{thm:connector-fd}}
\label{section:proof:thm:connector}

Let
\begin{equation*}
    \sigmab = \left\{x \in \feasible \ \middle|\ x_f = \sum_{v \in C} x_v \right\}
\end{equation*}
denote the set of all feasible vectors that satisfy connector inequality~\eqref{eq:connector-fd} associated with $f$ and $C$ at equality.
To make the structure of $\sigmab$ more explicit, we define
\[
    \Sigmab =
    \left\{
        U \subseteq V \ \middle|\
        \textnormal{either}\ C \subseteq U, \ \textnormal{or}\ U \ \textnormal{is not an $f$-connector of $G$ and $C \setminus U$ is a singleton}
    \right\}
\]
as the set of all vertex subsets of $G$ whose complement in $V$ is either disjoint from $C$, or an $f$-separator of $G$ intersecting $C$ at a single vertex.
Then, a feasible vector $x$ is in $\sigmab$ if and only if there is an element $U$ of $\Sigmab$ such that $x=x^{U}$.
For the convenience of studying connector inequality~\eqref{eq:connector-fd}, given any neighbor $v$ of $C$, we denote $\textnormal{BP}_v$ as the set consisting of $v$ and all vertices $c \in C$ such that $v$ is a $c$-bypass of $C$.

\emph{Necessity}.
Since the multi-separator polytope $\polytope$ is full-dimensional by \Cref{prop:dimension} and non-minimal connector inequalities are implied by other canonical inequalities by \Cref{lemma:ilp}, it follows that condition~\ref{item:connector-fd-1} is necessary for the face defined by $\sigmab$ to be a facet of $\polytope$.

So, without loss of generality, in the remainder of the necessity proof we assume that condition~\ref{item:connector-fd-1} always holds, i.e., $C$ is a minimal $f$-connector of $G$.

Suppose that condition~\ref{item:connector-fd-2} is not satisfied.
Then, there is an interaction $f^\prime$ that is different from $f$ and contained in $C$.
Let $C^\prime$ be the unique minimal $f^\prime$-connector of $G[C]$.
As the union of an $f'$-connector of $G$ and $C \setminus C'$ is an $f$-connector of $G$, we see that the inequality
\begin{equation}\label{eq:connector-fd-1}
    x_{f} \leq x_{f^\prime} + \sum_{v \in C\setminus C^\prime} x_v
\end{equation}
is valid for $\polytope$.
Thus, the connector inequality~\eqref{eq:connector-fd} can be written as a sum of two valid inequalities, i.e., inequality~\eqref{eq:connector-fd-1} and the connector inequality associated with $f^\prime$ and $C^\prime$.
This contradicts the assumption that the convex hull of $\sigmab$ is a facet of $\polytope$ and the fact that $\polytope$ is full-dimensional by \Cref{prop:dimension}.

Suppose that condition~\ref{item:connector-fd-3} does not hold; that is, there is an interaction $uw \in F$ with $u \in N_G(C)$ and $w \in C$ such that $u$ is a $w$-bypass of $C$.
Observe that the inequality $x_u \leq x_{uw}$ is valid as it is just a separator inequality.
We show that every $x \in \sigmab$ with $x_u = 0$ also satisfies $x_{uw} = 0$, from which it follows that the equality $x_{u} = x_{uw}$ holds for all $x \in \sigmab$, and thus the connector inequality cannot be facet-defining by \Cref{prop:dimension}.
Suppose $x_{u} = 0$ for some $x = x^{U} \in \sigmab$ with $U \in \Sigmab$.
Then we argue that $x_{v^\prime} = 0$ holds for all $v^\prime \in \textnormal{BP}_u$.
Since $\textnormal{BP}_u$ together with $N_G(u) \cap C \cap U$, which is nonempty by definitions of $u$ and $U$, forms a $uw$-connector of $G$, this yields $x_{uw} = 0$.
In fact, if $x_{v^\prime} = 1$ for some vertex $v^\prime \in \textnormal{BP}_u$, as $x \in \sigmab$, it follows that $x_f = 1$ and $x_{v^{\prime\prime}} = 0$ for all $v^{\prime\prime} \in C \setminus \{v^\prime\}$, which implies that $(C \cup \{u\}) \setminus \{v^\prime\}$ cannot be an $f$-connector of $G$, contradicting the fact that $u$ is a $v^\prime$-bypass of $C$.

Suppose that condition~\ref{item:connector-fd-4} does not hold; that is, there are two interactions $uw, uw^\prime \in F$ with $u \not\in N_G[C]$, $w \in N_G(C)$ and $w^\prime \in C$ such that $w$ is both a $\{\{u\}, C\}$-cut-vertex of $G$ and a $w^\prime$-bypass of $C$.
If $x_{uw^\prime} = 0$ for some $x \in \feasible$, then it is also true that $x_{uw} = 0$ as $w$ is a $\{\{u\}, C\}$-cut-vertex of $G$.
This implies that the inequality $x_{uw} \leq x_{uw^\prime}$ is valid.
If $x_{uw} = 0$ for some $x \in \sigmab$, just as in the previous paragraph, then we have $x_{v^\prime} = 0$ for all $v^\prime \in \textnormal{BP}_w$, which yields $x_{uw^\prime} = 0$, and hence the inequality $x_{uw} \geq x_{uw^\prime}$.
Therefore, the equality $x_{uw} = x_{uw^\prime}$ holds for all $x \in \sigmab$, which again leads to a contradiction by the same argument as above.

Suppose that condition~\ref{item:connector-fd-5} does not hold; that is, there are two distinct interactions $uw, uw^\prime \in F$ with $u \not\in N_G[C]$ and $w, w^\prime \in C$ such that the set
\[
    W = \left\{v \in N_G(C) \mid v \textnormal{ is both a} \ w \text{-bypass and a} \ w^\prime \textnormal{-bypass of} \ C\right\}
\]
is a $\{\{u\}, C\}$-separator of $G$.
If $x_{uw} = 0$ for some $x \in \sigmab$, then there exists a minimal $uw$-connector $C^{\prime}$ of $G$ such that $x_{r} = 0$ for all $r \in C^{\prime}$.
By definition, $C'$ is also a $\{\{u\}, C\}$-connector of $G$ and therefore must intersect $W$, which implies that there exists a $w$- and $w'$-bypass $v \in N_G(C)$ of $C$ satisfying $x_v = 0$.
As in the previous paragraph, by considering $x_v = 0$, we get $x_{v^\prime} = 0$ for all $v^\prime \in \textnormal{BP}_v$.
This implies that $x_{uw^\prime} = 0$, since $C'$ together with $\textnormal{BP}_v$ forms a $uw'$-connector of $G$, and hence the inequality $x_{uw} \geq x_{uw^\prime}$.
A symmetric argument yields that, if $x_{uw^\prime} = 0$ for some $x \in \sigmab$, then it holds that $x_{uw} = 0$, and hence the inequality $x_{uw} \leq x_{uw^\prime}$.
Therefore, the equality $x_{uw} = x_{uw^\prime}$ holds for all $x \in \sigmab$, which again leads to a contradiction as before.

\emph{Sufficiency}.
In order to show that connector inequality~\eqref{eq:connector-fd} defines a facet, we construct a basis of cardinality $\lvert V \rvert + \lvert F \rvert - 1$ of its associated linear space $\lin{\sigmab}$ by making use of the following results.

\begin{claim}\label{claim:connector}
    Assume that $C$ is a minimal $f$-connector of $G$, let $v \in N_G(C)$ be a neighbor of $C$.
    \begin{enumerate}[label=\textnormal{(\alph*)}]
        \item \label{item:connector-claim-1}
        It holds that
        \begin{equation}\label{eq:connector-1-1}
            \mathbbm{1}_{\{v\} \cup \{vw \in F \mid w \in \textnormal{BP}_v\}} \in \lin{\sigmab}\,,
        \end{equation}
        and, for any interaction $vw \in F$ with $w \in C \setminus \textnormal{BP}_v$,
        \begin{equation}\label{eq:connector-1-2}
            \mathbbm{1}_{\{vw\}} \in \lin{\sigmab}\,.
        \end{equation}
        \item \label{item:connector-claim-2}
        For any vertex $u$ that is not in $N_G[C]$, and such that there exists a minimal $\{\{u\}, N_G[C]\}$-connector of $G$ containing $v$, it holds that
        \begin{equation}\label{eq:connector-2-1}
            \mathbbm{1}_{\{uw \in F \mid w \in \textnormal{BP}_v\}} \in \lin{\sigmab}\,,
        \end{equation}
        and, for any interaction $uw \in F$ with $w \in C \setminus \textnormal{BP}_v$,
        \begin{equation}\label{eq:connector-2-2}
            \mathbbm{1}_{\{uw\}} \in \lin{\sigmab}\,.
        \end{equation}
    \end{enumerate}
\end{claim}
\begin{claimproof}
    As we shall see below, the constructions of the characteristic vectors in case~\ref{item:connector-claim-1} and case~\ref{item:connector-claim-2} can be done completely analogously.
    To this end, we assume that $u$ is a vertex that is not in $C$ and $v$ is a neighbor of $C$ that is on a minimal $\{\{u\}, N_G[C]\}$-connector of $G$.
    So, if $u$ is a neighbor of $C$, then $u = v$, which is precisely case~\ref{item:connector-claim-1}.
    Conversely, if $u$ is not a neighbor of $C$, then $u \neq v$, which is precisely case~\ref{item:connector-claim-2}.

    The remainder of this proof is structured as follows.
    Initially, we construct some auxiliary vectors in $\lin{\sigmab}$ which are characteristic vectors of specific sets of interactions that are incident with $u$ plus possibly $u$ itself.
    Subsequently, we show that the characteristic vectors in the statement of the claim can be obtained as linear combinations of these auxiliary vectors and thus also belong to $\lin{\sigmab}$.

    For our purposes here, let $\tilde{C}$ denote an arbitrary minimal $uv$-connector of $G$ with $\tilde{C} \cap N_G(C) = \{v\}$, which indeed exists by definition of $v$, and write $U = C \cup \tilde{C}$.
    We define
    \[
        A = U \setminus \{u\}\,,
    \]
    \[
        B = (U \setminus \{v\}) \cup \{u\}\,,
    \]
    and
    \[
        B(t) = U \setminus \{t\}
    \]
    for every $t \in C \setminus \textnormal{BP}_v$, i.e., $v$ is not a $t$-bypass of $C$.
    Then, by construction,
    \[
        A \setminus B = \{v\} \setminus \{u\}\,, \quad B \setminus A = \{u\}\,, \quad A \cap B = U \setminus \{u, v\}\,, \quad A \cup B = U\,,
    \]
    and
    \[
        A \setminus B(t) = \{t\}\,, \quad B(t) \setminus A = \{u\}\,, \quad A \cap B(t) = U \setminus \{u, t\}\,, \quad A \cup B(t) = U\,.
    \]
    Looking at the definition of a bypass vertex, we see that for all $t \in C \setminus \textnormal{BP}_v$, neither $B(t)$ nor $A \cap B(t)$ is an $f$-connector of $G$.
    Moreover, $C \setminus B(t) = C \setminus (A \cap B(t)) = \{t\}$.
    Also note that $A$, $B$, $A \cap B$, $A \cup B$ and $A \cup B(t)$ all contain the $f$-connector $C$.
    In view of the definition of $\Sigmab$, we arrive at the following observation:
    \begin{equation}\label{eq:connector-claim-observation}
        A, B, B(t), A \cap B, A \cup B, A \cap B(t), A \cup B(t) \in \Sigmab\,.
    \end{equation}
    It follows directly from the definitions of $Q_{AB}$ and $F_{AB}$ in \Cref{lemma:ie} that
    \[
        Q_{AB}
        = (B \setminus A) \cup \{f^\prime \in F \mid f^\prime \subseteq U, \ u \ \textnormal{is an $f^\prime$-cut-vertex of} \ G[U]\}
        = \{v\} \cup \{vv^\prime \in F \mid v^\prime \in C\}
    \]
    if $u = v$ (note that in this case we have $A =C$ and $B = C \cup \{v\}$), and
    \[
        F_{AB}
        = \{f^\prime \in F \mid f^\prime \subseteq U, \ u \ \textnormal{and} \ v \ \textnormal{are $f^\prime$-cut-vertices of} \ G[U]\}
        = \{uv^\prime \in F \mid v^\prime \in C \cup \{v\}\}
    \]
    otherwise.
    Similarly, for every $t \in C \setminus \textnormal{BP}_v$, we have
    \[
        F_{AB(t)}
        = \{f^\prime \in F \mid f^\prime \subseteq U, \ u \ \textnormal{and} \ t \ \textnormal{are $f^\prime$-cut-vertices of} \ G[U]\}
        = \{uv^\prime \in F \mid v^\prime \in C(v, t)\}\,,
    \]
    where $C(v, t)$ is the set of all vertices $v^\prime$ in $C$ such that $t$ is a $vv^\prime$-cut-vertex of $G[C \cup \{v\}]$.
    By the definition of $H_{AB}$ in \Cref{lemma:ie}, we have
    \[
        H_{AB} = \emptyset
    \]
    if $u = v$, and
    \[
        H_{AB} = \{f^\prime \in F \mid f^\prime \subseteq U \setminus \{u, v\}, \ \{u, v\} \ \textnormal{is a minimal $f^\prime$-separator of} \ G[U]\} = \emptyset
    \]
    otherwise, upon noting that $\{u, v\}$ is a minimal $f^\prime$-separator of $G[U]$ if and only if $U \setminus \{u\}$ and $U \setminus \{v\}$ are $f^\prime$-connectors of $G$ but not $U \setminus \{u, v\}$.
    Similarly, for every $t \in C \setminus \textnormal{BP}_v$, we have
    \[
        H_{AB(t)} = \{f^\prime \in F \mid f^\prime \subseteq U \setminus \{u, t\}, \ \{u, t\} \ \textnormal{is a minimal $f^\prime$-separator of} \ G[U]\} = \emptyset\,.
    \]
    Thus, it follows from \eqref{eq:connector-claim-observation} and \eqref{eq:ie-1} of \Cref{lemma:ie} that
    \begin{equation}\label{eq:connector-claim-0}
        \mathbbm{1}_{Q_{AB}} = \mathbbm{1}_{\{v\}} + \mathbbm{1}_{\{vv^\prime \in F \mid v^\prime \in C\}} \in \lin{\sigmab}
    \end{equation}
    when $u = v$.
    In addition, by \eqref{eq:connector-claim-observation} and \eqref{eq:ie-2} of \Cref{lemma:ie}, we have
    \begin{equation}\label{eq:connector-claim-1}
        \mathbbm{1}_{F_{AB}} - \mathbbm{1}_{H_{AB}} = \mathbbm{1}_{\{uv^\prime \in F \mid v^\prime \in C \cup \{v\}\}} \in \lin{\sigmab}
    \end{equation}
    when $u \neq v$, and
    \begin{equation}\label{eq:connector-claim-2}
        \mathbbm{1}_{F_{AB(t)}} - \mathbbm{1}_{H_{AB(t)}} = \mathbbm{1}_{\{uv^{\prime} \in F \mid v^{\prime} \in C(v, t)\}} \in \lin{\sigmab}
    \end{equation}
    for all $t \in C \setminus \textnormal{BP}_v$.

    To understand $C(v, t)$ more specifically, let us fix a natural linear order $c_0 < \dots < c_{\alpha} \leq \dots \leq c_{\beta} < \dots < c_n$ on $C$, i.e., $c_{i-1}c_{i} \in E$ for $i = 1, \dots, n$, where $\alpha$ (respectively, $\beta$) is the minimal (respectively, maximal) index such that $c_{\alpha}$ (respectively, $c_{\beta}$) is adjacent to $v$ in $G$.
    Then, $\textnormal{BP}_v \cap C = \{c_{\alpha + 1}, \dots, c_{\beta - 1}\}$.
    Note that $\alpha$ and $\beta$ could coincide, which can only happen in the case that $v$ is not a bypass of $C$.
    Let $t = c_h$ be a vertex of $C$ that is not in $\textnormal{BP}_v$, i.e., $h \leq \alpha$ or $h \geq \beta$.
    If $\alpha = \beta = h$, then $C(v, t) = C$.
    Otherwise, $C(v, t) = \{c_k \in C \mid k \leq h\}$ if $h \leq \alpha$, and $C(v, t) = \{c_k \in C \mid k \geq h\}$ if $h \geq \beta$.

    If we generalize our notation a little bit and define $uv$ to be $v$ whenever $u = v$, then \eqref{eq:connector-claim-0},~\eqref{eq:connector-claim-1} and~\eqref{eq:connector-claim-2} together imply that
    \begin{equation}\label{eq:connector-claim-3}
        \mathbbm{1}_{\{uw \in N_G(C) \cup F \mid w \in \textnormal{BP}_v\}} = \mathbbm{1}_{\{uv^\prime \in N_G(C) \cup F \mid v^\prime \in C \cup \{v\}\}} - \sum_{t \in \{c_{\alpha}, c_{\beta}\}} \mathbbm{1}_{F_{AB(t)}} \in \lin{\sigmab}\,,
    \end{equation}
    which yields \eqref{eq:connector-1-1} if $u = v$, and \eqref{eq:connector-2-1} if $u \neq v$.
    For the case $u \neq v$, examples of the construction above are depicted in \Cref{fig:connector-claim-1} and \Cref{fig:connector-claim-2} for $\alpha \neq \beta$ and $\alpha = \beta$, respectively.
    (As for the case $u = v$, in each of \Cref{fig:connector-claim-1} and \Cref{fig:connector-claim-2}, the only difference would be that the path between $u$ and $v$ should be contracted to a single vertex, namely $u = v$, and then the interaction $uv$ reduces to the vertex $v$.)

    It remains to establish \eqref{eq:connector-1-2} and \eqref{eq:connector-2-2}.
    Suppose $uw$ is an interaction with $w \in C \setminus \textnormal{BP}_v$.
    Let $h$ be the index of $w$, i.e., $w = c_h$, so that $h \leq \alpha$ or $h \geq \beta$.
    We distinguish two possible cases.
    \begin{enumerate}[label=\textnormal{(\Alph*)}]
        \item \label{item:claim-connector-1}
        $w$ is an endvertex of the path $G[C]$, i.e., either $h = 0$ or $h = n$.
        Let $w^\prime \in C$ be the unique adjacent vertex of $w$ on $C$.
        More explicitly, $w^\prime = c_1$ if $w = c_0$, and $w^\prime = c_{n-1}$ if $w = c_n$.
        By \eqref{eq:connector-claim-2}, we obtain
        \begin{equation}\label{eq:connector-claim-4}
            \mathbbm{1}_{\{uw\}} = \mathbbm{1}_{F_{AB(w)}} - \mathbbm{1}_{F_{AB(w^\prime)}} \in \lin{\sigmab}
        \end{equation}
        if $\alpha = \beta = h$, and
        \begin{equation}\label{eq:connector-claim-5}
            \mathbbm{1}_{\{uw\}} = \mathbbm{1}_{F_{AB(w)}} \in \lin{\sigmab}
        \end{equation}
        otherwise.
        \item \label{item:claim-connector-2}
        $w$ is an interior vertex of the path $G[C]$, i.e., $h \in \{1, \dots, n-1\} \setminus \{\alpha+1, \dots, \beta-1\}$.
        If $\alpha = \beta = h$, by \eqref{eq:connector-claim-2}, we get
        \begin{equation}\label{eq:connector-claim-6}
            \mathbbm{1}_{\{uw\}} = \mathbbm{1}_{F_{A B(w)}} - \mathbbm{1}_{F_{A B(c_{h-1})}} - \mathbbm{1}_{F_{A B(c_{h+1})}} \in \lin{\sigmab}\,.
        \end{equation}
        Otherwise, let $w^\prime \in C \setminus \textnormal{BP}_v$ denote the unique vertex that is adjacent to $w$, and such that $w$ is a $vw^\prime$-cut-vertex of $G[C \cup \{v\}]$.
        In other words, $w^\prime = c_{h-1}$ if $1 \leq h \leq \alpha$, and $w^\prime = c_{h+1}$ if $\beta \leq h \leq n-1$.
        Hence, $C(v, w) \setminus C(v, w^\prime) = \{w\}$.
        Then, by \eqref{eq:connector-claim-2}, we obtain
        \begin{equation}\label{eq:connector-claim-7}
            \mathbbm{1}_{\{uw\}} = \mathbbm{1}_{F_{AB(w)}} - \mathbbm{1}_{F_{AB(w^\prime)}} \in \lin{\sigmab}\,.
        \end{equation}
        An example is depicted in \Cref{fig:connector-claim-3}.
    \end{enumerate}
    In both case~\ref{item:claim-connector-1} and case~\ref{item:claim-connector-2}, it holds that $\mathbbm{1}_{\{uw\}} \in \lin{\sigmab}$, which yields \eqref{eq:connector-1-2} if $u = v$, and \eqref{eq:connector-2-2} if $u \neq v$.
    This completes the proof of the claim.
\end{claimproof}

\begin{figure}
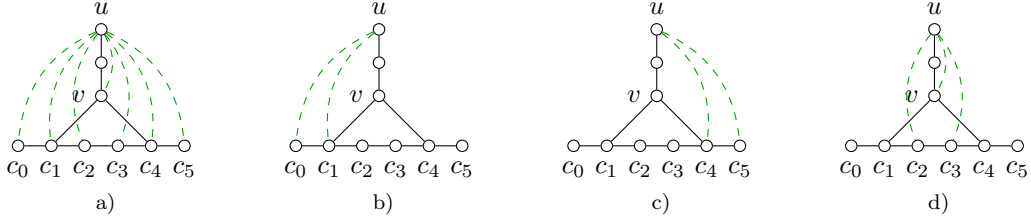

    \centering
    \begin{subfigure}[b]{0.2\textwidth}
		\centering
		\input{figures/connector/connector-claim-1-1.tex}
		\caption{}
		\label{fig:connector-claim-1-1}
	\end{subfigure}
    \hspace{3ex}
    \begin{subfigure}[b]{0.2\textwidth}
		\centering
		\input{figures/connector/connector-claim-1-2.tex}
		\caption{}
		\label{fig:connector-claim-1-2}
	\end{subfigure}
    \hspace{3ex}
    \begin{subfigure}[b]{0.2\textwidth}
		\centering
		\input{figures/connector/connector-claim-1-3.tex}
		\caption{}
		\label{fig:connector-claim-1-3}
	\end{subfigure}
    \hspace{3ex}
    \begin{subfigure}[b]{0.2\textwidth}
		\centering
		\input{figures/connector/connector-claim-1-4.tex}
		\caption{}
		\label{fig:connector-claim-1-4}
	\end{subfigure}
    \caption{This figure illustrates the construction of the characteristic vector in \eqref{eq:connector-claim-3} for $\alpha \neq \beta$. 
    Consider the graph $G=(V,E)$ depicted above in solid black and the $f$-connector $C=\{c_0,\dots,c_5\}$ for $f=c_0c_5$.
    Here it holds that $\alpha = 1$ and $\beta = 4$.
    Suppose that the set of interactions is $F=\binom{V}{2}$.
    Depicted in \subref{fig:connector-claim-1-1},~\subref{fig:connector-claim-1-2}, and~\subref{fig:connector-claim-1-3} in dashed green are the characteristic vectors $\mathbbm{1}_{F_{AB}}$, $\mathbbm{1}_{F_{AB(c_\alpha)}}$, and $\mathbbm{1}_{F_{AB(c_\beta)}}$, respectively.
    Depicted in \subref{fig:connector-claim-1-4} in dashed green is the characteristic vector of the interaction set $\{uw \in F \mid w \in \textnormal{BP}_v\}$ that is constructed in \eqref{eq:connector-claim-3}.}
    \label{fig:connector-claim-1}
\end{figure}

\begin{figure}
    \centering
    \begin{subfigure}[b]{0.1\textwidth}
		\centering
		\input{figures/connector/connector-claim-2-1.tex}
		\caption{}
		\label{fig:connector-claim-2-1}
	\end{subfigure}
    \hspace{3ex}
    \begin{subfigure}[b]{0.1\textwidth}
		\centering
		\input{figures/connector/connector-claim-2-2.tex}
		\caption{}
		\label{fig:connector-claim-2-2}
	\end{subfigure}
    \hspace{3ex}
    \begin{subfigure}[b]{0.1\textwidth}
		\centering
		\input{figures/connector/connector-claim-2-3.tex}
		\caption{}
		\label{fig:connector-claim-2-3}
	\end{subfigure}
    \caption{This figure illustrates the construction of the characteristic vector in \eqref{eq:connector-claim-3} for $\alpha = \beta$. 
    Consider the graph $G=(V,E)$ depicted above in solid black and the $f$-connector $C=\{c_0,c_1,c_2\}$ for $f=c_0c_2$.
    Here it holds that $\alpha = \beta = 1$.
    Suppose that the set of interactions is $F=\binom{V}{2}$.
    Depicted in \subref{fig:connector-claim-2-1} and \subref{fig:connector-claim-2-2} in dashed green are the characteristic vectors $\mathbbm{1}_{F_{AB}}$ and $\mathbbm{1}_{F_{AB(c_\alpha)}}$, respectively.
    Depicted in \subref{fig:connector-claim-2-3} in dashed green is the characteristic vector of the interaction set $\{uw \in F \mid w \in \textnormal{BP}_v\}$ that is constructed in \eqref{eq:connector-claim-3}.
    Notice that we have $\textnormal{BP}_v = \{v\}$ as $v$ is not a $C$-bypass.}
    \label{fig:connector-claim-2}
\end{figure}

\begin{figure}
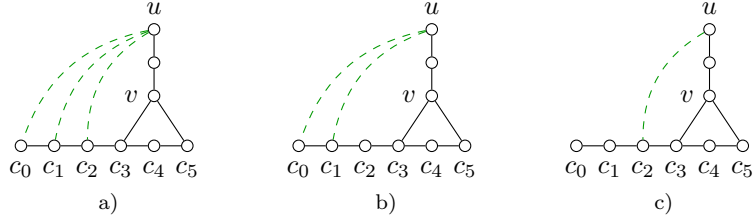

    \centering
    \begin{subfigure}[b]{0.2\textwidth}
		\centering
		\input{figures/connector/connector-claim-3-1.tex}
		\caption{}
		\label{fig:connector-claim-3-1}
	\end{subfigure}
    \hspace{3ex}
    \begin{subfigure}[b]{0.2\textwidth}
		\centering
		\input{figures/connector/connector-claim-3-2.tex}
		\caption{}
		\label{fig:connector-claim-3-2}
	\end{subfigure}
    \hspace{3ex}
    \begin{subfigure}[b]{0.2\textwidth}
		\centering
		\input{figures/connector/connector-claim-3-3.tex}
		\caption{}
		\label{fig:connector-claim-3-3}
	\end{subfigure}
    \caption{This figure illustrates the construction of the characteristic vector in \eqref{eq:connector-claim-7} of case~\ref{item:claim-connector-2}.
    Consider the graph depicted above in solid black and the $f$-connector $C=\{c_0,\dots,c_5\}$ for $f=c_0c_5$.
    Let $w = c_2 \notin \textnormal{BP}_v$.
    Then the unique vertex $w'$ defined in the proof of \Cref{claim:connector} is $w'=c_1$.
    Suppose that the set of interactions is $F=\binom{V}{2}$.
    Depicted in \subref{fig:connector-claim-3-1} and~\subref{fig:connector-claim-3-2} in dashed green are the characteristic vectors $\mathbbm{1}_{F_{AB(w)}}$ and $\mathbbm{1}_{F_{AB(w^{\prime})}}$, respectively.
    Depicted in \subref{fig:connector-claim-3-3} in dashed green is the characteristic vector $\mathbbm{1}_{\{uw\}}$ that is constructed in \eqref{eq:connector-claim-7}.
    }
    \label{fig:connector-claim-3}
\end{figure}

We are now ready to delve into the construction of a basis of $\lin{\sigmab}$ under conditions~\ref{item:connector-fd-1}~--~\ref{item:connector-fd-5} of \Cref{thm:connector-fd}.

First of all, by condition~\ref{item:connector-fd-1}, $C$ is a minimal $f$-connector of $G$.
So \Cref{claim:connector} applies.

Let $v \in V$.
If $v$ is not a vertex in the closed neighborhood of $C$, let $A = C$ and $B = C \cup \{v\}$, then $Q_{AB} = \{v\}$.
Thus, by \eqref{eq:ie-1} of \Cref{lemma:ie}, we get
\[
    \mathbbm{1}_{\{v\}} \in \lin{\sigmab}\,,
\]
since both $C$ and $C \cup \{v\}$ are in $\Sigmab$.
If $v$ is a vertex in $C$, we set $A = C$ and $B = C \setminus \{v\}$.
Then $Q_{AB}=\{v,f\}$.
Indeed, \(A\triangle B=\{v\}\), and by condition~\ref{item:connector-fd-2}, the only interaction in \(F\) contained in \(C\) whose connectivity is destroyed by deleting \(v\) is \(f\).
Therefore, it follows from \eqref{eq:ie-1} of \Cref{lemma:ie} that
\[
    \mathbbm{1}_{\{v, f\}} \in \lin{\sigmab}\,,
\]
since both $C$ and $C \setminus \{v\}$ are in $\Sigmab$.
For any vertex $v$ in the open neighborhood of $C$, condition~\ref{item:connector-fd-3} implies $\{vw \in F \mid w \in \textnormal{BP}_v\} = \emptyset$.
Then, by \eqref{eq:connector-1-1} of \Cref{claim:connector}, we get
\begin{equation*}
    \mathbbm{1}_{\{v\}} \in \lin{\sigmab}\,.
\end{equation*}
So, for every vertex $v \in V \setminus C$, its corresponding standard unit vector is in $\lin{\sigmab}$.

Now, we are left with constructing standard unit vectors associated with interactions that are different from $f$ under the given conditions in the theorem.
To achieve this, we partition the interactions $f^\prime = uw$ in $F \setminus \{f\}$ into the following classes according to the relative position between $C$ and their endvertices:
\begin{enumerate}[label=\textnormal{(\Alph*)}]
    \item \label{item:connector-sufficiency-1}
    $u, w \in C$.
    \item \label{item:connector-sufficiency-2}
    $u \in N_G(C)$ and $w \in C$.
    \item \label{item:connector-sufficiency-3}
    $u, w \in N_G(C)$.
    \item \label{item:connector-sufficiency-4}
    $u, w \in V \setminus N_G[C]$.
    \item \label{item:connector-sufficiency-5}
    $u \not\in N_G[C]$, $w \in N_G(C)$ and $w$ is not a $\{\{u\}, C\}$-cut-vertex of $G$.
    \item \label{item:connector-sufficiency-6}
    $u \not\in N_G[C]$, $w \in N_G(C)$, and $w$ is a $\{\{u\}, C\}$-cut-vertex of $G$.
    \item \label{item:connector-sufficiency-7}
    $u \not\in N_G[C]$, $w \in C$.
\end{enumerate}
Note that these classes are exhaustive and disjoint, i.e., every interaction in $F \setminus \{f\}$ lies in exactly one of these classes.

By condition~\ref{item:connector-fd-2}, there are no interactions in $F \setminus \{f\}$ that satisfy \ref{item:connector-sufficiency-1}.

Let $f^\prime = uw$ satisfy \ref{item:connector-sufficiency-2}.
By condition~\ref{item:connector-fd-3}, $u$ cannot be a $w$-bypass of $C$, i.e.~$w \in C \setminus \textnormal{BP}_u$.
Then, it follows from \eqref{eq:connector-1-2} of \Cref{claim:connector} that
\[
    \mathbbm{1}_{\{f^\prime\}} \in \lin{\sigmab}\,.
\]

Let $f^\prime = uw$ be an interaction that lies in \ref{item:connector-sufficiency-3}, or~\ref{item:connector-sufficiency-4}, or~\ref{item:connector-sufficiency-5}.
We consider two possible cases for $f^\prime$.
If there exists a minimal $f^\prime$-connector $C^{\prime}$ of $G$ that is disjoint from $N_G[C]$, then we set
\[
    A = (C \cup C^{\prime}) \setminus \{u\}\,, \qquad B = (C \cup C^{\prime}) \setminus \{w\}\,.
\]
Otherwise, we can choose a minimal $\{\{u\}, N_G[C]\}$-connector $C_1$ of $G$ and a minimal $\{\{w\}, N_G[C]\}$-connector $C_2$ of $G$ such that $w \not\in C_1$ and $u \not\in C_2$.
Set
\[
    A = (C \cup C_1 \cup C_2) \setminus \{u\}\,, \qquad B = (C \cup C_1 \cup C_2) \setminus \{w\}\,.
\]
In each case, we see that
\[
    F_{AB} = \{f^{\prime\prime} \in F \mid f^{\prime\prime} \subseteq A \cup B, \ u \ \textnormal{and} \ w \ \textnormal{are $f^{\prime\prime}$-cut-vertices of} \ G[A \cup B]\} = \{f^\prime\}\,,
\]
and
\[
    H_{AB} = \{f^{\prime\prime} \in F \mid f^{\prime\prime} \subseteq A \cap B, \ \{u, w\} \ \textnormal{is a minimal $f^{\prime\prime}$-separator of} \ G[A \cup B]\} = \emptyset.
\]
Then, by \eqref{eq:ie-2} of \Cref{lemma:ie}, we obtain
\[
    \mathbbm{1}_{\{f^\prime\}} \in \lin{\sigmab}\,,
\]
since $A$, $B$, $A \cap B$ and $A \cup B$ are in $\Sigmab$.

Let $f^\prime = uw$ be in \ref{item:connector-sufficiency-6}.
Looking at condition~\ref{item:connector-fd-4}, we see that $\{uw^\prime \in F \mid w^\prime \in \textnormal{BP}_w\} = \{uw\}$.
Then, by \eqref{eq:connector-2-1} of \Cref{claim:connector}, it holds that
\[
    \mathbbm{1}_{\{uw\}} \in \lin{\sigmab}\,.
\]

It remains to consider interactions that belong to \ref{item:connector-sufficiency-7}.
Let $f^\prime = uw$ be such an interaction, i.e., $u \not\in N_G[C]$ and $w \in C$, then we consider two possible cases:
\begin{enumerate}[label=\textnormal{(G.\arabic*)}]
    \item \label{item:connector-claim-7-1}
    There exists a vertex $v$ in the intersection of $N_G(C)$ and a minimal $\{\{u\}, N_G[C]\}$-connector of $G$ such that $w \in C \setminus \textnormal{BP}_v$.
    Then, from \eqref{eq:connector-2-2} of \Cref{claim:connector}, it follows that
    \[
        \mathbbm{1}_{\{f^\prime\}} \in \lin{\sigmab}\,.
    \]
    \item \label{item:connector-claim-7-2}
    For every $v \in N_G(C)$ that is on some minimal $\{\{u\}, N_G[C]\}$-connector of $G$, we have $w \in C \cap \textnormal{BP}_v$.
    Let us fix such a vertex $v$.
    By \eqref{eq:connector-2-1}, we know that
    \[
        \mathbbm{1}_{\{uw^\prime \in F \mid w^\prime \in \textnormal{BP}_v\}} \in \lin{\sigmab}\,.
    \]
    Furthermore, in the following we argue that $\mathbbm{1}_{\{uw^\prime\}} \in \lin{\sigmab}$ holds for all $uw^\prime \in F$ with $w^\prime \in \textnormal{BP}_v \setminus \{w\}$.
    Taking into account these two facts, we obtain
    \[
        \mathbbm{1}_{\{f^\prime\}} = \mathbbm{1}_{\{uw^{\prime} \in F \mid w^{\prime} \in \textnormal{BP}_v\}} - \sum_{uw^\prime \in F \colon w^\prime \in \textnormal{BP}_v \setminus \{w\}} \mathbbm{1}_{\{uw^\prime\}} \in \lin{\sigmab}\,.
    \]
    In order to prove that $\mathbbm{1}_{\{uw^\prime\}} \in \lin{\sigmab}$ holds for all $uw^\prime \in F$ with $w^\prime \in \textnormal{BP}_v \setminus \{w\}$, by case~\ref{item:connector-claim-7-1} and the previous discussion concerning \ref{item:connector-sufficiency-5}, it suffices to show the following:
    For any other interaction $uw^\prime \in F$ with $w^\prime \in \textnormal{BP}_v \setminus \{w\}$, there exists a vertex $v^\prime$ in the intersection of $N_G(C)$ and a minimal $\{\{u\}, N_G[C]\}$-connector of $G$ such that $w^\prime \not\in \textnormal{BP}_{v^\prime}$.
    Suppose to the contrary that there exists an interaction $uw^\prime \in F$ with $w^\prime \in \textnormal{BP}_v \setminus \{w\}$ such that for any vertex $v^\prime$ in the intersection of $N_G(C)$ and a minimal $\{\{u\}, N_G[C]\}$-connector of $G$ it holds that $w^\prime \in \textnormal{BP}_{v^\prime}$.
    If $w^\prime \in N_G(C)$, then $v$ coincides with $w^\prime$ and it is a $\{\{u\}, C\}$-cut-vertex of $G$.
    This is a contradiction to condition~\ref{item:connector-fd-4} and hence such an interaction $uw^\prime$ cannot exist.
    If $w^\prime \in C$, then the set of all neighbors of $C$ that are on minimal $\{\{u\}, N_G[C]\}$-connectors of $G$, which is clearly a $\{\{u\}, C\}$-separator of $G$, is a subset of
    \[
        \{v \in N_G(C) \mid v \textnormal{ is both a} \ w \text{-bypass and a} \ w^\prime \textnormal{-bypass of} \ C\}\,.
    \]
    This leads to a contradiction to condition~\ref{item:connector-fd-5}.
\end{enumerate}
Combining case~\ref{item:connector-claim-7-1} and case~\ref{item:connector-claim-7-2}, we see that $\mathbbm{1}_{\{f^\prime\}} \in \lin{\sigmab}$ holds for every interaction $f^\prime$ in \ref{item:connector-sufficiency-7}.

Looking at all the results obtained above, we arrive at the statement that the linear space $\lin{\sigmab}$ has a basis
\begin{equation*}
    \begin{aligned}
        {} & \{\mathbbm{1}_{\{v, f\}} \mid v \in C\}\\
        {}\cup{} & \{\mathbbm{1}_{\{v\}} \mid v \in V \setminus C\}\\
        {}\cup{} & \{\mathbbm{1}_{\{f^\prime\}} \mid f^\prime \in F \setminus \{f\}\}\,.
    \end{aligned}
\end{equation*}
Since the polytope $\Xi_{GF}$ is full-dimensional by \Cref{prop:dimension}, we conclude that the connector inequality~\eqref{eq:connector-fd} indeed defines a facet of $\polytope$ under the given conditions.
\qed

\subsection{Proof of Theorem~\ref{thm:separator-fd}}
\label{section:proof:thm:separator-fd}

In this section, we give a proof of \Cref{thm:separator-fd}, i.e., showing that conditions~\ref{item:separator-fd-1}~--\ref{item:separator-fd-4} are sufficient and necessary for separator inequality~\eqref{eq:separator-fd} to be facet-defining.
To determine the dimension of the corresponding face, we shall develop some general properties of the feasible vectors in this face.
More specifically, let
\begin{equation*}
	\sigmab =
	\left\{
		x \in \feasible \ \middle|\
		1 - x_f = \sum_{s \in S} (1 - x_s)
	\right\}
\end{equation*}
denote the set of all feasible vectors of the multi-separator problem~\eqref{eq:ms} that satisfy separator inequality~\eqref{eq:separator-fd} associated with $f$ and $S$ at equality.
Since $x^\emptyset=\mathbbm{1}_{V \cup F} \in \sigmab$, the linear space $\lin{\sigmab}$ associated with the affine hull of $\sigmab$ is
\[
    \lin{\sigmab} =
    \vspan{
        \{
            \mathbbm{1}_{V \cup F} - x \mid
            x \in \sigmab
        \}
    }
    \,.
\]
So inequality~\eqref{eq:separator} is facet-defining if and only if $\lin{\sigmab}$ has codimension~$1$ in $\mathbb{R}^{V \cup F}$.
Following the strategy described in \Cref{section:pseudoforest}, we first analyze $\lin{\sigmab}$ by decomposing it into several nicely behaved pieces, so that the structure of the feasible vectors belonging to each piece can be easily understood.
By combining structural information from different pieces, we are able to obtain a spanning set of $\lin{\sigmab}$.
Then an auxiliary loop graph is constructed from these spanning vectors so that our desired characterization follows as a consequence of \Cref{lemma:pseudoforest} in \Cref{section:pseudoforest}.
For this purpose, we begin with the introduction of a few useful notions and the development of their properties, only after which are we ready to show the sufficiency and necessity of the stated conditions.
As the first step, let us note that, for any $s \in S$, if we denote by $\sigmab_s$ the set of all feasible vectors $x \in \sigmab$ satisfying $x^{-1}(0) \cap S \subseteq \{s\}$, and by $\sigmab_{*}$ the set of all feasible vectors $x \in \sigmab$ satisfying $x^{-1}(0) \cap S = \emptyset$, then $\sigmab$ admits the decomposition $\sigmab = \cup_{s \in S \cup \{*\}} \sigmab_s$.
This in turn implies
\begin{equation}\label{eq:separator-decomposition-0}
    \lin{\sigmab} =
    \sum_{s \in S \cup \{*\}}
    \lin{\sigmab_s}
    \,.
\end{equation}

Before proceeding further, let us consider vertices not in $S$ and interactions not in $F(S)$.
Because their corresponding standard unit vectors, which will later be completed to a basis of $\lin{\sigmab}$, can be easily constructed as follows.
For any vertex $v \in V\setminus S$, we have $Q_{\emptyset \{v\}} = \{v\}$.
Since both $x^{\emptyset}$ and $x^{\{v\}}$ belong to $\sigmab_*$, the corresponding standard unit vector $\mathbbm{1}_{\{v\}}$ is in $\lin{\sigmab_*}$ by \eqref{eq:ie-1} of \Cref{lemma:ie}.
For any interaction $f^\prime = uw \in F$ that is not in $F(S)$, there exists a minimal $f^\prime$-connector $C$ of $G$ such that $C \cap S = \emptyset$.
Set $A = C \setminus \{w\}$ and $B = C \setminus \{u\}$, then all $x^A$,~$x^B$,~$x^{A \cap B}$ and~$x^{A \cup B}$ belong to $\sigmab_*$.
Furthermore, $F_{AB} = \{f^\prime\}$ and $H_{AB} = \emptyset$.
Again, by \eqref{eq:ie-2} of \Cref{lemma:ie}, the standard unit vector $\mathbbm{1}_{\{f^\prime\}}$ is in $\lin{\sigmab_*}$.
Thus, we have constructed all the standard unit vectors of vertices in $V \setminus S$ and interactions in $F \setminus F(S)$.
Since the summand $\lin{\sigmab_{*}}$ in \eqref{eq:separator-decomposition-0} is a subspace of $\mathbb{R}^{V \setminus S} \oplus \mathbb{R}^{F \setminus F(S)}$, the result we just obtained can be summarized as
\begin{equation}\label{eq:separator-decomposition-1}
    \lin{\sigmab} =
    \left(
        \mathbb{R}^{V \setminus S} \oplus
        \mathbb{R}^{F \setminus F(S)}
    \right)
    +
    \sum_{s \in S} \lin{\sigmab_s}
    \,.
\end{equation}

This in particular indicates that in our further development we only need to concentrate on those elements in $S \cup F(S)$.
On the other hand, constructing the remaining basis vectors is not straightforward.
To achieve this, we utilize the properties of $\sigmab$ developed (till \Cref{claim:separator-equivalence}), to further refine the decomposition~\eqref{eq:separator-decomposition-1}.

To begin our development of properties of $\sigmab$, let us assume that condition~\ref{item:separator-fd-1} of \Cref{thm:separator-fd} holds, i.e., $S$ is a minimal $f$-separator of $G$, so that all the notions introduced in \Cref{def:separator} are valid.
Note that this assumption is justified, because, by \Cref{lemma:ilp}, we already know that the minimality of $S$ is a necessary condition for the associated separator inequality to be facet-defining.
For notational simplicity, given any $s \in S$, we write $C(s)$ for the set of all $f$-cut-vertices of $G(S, s)$.

In the following we partition interactions that are connected in $G(S, s)$ into several classes for each fixed $s \in S$.
Such a partition is motivated by a simple fact about $\sigmab$ as explained below.
For any fixed $s \in S$ and any $x \in \sigmab_s \setminus \sigmab_*$, it holds that $x_f = 0$,~$x_v = 0$ for all $v \in C(s)$ and that $x_{s^\prime} = 1$ for all $s^\prime \in S \setminus \{s\}$.
Since every two vertices of $C(s)$ are connected in $G[V \cap x^{-1}(0)]$, it follows that $x_{uv} = x_{uv'}$ for all $uv, uv' \in F(S)$ with $v, v' \in C(s)$;
that is, two such related interactions $uv$,~$uv'$ are always assigned the same value by all the feasible vectors $x \in \sigmab_s$.
This suggests that, if we restrict ourselves to $\sigmab_s$, then it is natural to regard any two adjacent interactions as equivalent which are connected in $G(S, s)$ and whose non-common endvertices belong to $C(s)$.
More importantly, it is desirable to have the property that the characteristic vectors of equivalence classes belong to $\lin{\sigmab}$, from which we can proceed further to construct basis vectors of $\lin{\sigmab}$.
The above consideration motivates us to introduce the following definition, which plays a fundamental role in our approach.

With the given minimal $f$-separator $S$ of $G$, we define the set $[f^\prime]_s \subseteq S \cup F(S)$ for any $f^\prime = uw \in F(S)$ and any $s \in S$ according to the following four possible cases:
\begin{enumerate}[label=\textnormal{(\roman*)}]
    \item If $s \not\in S(f^\prime)$, then $[f^\prime]_s = \emptyset$.
    \item If $s \in S(f^\prime)$ and $u, w \not\in C(s)$, then $[f^\prime]_s = \{f^\prime\}$.
    \item If $s \in S(f^\prime)$ and there is precisely one endvertex, say $w$, of $f^\prime$ that is in $C(s)$, then $[f^\prime]_s = \{u w^{\prime} \in F(S) \mid w^{\prime} \in C(s)\}$.
    \item If $s \in S(f^\prime)$ and $u, w \in C(s)$, then $[f^\prime]_s = \{s\} \cup \{f^{\prime\prime} \in F(S) \mid f^{\prime\prime} \subseteq C(s)\}$.
\end{enumerate}

Note that for each fixed $s \in S$, the collection of all the sets $[f^\prime]_s$ for $f^\prime \in F(S)$ with $s \in S(f^\prime)$ gives a partition of the set consisting of $s$ and all interactions $f^\prime \in F(S)$ with $s \in S(f^\prime)$.
For example, in \Cref{fig:separator-not-facet}~(e), the vertex $s_1$ induces the partition $\{\{s_1, f, f_1\}, \{f_2, f_3\}\}$ and $s_2$ induces the partition $\{\{s_2, f, f_3\}, \{f_1,f_2\}\}$; in \Cref{fig:separator-not-facet}~(f), the vertex $s_1$ induces the partition $\{\{s_1, f\}, \{f_1, f_2\}, \{f_3,f_4\}\}$, and $s_2$ induces the partition $\{\{s_2, f\}, \{f_1,f_4\}, \{f_2, f_3\}\}$.

The following result, which is employed to demonstrate both the sufficiency and the necessity of \Cref{thm:separator-fd} afterward, reinforces our original idea to introduce $[f^\prime]_s$.

\begin{claim}\label{claim:separator-equivalence}
    Let $f^\prime \in F(S)$ and $s \in S(f^\prime)$, then for any $x \in \sigmab_s$ it holds that
    \begin{equation}\label{eq:separator-equivalence-1}
        \forall g \in [f^\prime]_s \colon x_{g} = x_{f^\prime}\,.
    \end{equation}
    Moreover, we have
    \begin{equation}\label{eq:separator-equivalence-2}
        \mathbbm{1}_{[f^\prime]_s} \in \lin{\sigmab_s}\,.
    \end{equation}
\end{claim}

\begin{claimproof}
To show the first statement, i.e.~\eqref{eq:separator-equivalence-1}, let $x \in \sigmab_s$ and let $g$ be an arbitrary element in $[f^\prime]_s$.
As $[f^\prime]_s \subseteq S \cup F(S)$, the statement is clear when $x \in \sigmab_*$.
If $x \in \sigmab_s \setminus \sigmab_*$, then $x_s = 0$, $x_f = 0$, and $x_{s^\prime} = 1$ for all $s^\prime \in S \setminus \{s\}$.
So there is a minimal $f$-connector $C$ of $G(S, s)$ with $x|_C = 0$.
Suppose $x_{f^\prime} = 0$.
Then we can find a minimal $f^\prime$-connector $C^\prime$ of $G(S, s)$ with $x|_{C^\prime} = 0$.
From the definition of $[f^\prime]_s$, we see that $g$ is connected in $G[C \cup C^{\prime}]$, and hence $x_{g} = 0$.
As $f^\prime \in [g]_s$, a similar argument yields that $x_g = 0$ implies $x_{f^\prime} = 0$.
Therefore the desired equality $x_{f^\prime} = x_{g}$ follows.

For the second statement, i.e.~\eqref{eq:separator-equivalence-2}, we prove it by applying \Cref{lemma:ie} to construct the characteristic vector $\mathbbm{1}_{[f^\prime]}$ as a linear combination of vectors in $\lin{\sigmab_s}$.
Let us start with introducing a few vertex subsets that will be used in the subsequent construction.
Write $f=pq$ and $f^\prime=uw$.
Let $C_1$ (respectively, $C_2$) be a minimal $ps$-connector (respectively, $qs$-connector) of $G(S, s)$.
If $u$ (respectively, $w$) is not in $C(s)$, then it is clear that we can choose $C_1$,~$C_2$ such that $u \not\in C_1 \cup C_2$ (respectively, $w \not\in C_1 \cup C_2$).
Furthermore, if both $u$ and $w$ are not in $C(s)$, then $C_1$ and $C_2$ can be chosen such that $u,w \not\in C_1 \cup C_2$.
Note that $C_1 \cup C_2$ is a minimal $f$-connector of $G$.
In any case, as $s \in S(f^\prime)$, we can pick a minimal $\{\{u\}, N_G[C_1 \cup C_2]\}$-connector $C^\prime_1$ of $G(S, s)$ and a minimal $\{\{w\}, N_G[C_1 \cup C_2]\}$-connector $C^\prime_2$ of $G(S, s)$, respectively.
Write $U = C_1 \cup C_2 \cup C^\prime_1 \cup C^\prime_2$.
Examples of the sets constructed here can be found in \Cref{fig:separator-claim-equivalence}.
Below we show that $\mathbbm{1}_{[f^\prime]_s} \in \lin{\sigmab_s}$ by distinguishing three possible cases as in the definition of $[f^\prime]_s$.
\begin{enumerate}[label=\textnormal{(\Alph*)}]
    \item If both $u$ and $w$ are not $f$-cut-vertices of $G(S, s)$, then we set $A = U \setminus \{u\}$ and $B = U \setminus \{w\}$.
    Thus, all $x^A$,~$x^B$,~$x^{A \cap B}$ and~$x^{A \cup B}$ belong to $\sigmab_s$.
    Furthermore, $F_{AB} = \{f^\prime\}$ and $H_{AB} = \emptyset$.
    By \eqref{eq:ie-2} of \Cref{lemma:ie}, we get $\mathbbm{1}_{[f^\prime]_s} = \mathbbm{1}_{\{f^\prime\}} \in \lin{\sigmab_s}$.
    \item If $u$ is not an $f$-cut-vertex of $G(S, s)$ but $w$ is an $f$-cut-vertex of $G(S, s)$, then we set $A = U \setminus \{u\}$ and $B = U \setminus \{s\}$.
    It follows by definition that $F_{AB} = \{uw^\prime \in F \mid w^\prime \in C_1 \cup C_2 \ \textnormal{and} \ s \ \textnormal{is a} \ uw^\prime\textnormal{-cut-vertex of} \ G[U]\}$ and $H_{AB} = \emptyset$.
    Thus, \eqref{eq:ie-2} of \Cref{lemma:ie} immediately gives $\mathbbm{1}_{F_{AB}} \in \lin{\sigmab_s}$, upon noting that $x^A$,~$x^B$,~$x^{A \cap B}$ and~$x^{A \cup B}$ belong to $\sigmab_s$.
    On the other hand, we have the identity
    \[
        \mathbbm{1}_{[f^\prime]_s} = \mathbbm{1}_{F_{AB}} - \sum_{uw^\prime \in F_{AB} \setminus F(S)} \mathbbm{1}_{\{uw^\prime\}} - \sum_{uw^\prime \in F_{AB} \cap F(S) \colon w^\prime \not\in C(s)} \mathbbm{1}_{\{uw^\prime\}}\,.
    \]
    Therefore, $\mathbbm{1}_{[f^{\prime}]_s} \in \lin{\sigmab_s}$, since $\mathbbm{1}_{F_{AB}} \in \lin{\sigmab_s}$ as we just saw, and $\mathbbm{1}_{\{uw^\prime\}} \in \lin{\sigmab_s}$ for every interaction $uw^\prime \not\in F(S)$ by the discussion preceding this claim, and $\mathbbm{1}_{\{uw^\prime\}} \in \lin{\sigmab_s}$ for every interaction $uw^\prime \in F(S)$ with $w^\prime \in (C_1 \cup C_2) \setminus C(s)$ by the previous case.
    See \Cref{fig:separator-claim-equivalence}~\subref{fig:claim-equivalence-a} for an example.
    Similarly, if $u$ is an $f$-cut-vertex of $G(S, s)$ but $w$ is not an $f$-cut-vertex of $G(S, s)$, then we can also obtain $\mathbbm{1}_{[f^\prime]_s} \in \lin{\sigmab_s}$.
    \item Now, suppose both $u$ and $w$ are $f$-cut-vertices of $G(S, s)$.
    Set $A = U \setminus \{s\}$ and $B = U$.
    Then, $x^A, x^B \in \sigmab_s$ and $Q_{AB} = \{s\} \cup \{u^\prime w^\prime \in F(S) \mid u^\prime \in C_1 \ \textnormal{and} \ w^\prime \in C_2\}$.
    By \eqref{eq:ie-1} of \Cref{lemma:ie}, we get $\mathbbm{1}_{Q_{AB}} \in \lin{\sigmab_s}$.
    Moreover, it holds that
    \[
        \mathbbm{1}_{[f^\prime]_s} = \mathbbm{1}_{Q_{AB}} - \sum_{[u^\prime w^\prime]_s \colon u^\prime \in C_1 \setminus C(s) \ \textnormal{or} \ w^\prime \in C_2 \setminus C(s)} \mathbbm{1}_{[u^\prime w^\prime]_s}\,.
    \]
    Notice that here the sum is taken over all distinct sets $[u^\prime w^\prime]_s$ with $u^\prime \in C_1 \setminus C(s)$ or $w^\prime \in C_2 \setminus C(s)$.
    It is clear that $x^A$,~$x^B$,~$x^{A \cap B}$ and~$x^{A \cup B}$ are all in $\sigmab_s$.
    By the results above, we have already seen that $\mathbbm{1}_{[u^\prime w^\prime]_s} \in \lin{\sigmab_s}$ for every interaction $u^\prime w^\prime \in F(S)$ with $u^\prime \in C_1 \setminus C(s)$ or $w^\prime \in C_2 \setminus C(s)$.
    Thus once again $\mathbbm{1}_{[f^\prime]_s} \in \lin{\sigmab_s}$.
    See \Cref{fig:separator-claim-equivalence}~\subref{fig:claim-equivalence-b} for an example.
\end{enumerate}
The above discussion has established \eqref{eq:separator-equivalence-2} in all possible cases, and hence completes the proof.

\begin{figure}
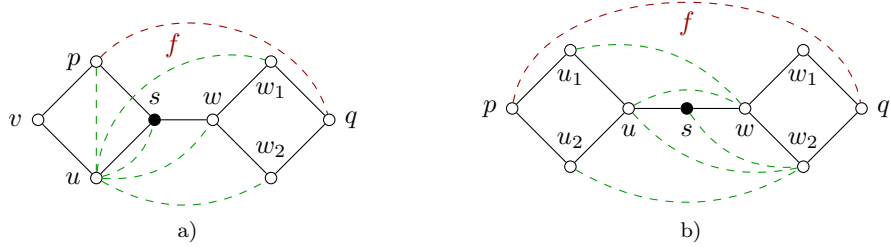

    \centering
    \begin{subfigure}[b]{0.4\textwidth}
        \centering
        \input{figures/separator/claim-equivalence-a.tex}
        \caption{}
        \label{fig:claim-equivalence-a}
    \end{subfigure}
	\hspace{3ex}
    \begin{subfigure}[b]{0.4\textwidth}
        \centering
        \input{figures/separator/claim-equivalence-b.tex}
        \caption{}
        \label{fig:claim-equivalence-b}
    \end{subfigure}
    \caption{Depicted above are examples illustrating the constructions in the proof of \Cref{claim:separator-equivalence}.
    Each figure consists of a graph $G$ (solid black), a designated interaction $f=pq$ (dashed red), a minimal $f$-separator $S = \{s\}$ of $G$, and other interactions (dashed green).
    In \subref{fig:claim-equivalence-a}, if $C_1 = \{p, s\}$, $C_2 = \{s, w, w_1, q\}$, $C^{\prime}_1 = \{u\}$ and $C^{\prime}_2 = \{w\}$, then $U = \{u, p, s, w, w_1, q\}$, $A = U \setminus \{u\}$ and $B = U \setminus \{s\}$.
    Thus, $F_{AB} = \{up, us, uw, uw_1\}$ and $[uw]_s = \{us, uw\}$.
    In \subref{fig:claim-equivalence-b}, if $C_1 = \{p, u_1, u, s\}$, $C_2 = \{s, w, w_2, q\}$, $C^{\prime}_1 = \{u\}$ and $C^{\prime}_2 = \{w\}$, then $U = \{p, u_1, u, s, w, w_2, q\}$, $A = U \setminus \{s\}$ and $B = U$.
    Thus, $Q_{AB} = \{uw, u_1w, uw_2, sw_2, s, f\}$, $[uw]_s = \{s, f, uw\}$, $[u_1w]_s = \{u_1w\}$, and $[uw_2]_s = \{uw_2, sw_2\}$.}
    \label{fig:separator-claim-equivalence}
\end{figure}
\end{claimproof}

As an immediate consequence of \Cref{claim:separator-equivalence}, we have
\[
    \lin{\sigmab_s} =
    \mathbb{R}^{V \setminus S} \oplus
    \mathbb{R}^{F \setminus F(S)} \oplus
    \sum_{f^\prime \in F(S)} \vspan{\mathbbm{1}_{[f^\prime]_s}}
\]
for every $s \in S$.
Combining with \eqref{eq:separator-decomposition-1} yields
\begin{equation}\label{eq:separator-decomposition-2}
    \lin{\sigmab} =
    \mathbb{R}^{V \setminus S} \oplus
    \mathbb{R}^{F \setminus F(S)} \oplus
    \sum_{s \in S, f^\prime \in F(S)} \vspan{\mathbbm{1}_{[f^\prime]_s}}
    \,.
\end{equation}

Let us take a closer look at the last direct summand of the above decomposition.
Firstly, we observe that if $f^\prime \in F(S)$ satisfies $f^\prime \cap C(s) = \emptyset$ for some $s \in S(f^\prime)$, then $\mathbbm{1}_{\{f^\prime\}} \in \lin{\sigmab_s}$ by \eqref{eq:separator-equivalence-2} of \Cref{claim:separator-equivalence}, since $[f^\prime]_s = \{f^\prime\}$ in this case.
Secondly, for each vertex $s \in S$ there is exactly one set, i.e., $[f]_s$, that contains $s$ among all $[f^\prime]_{s^\prime}$ with $s^\prime \in S$ and $f^\prime \in F(S)$.
Similarly, each $sv \in F(S)$ with $s \in S$ and $v \in G(S, s) \setminus C(s)$ is contained in $[sv]_s$ only.
Thirdly, for every $uw \in F(S)$ such that $u$ is non-degenerate, $u \not\in C(s)$ for some $s \in S(uw)$, and $w \in C(s)$ for any $s \in S(uw)$, there is one and only one set, i.e., $[uw]_s$ for some $s \in S(uw)$ with $us \not\in F$, that contains $uw$ among all $[f^\prime]_{s^\prime}$ with $s^\prime \in S$ and $f^\prime \in F(S)$.
Taking these observations into account, the last term of \eqref{eq:separator-decomposition-2} can be rewritten as
\begin{align}
    \sum_{s \in S, f^\prime \in F(S)} \vspan{\mathbbm{1}_{[f^\prime]_s}} \cong \mathbb{R}^{S}
    &\oplus \mathbb{R}^{\{f^\prime \in F(S) \mid \exists s \in S(f^\prime)\colon f^\prime \cap C(s) = \emptyset\}}
    \oplus \mathbb{R}^{\{sv \in F(S) \mid s \in S,\, v \in G(S, s) \setminus C(s)\}}\nonumber\\
    &\oplus \sum_{s \in S, f^\prime \in F_{\parallel}} \vspan{\mathbbm{1}_{[f^\prime]_s}}\oplus \sum_{s \in S, f^\prime \in F(S, s)} \vspan{\mathbbm{1}_{[f^\prime]_s}}\label{eq:separator-decomposition-2-span}
    \,,
\end{align}
where
\[
    F_{\parallel} =
    \left\{
        uw \in F(S) \;\middle|\;
        \begin{aligned}
            & S \in \{\{u\}, f\}\textnormal{-separators}(G),\\
            & u \ \textnormal{is non-degenerate, and}\\
            & \forall s \in S(uw) \colon w \in C(s)
        \end{aligned}
    \right\}
\]
and
\[
    F(S, s) =
	\left\{
		f^{\prime} \in F(S) \setminus \tbinom{C(s)}{2} \;\middle|\;
		\begin{aligned}
            & s \in S(f^{\prime}), \forall f^{\prime\prime} \in [f^{\prime}]_s \colon s \not\in f^{\prime\prime},\\
			& \forall s^\prime \in S(f^{\prime}) \colon f^{\prime} \cap C(s^\prime) \neq \emptyset,\\
            & \forall v \in f^\prime \colon S \not\in \{\{v\}, f\}\textnormal{-separators}(G)
		\end{aligned}
	\right\}
\]
for every $s \in S$.
Note that if $f^\prime \in F(S, s)$ for some $s \in S$, then neither of the endvertices of $f^\prime$ is separated from $f$ by $S$ in $G$.

Combining \eqref{eq:separator-decomposition-2} and \eqref{eq:separator-decomposition-2-span}, we obtain
\begin{equation}\label{eq:separator-decomposition-3}
    \lin{\sigmab} \cong
    \mathbb{R}^{V}
    \oplus
    \mathbb{R}^{
        (
            F \setminus F_{\varnothing}
        )
        \setminus
        (
            \cup_{s \in S} F(S, s) \cup F_{\parallel}
        )
        }
    \oplus
    \sum_{s \in S, f^\prime \in F_{\parallel}} \vspan{\mathbbm{1}_{[f^\prime]_s}}
    \oplus
    \sum_{s \in S, f^\prime \in F(S, s)} \vspan{\mathbbm{1}_{[f^\prime]_s}}
    \,,
\end{equation}
where
\[
    F_{\varnothing} =
    \left\{
        f^\prime \in F(S) \;\middle|\;
        \begin{aligned}
            & \forall s \in S(f^\prime) \colon \ \textnormal{either} \ f^\prime \subseteq C(s),\ \textnormal{or}\\
            & f^\prime \setminus (C(s) \setminus S) = \{v\} \ \textnormal{and} \ vs \in F
        \end{aligned}
    \right\}
    \,.
\]

\begin{figure}
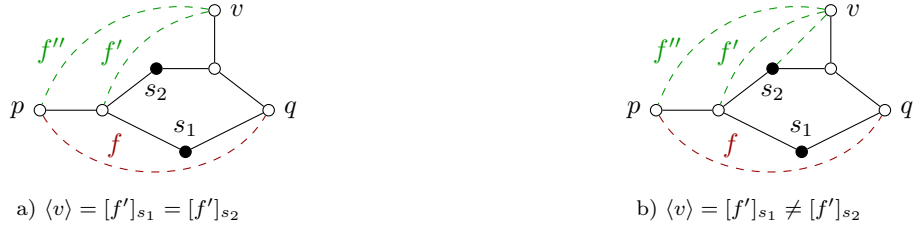

    \centering
    \begin{subfigure}[b]{0.45\textwidth}
        \centering
        \input{figures/separator/equivalence-a.tex}
        \caption{$\langle v \rangle = [f^{\prime}]_{s_1} = [f^{\prime}]_{s_2}$}
        \label{fig:equivalence-a}
    \end{subfigure}
    \hfill
    \begin{subfigure}[b]{0.45\textwidth}
        \centering
        \input{figures/separator/equivalence-b.tex}
        \caption{$\langle v \rangle = [f^{\prime}]_{s_1} \neq [f^{\prime}]_{s_2}$}
        \label{fig:equivalence-b}
    \end{subfigure}
    \caption{Depicted above are examples illustrating the notion $\langle v \rangle$.
    Each figure consists of a graph $G$ (solid black), a designated interaction $f=pq$ (dashed red), a minimal $f$-separator $S = \{s_1, s_2\}$ of $G$, and other interactions (dashed green).
    In \subref{fig:equivalence-a}, $[f^{\prime}]_{s_1} = \{f^\prime, f^{\prime\prime}\} = [f^{\prime}]_{s_2}$, hence $\langle v \rangle = [f^{\prime}]_{s_1} = [f^{\prime}]_{s_2}$.
    As a comparison, we have $[f^\prime]_{s_1} = \{f^\prime, f^{\prime\prime}\}$ and $[f^\prime]_{s_2} = \{f^\prime, f^{\prime\prime}, vs_2\}$ in \subref{fig:equivalence-b}, hence $\langle v \rangle = [f^\prime]_{s_1} \neq [f^\prime]_{s_2}$.
    }
    \label{fig:separator-equivalence}
\end{figure}

Now it is clear that in order to determine the dimension of $\lin{\sigmab}$, we just need to understand the last direct summand of decomposition~\eqref{eq:separator-decomposition-3}.
To this end, in the next step we construct an auxiliary loop graph so that its (vertex-edge) incidence matrix has full column rank if and only if so does the last direct summand of \eqref{eq:separator-decomposition-3}.
Then we can utilize \Cref{lemma:pseudoforest} to decide when $\lin{\sigmab}$ has codimension~$1$.
More specifically, we require that, for all vertices of the auxiliary loop graph that are not incident with loops, the corresponding characteristic vectors of incident edges are given exactly by $\mathbbm{1}_{[f^\prime]_s}$ with $s \in S$ and $f^\prime \in F(S, s)$.
Notice that we exclude vertices incident with loops in the above requirement because in the last direct summand of \eqref{eq:separator-decomposition-3} it is possible that some $f^\prime$ is contained in only one equivalence class $[f^\prime]_s$.
The exact construction will be given below after making one more vital observation.

To achieve our goal of interpreting the last direct summand of \eqref{eq:separator-decomposition-3} as the incidence matrix of some loop graph, let us observe that for any $v \in V \setminus S$ that is not an $f$-cut-vertex of $G$ and any $f^\prime \in \cup_{s \in S} F(S, s)$ with $v \in f^\prime$, it holds that $[f^\prime]_{s_1} = [f^\prime]_{s_2}$ for all $s_1, s_2 \in S(f^\prime)$ satisfying $v \not \in C(s_1) \cup C(s_2)$ and $vs_1, vs_2 \not\in F$.
See \Cref{fig:separator-equivalence} for illustration.
This observation motivates us to introduce $\langle v \rangle$ to represent the common set $[f^\prime]_{s_1} = [f^\prime]_{s_2}$ and hence the set of edges incident with $v$.
More formally, for any $v \in V$, we define $\langle v \rangle$ as the set of all those interactions $f^\prime \in F(S)$ with $v \in f^\prime$ that satisfy the following two conditions:
\begin{enumerate}[label=\textnormal{(\alph*)}]
    \item $f^\prime \cap C(s) \neq \emptyset$ for any $s \in S(f^\prime)$.
    \item There exists $s \in S(f^\prime)$ such that $v \not\in C(s)$ and $vs \not\in F$.
\end{enumerate}
Notice that if $v \in V$ is degenerate, then $\langle v \rangle$ is empty.
Moreover, the union of all $\langle v \rangle$ with $v$ not separated from $f$ by $S$ in $G$ is exactly $\cup_{s \in S} F(S, s)$.
We construct an auxiliary loop graph $G_{\infty} = (V_{\infty}, E_{\infty})$ as follows:
\begin{itemize} 
    \item $V_{\infty}$ is the set of all vertices of $G$ that are not separated from $f$ by $S$ in $G$.
    \item $E_{\infty}$ consists of all loops $\{v\}$ incident with degenerate vertices $v$ in $V_{\infty}$ and all interactions in $\cup_{v \in V_{\infty}} \langle v \rangle$.
\end{itemize}
Note that the loopless part of $G_{\infty}$ is bipartite, whose two parts are given by the two vertex subsets $\{v \in V \setminus S \mid S \not\in uv\textnormal{-separators}(G)\}$ for $u \in f$.
This means that $G_{\infty}$ has no odd cycles except loops.
It follows from the definition of $\langle v \rangle$ that
\[
    \sum_{s \in S, f^\prime \in F_{\parallel}}
    \vspan{
        \mathbbm{1}_{[f^\prime]_s}
    }
    =
    \bigoplus_{\langle v \rangle \colon v \in V \setminus V_{\infty}}
    \vspan{
        \mathbbm{1}_{\langle v \rangle}
    }
\]
and
\begin{equation*}
    \sum_{s \in S, f^\prime \in F(S, s)} \vspan{\mathbbm{1}_{[f^\prime]_s}}
    =
    \sum_{v \in V_{\infty}} \vspan{\mathbbm{1}_{\langle v \rangle}}
    \,.
\end{equation*}
Putting it all together, we have
\begin{equation}\label{eq:separator-decomposition-4}
    \lin{\sigmab} \cong
    \mathbb{R}^{V}
    \oplus
    \mathbb{R}^{
        (
            F \setminus F_{\varnothing}
        )
        \setminus
        \cup_{v \in V} \langle v \rangle
    }
    \oplus
    \bigoplus_{
        \langle v \rangle \colon
        v \in V \setminus V_{\infty}
    }
    \vspan{
        \mathbbm{1}_{\langle v \rangle}
    }
    \oplus
    \sum_{v \in V_{\infty}} \vspan{\mathbbm{1}_{\langle v \rangle}}
    \,.
\end{equation}
We further observe that the row space of the incidence matrix of $G_{\infty}$ is
\[
    \bigoplus_{\{v \in V_{\infty} \colon \{v\} \in E_{\infty}\}} \vspan{\mathbbm{1}_{\{e \in E_{\infty} \mid v \in e\}}}
    \oplus
    \sum_{v \in V_{\infty}} \vspan{\mathbbm{1}_{\langle v \rangle}}
    \,.
\]

Taking the preceding observation into account and looking at the decomposition~\eqref{eq:separator-decomposition-4}, we see that $\lin{\sigmab}$ has codimension~$1$ in $\mathbb{R}^{V \cup F}$ if and only if
\begin{enumerate}[label=\textnormal{(\roman*$'$)}]
    \setcounter{enumi}{1}
    \item \label{item:separator-fd-2'}
    $F_{\varnothing} = \{f\}$.
    \item \label{item:separator-fd-3'}
    $\lvert \langle v \rangle \rvert \leq 1$ for every $v \in V$ such that $S$ is a $\{\{v\}, f\}$-separator of $G$.
    \item \label{item:separator-fd-4'}
    The incidence matrix of $G_{\infty}$ has full column rank.
\end{enumerate}
Moreover, by \Cref{lemma:pseudoforest} and the fact that the loopless part of $G_{\infty}$ is bipartite, condition~\ref{item:separator-fd-4'} is equivalent to the statement that $G_{\infty}$ is a pseudoforest without cycles of length~$\geq 2$.

Finally, we claim that, assuming $S$ is a minimal $f$-separator of $G$, conditions~\ref{item:separator-fd-2},~\ref{item:separator-fd-3} and~\ref{item:separator-fd-4} in \Cref{thm:separator-fd} are equivalent to conditions~\ref{item:separator-fd-2'},~\ref{item:separator-fd-3'} and~\ref{item:separator-fd-4'} above.
This is easy to see:
conditions~\ref{item:separator-fd-2} and~~\ref{item:separator-fd-3} together are equivalent to conditions~\ref{item:separator-fd-2'} and~\ref{item:separator-fd-3'};
under conditions~\ref{item:separator-fd-1}~--~\ref{item:separator-fd-3}, condition~\ref{item:separator-fd-4} is equivalent to condition~\ref{item:separator-fd-4'}.
This concludes the proof of \Cref{thm:separator-fd}.
\qed

\subsection{Proof of Lemma~\ref{lemma:separator-polynomial-time}}
\label{section:proof:lemma:separator-polynomial-time}

Condition~\ref{item:separator-fd-1} can be verified by checking whether $S \setminus \{s\}$ is still an $f$-separator of $G$ for every $s \in S$, which can be done clearly in polynomial time.

In view of \Cref{remark:seprator-def-P}, we see that the decision problems for the remaining conditions of \Cref{thm:separator-fd} can be solved in polynomial time.

\subsection{Proof of Lemma~\ref{lemma:separator-three-disjoint-paths}}
\label{section:proof:lemma:separator-three-disjoint-paths}

Necessity is clear, since both $p$ and $v$ are degenerate.

To see sufficiency, we first note that $\HH(f, S)$ is acyclic due to the existence of three such internally vertex-disjoint $f$-paths in $G$, since otherwise any cycle of $\HH(f, S)$, whose vertices must not be separated from $f$ by $S$ in $G$ by condition~\ref{item:separator-fd-3}, has an edge $f'$ whose two endvertices are not $f$-cut-vertices of $G(S, s)$ for some $s \in S(f') = S$, contradicting the definition of $\HH(f, S)$.
Let $P = (V_P, E_P)$ be a path in $\HH(f, S)$ with $E_P \neq \emptyset$ whose only degenerate vertices are its two endvertices.
For the desired result, we just need to show that $P$ has only one edge, which shares one common vertex with $f$ and whose other endvertex is not separated from $f$ by $S$ in $G$ and is connected to every vertex of $S$ by an interaction in $F$.
Indeed, it follows from condition~\ref{item:separator-fd-3} that none of the vertices of $P$ is separated from $f$ by $S$ in $G$.
Then, by the definition of $\HH(f, S)$, for any edge $e$ of $P$ and any $f$-path $P^{\prime} = (V_{P'}, E_{P'})$ in $G$ that intersects $S$ in one vertex, we have $e \cap V_{P^{\prime}} \neq \emptyset$.
However, as it is assumed that there exist at least three such $f$-paths that are internally vertex-disjoint, this implies that $e \cap f \neq \emptyset$, i.e., $e$ shares a common vertex with $f$.
Because the endvertices of $f$ are degenerate and the internal vertices of $P$ are assumed to be non-degenerate, $P$ has exactly one edge $e$.
Condition~\ref{item:separator-fd-2} requires that $e$ neither intersects $S$ nor is a pair of $f$-cut-vertices of $G$, and hence forces the other degenerate endvertex of $e$, the one that is not incident with $f$, to be connected to every vertex in the separator $S$ by an interaction.

\subsection{Proof of Lemma~\ref{lemma:path-valid}}
\label{section:proof:lemma:path-valid}

Let $P = (V_P, E_P)$.
By multiplying \eqref{eq:path} by a factor $2$ and rearranging, it can be written as
\begin{equation}\label{eq:path-valid}
    (2 x_f - x_p - x_q) \leq \sum_{uw \in E_P} (2 x_{uw} - x_u - x_w)\,,
\end{equation}
where $p$ and $q$ are the two endvertices of $f$.
Note that, for any $x \in \feasible$, both $2 x_f - x_p - x_q$ and $2 x_{uw} - x_u - x_w$ for every $uw \in E_P$ can only take values $0, 1, 2$.
To prove validity, we verify each possible case below.

If $2 x_f - x_p - x_q = 0$, then \eqref{eq:path-valid} is obviously satisfied.
If $2 x_f - x_p - x_q = 1$, then it can only happen that $x_f = 1$ and $x_p \neq x_q$.
In this case, the path $P$ contains an edge $uw \in F$ whose endvertices $u$ and $w$ take different values, hence $2 x_{uw} - x_u - x_w = 1$ and \eqref{eq:path-valid} holds.
If $2 x_f - x_p - x_q = 2$, then the only possibility is that $x_f = 1$ and $x_p = x_q = 0$.
In this case, the edges of $P$ cannot all take the value $0$ by the definition of $\feasible$.
From this we see that the path $P$ either possesses an edge $uw \in F$ with $x_{uw} = 1$,~$x_u = x_w = 0$, or at least two distinct edges $u_1 w_1, u_2 w_2 \in F$ with $x_{u_1 w_1} = x_{u_2 w_2} = 1$,~$x_{u_1} \neq x_{w_1}$,~$x_{u_2} \neq x_{w_2}$, which again implies that \eqref{eq:path-valid} holds.
Therefore, we have shown that inequality~\eqref{eq:path-valid} is satisfied by all $x \in \feasible$, and hence \eqref{eq:path} is valid.

\subsection{Proof of Theorem~\ref{thm:path-facet}}
\label{section:proof:thm:path-facet}

For convenience, we fix a linear order on $P$, i.e. let
\[
    V_P = \{v_0, v_1, \dots, v_n\}
\]
and
\[
    E_P = \{e_1 = v_0v_1, \dots, e_n = v_{n-1}v_n\}\,.
\]
Hence $f = v_0 v_n$.
We denote
\[
    \sigmab = \left\{x \in \feasible \mid x_f = \sum_{e \in E_P} x_e - \sum_{v \in V_P \setminus f} x_v\right\}
\]
as the set of all feasible vectors that satisfy \eqref{eq:path} with equality.

(Necessity)
Suppose that $P$ has a chord $f^\prime = v_i v_j$ in $F \setminus \{f\}$, where \(0\leq i<j\leq n\) and \(j\geq i+2\).
Then, we consider a second $f$-path $P^\prime = (V_{P^\prime}, E_{P^\prime})$ in $(V, F)$ obtained by replacing the part between $v_i$ and $v_j$ with $f^\prime$, that is, $V_{P^\prime} = \{v_0, \dots, v_i, v_j, \dots, v_n\}$ and $E_{P^\prime} = \{e_1, \dots, e_i, e_{j+1}, \dots, e_n\} \cup \{f^\prime\}$.
By \Cref{lemma:path-valid}, it follows that the inequality
\begin{equation}\label{eq:path-3}
    x_f \leq x_{f^\prime} + \sum_{k \in \{1, \dots, n\} \setminus \{i+1, \dots, j\}} x_{e_k} - \sum_{l \in \{1, \dots, n-1\} \setminus \{i+1, \dots, j-1\}} x_{v_l}
\end{equation}
is valid for $\polytope$.
Now we notice that \eqref{eq:path} is simply a sum of \eqref{eq:path-3} and the path inequality
\[
    x_{f^\prime} \leq \sum_{k=i+1}^j x_{e_k} - \sum_{l = i+1}^{j-1} x_{v_l}\,.
\]
Consequently, inequality \eqref{eq:path} does not define a facet of $\polytope$ under the given condition.

(Sufficiency)
Assume that the path $P$ is chordless in the graph $(V, F \setminus \{f\})$.
In other words, there is no interaction $f^\prime \in F \setminus (E_P \cup \{f\})$ such that $f^\prime \subseteq V_P$.

Looking at this inequality, we observe that, for any vertex subset $U$ whose intersection with $V_P$ is either empty or induces a path in $P$ that also contains an endvertex of $P$, the induced feasible vector $x^{U}$ satisfies this inequality tightly and hence belongs to $\sigmab$.
To see this more explicitly, as the case when either $U \cap V_P = \emptyset$ or $V_P \subseteq U$ is obvious, suppose without loss of generality that $U \cap V_P = \{v_0, \dots, v_h\}$ for some $h \in \{0, \dots, n-1\}$.
Then, $x^{U}_f = 1$, $\sum_{e \in E_P} x^{U}_e = n - h$ and $\sum_{v \in V_P \setminus f} x^{U}_v = n - h - 1$, and hence the observation follows.

In the following, we construct the associated standard unit vector for each vertex that is not in the interior of $P$ and each interaction that is not in $E_P \cup \{f\}$.

Firstly, let $v$ be a vertex that is not in the interior of $P$.
Since the all-one vector $x^{\emptyset} = \mathbbm{1}_{V \cup F}$ and $x^{\{v\}}$ belong to $\sigmab$ by the observation above, we have
\[
    \mathbbm{1}_{\{v\}} = \mathbbm{1}_{V \cup F} - x^{\{v\}} \in \lin{\sigmab}\,.
\]

Secondly, let $f^\prime = st \in F$ be an interaction such that neither of the endvertices of $f^\prime$ is in $V_P$.
If $V_P$ is not an $f^\prime$-separator of $G$, we can choose a chordless $f^\prime$-path $P^\prime$ in $G$ such that the two vertex sets of $P$ and $P^\prime$ are disjoint.
Set
\begin{equation*}
    A = V_{P^\prime} \setminus \{t\}\,, \quad B = V_{P^\prime} \setminus \{s\}\,.
\end{equation*}
Then $F_{AB} = \{f^\prime\}$ and $H_{AB} = \emptyset$.
By \eqref{eq:ie-2} of \Cref{lemma:ie}, we get
\[
    \mathbbm{1}_{\{f^\prime\}} \in \lin{\sigmab}\,,
\]
since $x^{A}$, $x^{B}$, $x^{A \cap B}$ and $x^{A \cup B}$ are in $\sigmab$ by the observation above.
Otherwise, $V_P$ is an $f^\prime$-separator of $G$.
Then, we take a chordless $\{s, V_P\}$-path $P^\prime$ and a chordless $\{t, V_P\}$-path $P^{\prime\prime}$ in $G$ and set
\[
    U = V_P \cup V_{P^\prime} \cup V_{P^{\prime\prime}}, \quad A = U \setminus \{s\}\,, \quad B = U \setminus \{t\}\,.
\]
As $V_P$ is an $f^\prime$-separator of $G$, none of $P^{\prime}$ and $P^{\prime\prime}$ contains an interior vertex of another.
Thus, $F_{AB} = \{f^\prime\}$ and $H_{AB} = \emptyset$.
From \eqref{eq:ie-2} of \Cref{lemma:ie} and the fact that $x^{A}$, $x^{B}$, $x^{A \cap B}$ and $x^{A \cup B}$ are in $\sigmab$ by the above observation, it follows that
\begin{equation*}
    \mathbbm{1}_{\{f^\prime\}} \in \lin{\sigmab}\,.
\end{equation*}
An example is depicted in \Cref{fig:path-fd}~\subref{fig:path-fd-1}.

\begin{figure}
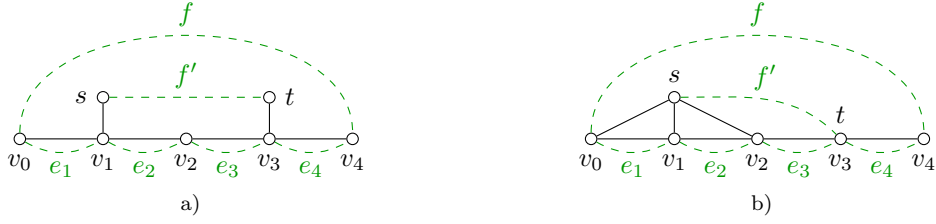

	\centering
    \begin{subfigure}[b]{0.45\textwidth}
        \centering
        \input{figures/path-ineq-proof-1.tex}
        \caption{}
        \label{fig:path-fd-1}
    \end{subfigure}
    \hspace{4ex}
    \begin{subfigure}[b]{0.45\textwidth}
        \centering
        \input{figures/path-ineq-proof-2.tex}
        \caption{}
        \label{fig:path-fd-2}
    \end{subfigure}
	\caption{Examples of the constructions in the proof of \Cref{thm:path-facet}.
    Each subfigure consists of a graph $G$ (solid black) , a path $P$ on the vertex set $\{v_0, v_1, v_2, v_3, v_4\}$, and interactions (dashed green).
    In \subref{fig:path-fd-1}, the paths induced by the edges $sv_1$ and $tv_3$ can be chosen as $P^\prime$ and $P^{\prime\prime}$ respectively; in \subref{fig:path-fd-2}, the path induced by any of the edges $sv_0$, $sv_1$, $sv_3$ can be chosen as $P^\prime$.}
    \label{fig:path-fd}
\end{figure}

Thirdly, we will consider all interactions with one endvertex in $V_P$.
To that end, let $s \in V \setminus V_P$ be an arbitrary but fixed vertex and let $P^\prime$ be a chordless $\{s, V_P\}$-path in $G$.
Let $u$ be the unique vertex in $V_{P^\prime} \cap N_G(V_P)$ and let $n^{-}$ (respectively, $n^{+}$) denote the minimal (respectively, maximal) index $i$ such that $v_i$ is adjacent to $u$ in $G$.
For any $f^\prime = sv_m \in F$ with $m < n^+$, we set
\[
    U = (V_{P^\prime} \setminus V_P) \cup \{v_l \in V_P \mid l \geq m\}, \quad A = U \setminus \{s\}\,, \quad B = U \setminus \{v_m\}\,.
\]
It can be checked that $F_{AB} = \{f^\prime\}$ and $H_{AB} = \emptyset$.
By \eqref{eq:ie-2} of \Cref{lemma:ie}, we obtain
\begin{equation*}
    \mathbbm{1}_{\{f^\prime\}} \in \lin{\sigmab}\,,
\end{equation*}
since $x^{A}$, $x^{B}$, $x^{A \cap B}$ and $x^{A \cup B}$ are in $\sigmab$ by the observation above.
By a symmetric argument, it can be obtained that the associated standard unit vector $\mathbbm{1}_{\{f^\prime\}}$ also lies in $\lin{\sigmab}$ for the case $f^\prime = sv_m \in F$ with $m > n^{-}$.
An example is shown in \Cref{fig:path-fd}~\subref{fig:path-fd-2}.
Now, for the remaining case $f^\prime = sv_m \in F$ with $m = n^{-} = n^{+}$, we set
\[
    A = (V_P \cup V_{P^\prime}) \setminus \{s\}\,, \quad B = V_P \cup V_{P^\prime}\,.
\]
Then, $Q_{AB} = \{st^\prime \in F \mid t^\prime \in A\} \cup \{s\}$.
From the above observation, $x^{A}$, $x^{B}$, $x^{A \cap B}$ and $x^{A \cup B}$ are all in $\sigmab$.
It then follows from \eqref{eq:ie-1} of \Cref{lemma:ie} that $\mathbbm{1}_{Q_{AB}} \in \lin{\sigmab}$.
Moreover, the previous discussion shows that the standard unit vector associated with every element in $Q_{AB}$ lies in $\lin{\sigmab}$ except for $f^\prime$.
Therefore, it holds that
\[
    \mathbbm{1}_{\{f^\prime\}} = \mathbbm{1}_{Q_{AB}} - \sum_{g \in Q_{AB} \setminus \{f^\prime\}} \mathbbm{1}_{\{g\}} \in \lin{\sigmab}\,.
\]
Since the vertex $s$ has been chosen arbitrarily we have $\mathbbm{1}_{\{f'\}} \in \lin{\sigmab}$ for all $f' \in F$ with exactly one endvertex in $V_P$.

By the previous discussion and our assumption that $P$ is chordless in the graph $(V, F \setminus \{f\})$, we immediately see that the standard unit vector associated with every interaction that is not in $E_P \cup \{f\}$ is contained in $\lin{\sigmab}$.

Next, for every interior vertex $v$ of $P$ and every edge $e$ of $P$, we show that both the vectors $\mathbbm{1}_{\{e\}} + \mathbbm{1}_{\{f\}}$ and $\mathbbm{1}_{\{v\}} - \mathbbm{1}_{\{f\}}$ are elements of $\lin{\sigmab}$.
In view of the preceding paragraph, for our purpose, it is sufficient to consider the special case $F = E_P \cup \{f\}$.
Let us first deal with an edge $e$ of $P$.
Suppose $e = e_{k}$ for some $k \in \{1, \dots, n\}$.
Note that both $x^{\{v_l \in V_P \mid l < k\}}$ and $x^{\{v_l \in V_P \mid l \geq k\}}$ are members of $\sigmab$ by the observation above.
Then, we have
\[
    \mathbbm{1}_{\{e_k\}} + \mathbbm{1}_{\{f\}} = x^{\{v_l \in V_P \mid l < k\}} + x^{\{v_l \in V_P \mid l \geq k\}} - \mathbbm{1}_{V \cup F} - \mathbbm{1}_{V \setminus V_P}\,.
\]
It lies in $\lin{\sigmab}$ since both $\mathbbm{1}_{V \cup F}$ and $\mathbbm{1}_{V \setminus V_P}$ belong to $\sigmab$ by the observation above.
Next, let $v = v_m$ be an interior vertex of $P$ for some $m \in \{1, \dots, n-1\}$.
As the vectors $\mathbbm{1}_{\{e\}} + \mathbbm{1}_{\{f\}}$ are known to be in $\lin{\sigmab}$ for all $e \in E_P$, we can iteratively calculate $\mathbbm{1}_{\{v\}} - \mathbbm{1}_{\{f\}}$ from $m=1$ to $m=n-1$ by
\[
    \mathbbm{1}_{\{v_m\}} - \mathbbm{1}_{\{f\}} = \mathbbm{1}_{V \cup F} - x^{\{v_l \in V_P \mid l \leq m\}} - \mathbbm{1}_{\{v_0\}} - \sum_{1 \leq l < m} (\mathbbm{1}_{\{v_l\}} - \mathbbm{1}_{\{f\}}) - \sum_{1 \leq l \leq m} (\mathbbm{1}_{\{e_l\}} + \mathbbm{1}_{\{f\}})\,,
\]
which also lies in $\lin{\sigmab}$ for every $m \in \{1, \dots, n-1\}$.

Thus, we have shown that the linear space $\lin{\sigmab}$ has a basis
\begin{equation*}
    \begin{aligned}
        {} & \{\mathbbm{1}_{\{v\}} \mid v \in (V \setminus V_P) \cup f\}\\
        {}\cup{} &\{\mathbbm{1}_{\{f^\prime\}} \mid f^\prime \in F \setminus (E_P \cup \{f\})\}\\
        {}\cup{} &\{\mathbbm{1}_{\{v\}} - \mathbbm{1}_{\{f\}} \mid v \in V_P \setminus f\}\\
        {}\cup{} &\{\mathbbm{1}_{\{e\}} + \mathbbm{1}_{\{f\}} \mid e \in E_P\}\,,
    \end{aligned}
\end{equation*}
which implies that it has codimension $1$, and hence so does the face $\aff(\sigmab)$.
Since the polytope $\polytope$ is full-dimensional by \Cref{prop:dimension}, we conclude that the path inequality \eqref{eq:path} is facet-defining for $\polytope$ under the specified assumption.
\qed

\subsection{Proof of Lemma~\ref{lemma:intersection-valid}}
\label{section:proof:lemma:intersection-valid}

\emph{Necessity}. Assume that the intersection inequality is valid for $\polytope$.
Without loss of generality, we may assume that $\{v_2, w_2\}$ is not a $\{v_1, w_1\}$-separator of $G$.
Then, there exists a vertex subset $U$ such that $v_1, w_1 \in U$ and $v_2, w_2 \not\in U$.
Thus, we have $x^U_{v_1w_2} = 1$ and $x^U_{v_2w_1} = 1$.
Since $x^U_{v_1w_1} = 0$, the vector $x^U$ violates the intersection inequality, which contradicts our assumption.

\emph{Sufficiency}. Assume that $\{v_1, w_1\}$ is a $\{v_2, w_2\}$-separator of $G$ and $\{v_2, w_2\}$ is a $\{v_1, w_1\}$-separator of $G$.
Let $x \in \feasible$ be an arbitrary feasible vector.
We need to show that the intersection inequality is satisfied by $x$.
Let us first consider the following two cases.
\begin{enumerate}
    \item Case~(a): $x_{v_1w_1} = x_{v_2w_2} = 0$. This implies the existence of a vertex subset $U$ such that $x|_U = 0$ and $U$ is both a $\{v_1, w_1\}$- and a $\{v_2, w_2\}$-connector of $G$.
    Since $\{w_1, w_2\} \in E$, the set $U$ is also a $\{v_1, w_2\}$- and a $\{v_2, w_1\}$-connector of $G$.
    Consequently, $x_{v_1w_2} = 0$ and $x_{v_2w_1} = 0$, implying that \(x\) satisfies the intersection inequality.
    \item Case~(b): $x_{v_1w_1} = 0$ and $x_{v_2w_2} = 1$. In this case, there exists a vertex subset \(U\) such that \(x|_U = 0\) and \(U\) is a \(\{v_1, w_1\}\)-connector of \(G\).
    Since both $\{v_1, v_2\}$ and $\{v_2, w_2\}$ are in $E$, the set $U \cup \{v_2, w_2\}$ is also a $\{v_2, w_2\}$-connector of $G$.
    Thus, if $x_{v_2} = x_{w_2} = 0$, then $x_{v_2w_2} = 0$, contradicting our assumption.
    Hence, at least one of $x_{v_2}$ and $x_{w_2}$ takes the value $1$.
    Note that, $x_{v_2}$ and $x_{w_2}$ cannot both be $1$. Otherwise, since $\{v_2, w_2\}$ is a $\{v_1, w_1\}$-separator of $G$, we would have $x_{v_1w_1} = 1$, contradicting our assumption.
    Therefore, exactly one of $x_{v_2}$ and $x_{w_2}$ takes the value $1$, which implies that either $x_{v_1w_2} = 0$ or $x_{v_2w_1} = 0$, and hence the intersection inequality is satisfied by $x$.
\end{enumerate}
Note that the case $x_{v_1w_1} = 1$ and $x_{v_2w_2} = 0$ can be handled by a symmetric argument, and the case $x_{v_1w_1} = x_{v_2w_2} = 1$ is trivial.
Thus, all possible cases have been considered, and we have shown that the intersection inequality is valid for $\polytope$ under the given condition.
\qed

\subsection{Proof of Theorem~\ref{theorem:intersection-facet}}
\label{section:proof:theorem:intersection-facet}

\emph{Necessity}. This follows immediately from \Cref{lemma:intersection-valid}.

\emph{Sufficiency}. Assume that $\{v_1, w_1\}$ is a $\{v_2, w_2\}$-separator of $G$ and $\{v_2, w_2\}$ is a $\{v_1, w_1\}$-separator of $G$.
Under these assumptions, we construct a basis of the linear space $\lin{\sigmab}$, where $\sigmab$ is the set of all feasible vectors that satisfy the intersection inequality with equality.

Let us consider an arbitrary vertex $v$.
Set \(A=\emptyset\) and \(B=\{v\}\).
Then \(x^A,x^B,x^{A\cap B},x^{A\cup B}\in \sigmab\), because no singleton vertex set connects any of the four interactions appearing in the intersection inequality.
Moreover, $Q_{AB}=Q_{\emptyset,\{v\}}=\{v\}$.
Thus, by \eqref{eq:ie-1} of \Cref{lemma:ie}, we have $\mathbbm{1}_{\{v\}}\in \lin{\sigmab}$.

We first show that the following two conditions cannot both hold:
\begin{enumerate}[label=\textnormal{(\alph*)}]
	\item $\{v_1, w_2\}$ is an $\{v_2, w_1\}$-separator of $G$.
	\item $\{v_2, w_1\}$ is an $\{v_1, w_2\}$-separator of $G$.
\end{enumerate}
Indeed, if both conditions hold, our initial separator assumptions together with the fact that \(\{v_1, v_2\}, \allowbreak \{w_1, w_2\} \in E\) imply that any minimal $\{v_1, w_2\}$-connector $C_1$ of $G$ must contain $\{v_2, w_1\}$, and symmetrically, any minimal $\{v_2, w_1\}$-connector $C_2$ of $G$ must contain $\{v_1, w_2\}$.
Consequently, the vertex subset \(C_1 \setminus \{v_1, w_2\}\) is a minimal \(\{v_2, w_1\}\)-connector of $G$ disjoint from \(\{v_1, w_2\}\), and \(C_2 \setminus \{v_2, w_1\}\) is a minimal \(\{v_1, w_2\}\)-connector of $G$ disjoint from \(\{v_2, w_1\}\).
This yields a contradiction.

By symmetry, we may assume henceforth that $\{v_1, w_2\}$ is not a $\{v_2, w_1\}$-separator of $G$.
Let $C^\prime$ denote the union of $\{v_1, w_2\}$ and a minimal $\{v_2, w_1\}$-connector of $G$ disjoint from $\{v_1, w_2\}$.
By our initial assumption and the fact that \(\{v_1, v_2\}, \allowbreak \{w_1, w_2\} \in E\), the induced subgraph \(G[C^\prime]\) is either a chordless path, if \(v_1w_2\notin E\), or a chordless cycle, if \(v_1w_2\in E\).

\begin{figure}
	\tikzstyle{interaction1}=[draw=mycolor1]
	\tikzstyle{interaction2}=[draw=mycolor2]
	\tikzstyle{interaction3}=[draw=mycolor3]
	(1)
    \input{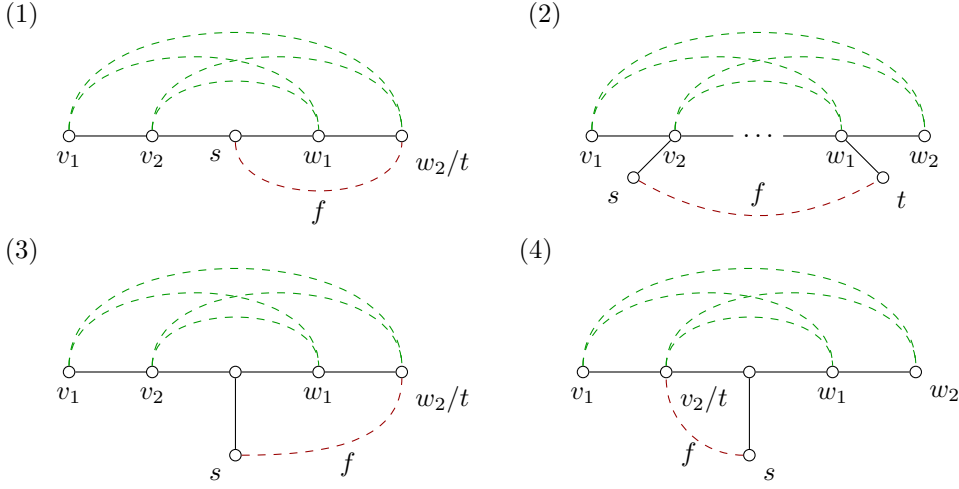}
	\hspace{2ex}
	(2)
	\begin{tikzpicture}[baseline={([yshift={-\ht\strutbox}]current bounding box.north)}]
		\node[vertex, label=below:$v_1$] at (0, 0) (v1) {};
		\node[vertex, label=below:$v_2$] at (1, 0) (v2) {};
		\node[vertex, draw=white] at (2, 0) (u) {$\ \ldots\ $};
		\node[vertex, label=below:$w_1$] at (3, 0) (w1) {};
		\node[vertex, label=below:$w_2$] at (4, 0) (w2) {};
		\node[vertex, label=below left:$s$] at (0.5, -0.5) (s) {};
		\node[vertex, label=below right:$t$] at (3.5, -0.5) (t) {};
		\draw (v1) -- (v2);
		\draw (v2) -- (u);
		\draw (u) -- (w1);
		\draw (w1) -- (w2);
		\draw (s) -- (v2);
		\draw (t) -- (w1);
		\path[interaction1][dashed] (v1) to[out=90, in=90] node[midway, above] {} (w2);
		\path[interaction1][dashed] (v2) to[out=90, in=90] node[midway, below] {} (w1);
		\path[interaction1][dashed] (v1) to[out=90, in=90] node[near start, below] {} (w1);
		\path[interaction1][dashed] (v2) to[out=90, in=90] node[near end, below] {} (w2);
		\path[interaction2][dashed] (s) to[out=-30, in=210] node[midway, above] {$f$} (t);
	\end{tikzpicture}

	(3)
	\begin{tikzpicture}[baseline={([yshift={-\ht\strutbox}]current bounding box.north)}]
		\node[vertex, label=below:$v_1$] at (0, 0) (v1) {};
		\node[vertex, label=below:$v_2$] at (1, 0) (v2) {};
		\node[vertex] at (2, 0) (r) {};
		\node[vertex, label=below:$w_1$] at (3, 0) (w1) {};
		\node[vertex, label=below right:$w_2/t$] at (4, 0) (w2) {};
		\node[vertex, label=below left:$s$] at (2, -1) (s) {};
		\draw (v1) -- (v2);
		\draw (v2) -- (r);
		\draw (r) -- (w1);
		\draw (w1) -- (w2);
		\draw (r) -- (s);
		\path[interaction1][dashed] (v1) to[out=90, in=90] node[midway, above] {} (w2);
		\path[interaction1][dashed] (v2) to[out=90, in=90] node[midway, below] {} (w1);
		\path[interaction1][dashed] (v1) to[out=90, in=90] node[near start, below] {} (w1);
		\path[interaction1][dashed] (v2) to[out=90, in=90] node[near end, below] {} (w2);
		\path[interaction2][dashed] (s) to[out=0, in=-90] node[midway, below] {$f$} (w2);
	\end{tikzpicture}
	\hspace{2ex}
	(4)
	\begin{tikzpicture}[baseline={([yshift={-\ht\strutbox}]current bounding box.north)}]
		\node[vertex, label=below:$v_1$] at (0, 0) (v1) {};
		\node[vertex, label=below right:$v_2/t$] at (1, 0) (v2) {};
		\node[vertex] at (2, 0) (r) {};
		\node[vertex, label=below:$w_1$] at (3, 0) (w1) {};
		\node[vertex, label=below right:$w_2$] at (4, 0) (w2) {};
		\node[vertex, label=below right:$s$] at (2, -1) (s) {};
		\draw (v1) -- (v2);
		\draw (v2) -- (r);
		\draw (r) -- (w1);
		\draw (w1) -- (w2);
		\draw (r) -- (s);
		\path[interaction1][dashed] (v1) to[out=90, in=90] node[midway, above] {} (w2);
		\path[interaction1][dashed] (v2) to[out=90, in=90] node[midway, below] {} (w1);
		\path[interaction1][dashed] (v1) to[out=90, in=90] node[near start, below] {} (w1);
		\path[interaction1][dashed] (v2) to[out=90, in=90] node[near end, below] {} (w2);
		\path[interaction2][dashed] (s) to[out=180, in=-90] node[midway, below] {$f$} (v2);
	\end{tikzpicture}
	\caption{Depicted are graph (solid black) and interactions (dashed). In each subfigure, the green interactions define an intersection equality. For each red interaction, its corresponding standard unit vector is constructed in the proof of \Cref{theorem:intersection-facet}.}
	\label{fig:intersection-proof}
\end{figure}

For any interaction $f = \{s, t\} \in F$ that is not separated by $\{v_1, v_2\}$ (resp. $\{w_1, w_2\}$) in $G$, we choose a minimal $f$-connector $C$ of $G$ such that $C \cap \{v_1, v_2\} = \emptyset$ (resp. $C \cap \{w_1, w_2\} = \emptyset$).
Set $A = C \setminus \{s\}$ and $B = C \setminus \{t\}$, then $F_{AB} = \{f\}$ and $H_{AB} = \emptyset$.
Moreover, $x^A$, $x^B$, $x^{A \cap B}$ and $x^{A \cup B}$ are all in $\sigmab$.
For example, if $C \cap \{v_1, v_2\} = \emptyset$, then every set among $A$, $B$, $A \cap B$ and $A \cup B$ is also disjoint from $\{v_1, v_2\}$. Hence none of the four special interactions can be connected, because each of them has one endpoint in $\{v_1, v_2\}$. Therefore all four special interaction variables are equal to \(1\), and the intersection inequality holds at equality: $1+1=1+1$.
By \eqref{eq:ie-2} of \Cref{lemma:ie}, we have
\[
	\mathbbm{1}_{\{f\}} \in \lin{\sigmab}\,.
\]
See \Cref{fig:intersection-proof}-(1) for an example.

Let \(f = \{s, t\} \in F\) be an interaction that is separated by both \(\{v_1, v_2\}\) and \(\{w_1, w_2\}\) in \(G\), and satisfies \(f \not\subseteq \{v_{1},v_{2},w_{1},w_{2}\}\).
We distinguish two possible cases depending on whether the two sets $f$ and $\{v_1, v_2, w_1, w_2\}$ intersect.

First, suppose \(f\) does not intersect \(\{v_1, v_2, w_1, w_2\}\); then \(f\) is also disjoint from \(C^{\prime }\).
Hence, we can find a minimal $\{\{s\}, C^\prime\}$-connector $U_1$ of $G$ and a minimal $\{\{t\}, C^\prime\}$-connector $U_2$ of $G$.
Set $A = (U_1 \cup U_2 \cup C^\prime) \setminus \{s\}$ and $B = (U_1 \cup U_2 \cup C^\prime) \setminus \{t\}$.
By construction, $F_{AB} = \{f\}$ and $H_{AB} = \emptyset$.
Indeed, the set \(A\cup B\) contains an \(st\)-connector, whereas \(A\) misses \(s\) and \(B\) misses \(t\). Thus \(f\in F_{AB}\). No other interaction belongs to \(F_{AB}\), because every other interaction contained in \(A\cup B\) remains connected in at least one of \(A\) or \(B\). Moreover, deleting \(s\) and \(t\) simultaneously does not destroy any connection among endpoints lying in \(A\cap B\), so \(H_{AB}=\emptyset\).
By \eqref{eq:ie-2} of \Cref{lemma:ie}, we obtain
\[
    \mathbbm{1}_{\{f\}} \in \lin{\sigmab}\,,
\]
since $x^A$, $x^B$, $x^{A \cap B}$ and $x^{A \cup B}$ are in $\sigmab$.
See \Cref{fig:intersection-proof}-(2) for an illustration.

Next, we consider the case where \(f\) intersects \(\{v_1, v_2, w_1, w_2\}\).
Suppose that $t$ is the only element in the intersection of \(f\) and \(\{v_1, v_2, w_1, w_2\}\), and let \(U\) be a minimal \(\{s, C^\prime\}\)-connector of \(G\).
If \(t \in \{v_1, w_2\}\), we define $A = (U \cup C^\prime) \setminus \{s\}$ and $B = (U \cup C^\prime) \setminus \{t\}$.
A straightforward verification like above shows that $F_{AB} = \{f\}$ and $H_{AB} = \emptyset$.
By \eqref{eq:ie-2} of \Cref{lemma:ie}, we obtain
\[
    \mathbbm{1}_{\{f\}} \in \lin{\sigmab}\,,
\]
since $x^A$, $x^B$, $x^{A \cap B}$ and $x^{A \cup B}$ are in $\sigmab$.
See \Cref{fig:intersection-proof}-(3) for an illustration.
Symmetrically, if \(t \in \{v_2, w_1\}\) (say, \(t = v_2\)), we set $A = (U \cup C^\prime) \setminus \{s, v_1\}$ and $B = (U \cup C^\prime) \setminus \{t, v_1\}$.
It follows that $F_{AB} = \{f\}$ and $H_{AB} = \emptyset$.
Applying \eqref{eq:ie-2} of \Cref{lemma:ie} yields
\[
    \mathbbm{1}_{\{f\}} \in \lin{\sigmab}\,,
\]
since $x^A$, $x^B$, $x^{A \cap B}$, $x^{A \cup B}$ are in $\sigmab$.
This case is illustrated in \Cref{fig:intersection-proof}-(4).

Finally, we consider the remaining elements $\{v_1, w_2\}$, $\{v_2, w_1\}$, $\{v_1, w_1\}$, and $\{v_2, w_2\}$.
One can readily verify that the following equalities hold:
\[
    \mathbbm{1}_{\{v_1w_1, v_1w_2\}} = x^{C^\prime \setminus \{v_1\}} - x^{C^\prime} - \mathbbm{1}_{\{v_1\}} - \sum_{v_1v \in F \colon v \in C^\prime \setminus \{w_1, w_2\}} \mathbbm{1}_{\{v_1v\}}\,,
\]
\[
    \mathbbm{1}_{\{v_2w_2, v_1w_2\}} = x^{C^\prime \setminus \{w_2\}} - x^{C^\prime} - \mathbbm{1}_{\{w_2\}} - \sum_{w_2w \in F \colon w \in C^\prime \setminus \{v_1, v_2\}} \mathbbm{1}_{\{w_2w\}}\,,
\]
and
\[
    \mathbbm{1}_{\{v_2w_2, v_2w_1\}} = x^{C^\prime \setminus \{v_1, v_2\}} - x^{C^\prime \setminus \{v_1\}} - \mathbbm{1}_{\{v_2\}} - \sum_{v_2v \in F \colon v \in C^\prime \setminus \{w_1, w_2\}} \mathbbm{1}_{\{v_2v\}}\,.
\]
These three vectors all belong to $\lin{\sigmab}$, since the zero vector $x^{V}$ and the vectors on the right-hand sides of the equalities above are in $\sigmab$.

By all the results above, we see that the linear space $\lin{\sigmab}$ has a basis
\[
    \left\{\mathbbm{1}_{\{g\}} \mid g \in (V \cup F) \setminus \{v_1w_2, v_2w_1, v_1w_1, v_2w_2\}\right\} \cup \left\{\mathbbm{1}_{\{v_1w_1, v_1w_2\}}, \mathbbm{1}_{\{v_2w_2, v_1w_2\}}, \mathbbm{1}_{\{v_2w_2, v_2w_1\}}\right\}\,.
\]
Since $\polytope$ is full-dimensional by \Cref{prop:dimension}, we conclude that the intersection inequality~\eqref{eq:ineq-intersection} defines a facet of $\polytope$ under the given conditions.
\qed

\subsection{Proof of Lemma~\ref{lemma:bqp-msp}}
\label{section:proof:lemma:bqp-msp}

We first prove \ref{item:bqp-msp-1}. 
Let \(x \in \feasible\), and define \(z\in \{0,1\}^{V\cup F}\) by $z_v=1-x_v$ for all $v\in V$ and $z_{uv}=z_u z_v$ for all $uv\in F$.
Then \(z\) is a vertex of \(\bqp_{(V,F)}\).

For every interaction \(uv\in F\), if \(x_{uv}=0\), then \(u\) and \(v\) are connected in \(G[x^{-1}(0)\cap V]\), and in particular \(x_u=x_v=0\). Hence \(z_u=z_v=1\), so \(z_{uv}=1\). Therefore $1-x_{uv}\leq z_{uv}$ for all $uv\in F$.
Moreover, if \(uv\in E\), then \(u\) and \(v\) are connected in \(G[x^{-1}(0)\cap V]\) if and only if \(x_u=x_v=0\). 
Thus $1-x_{uv}=z_{uv}$ for all \(uv\in E\).
Since \(a_{uv}\geq 0\) for all \(uv\in F\setminus E\), we obtain $a(1-x)\leq az$.
By validity of \(az\leq \alpha\) for \(\bqp_{(V,F)}\), we have \(az\leq \alpha\). 
Hence $a(1-x)\leq\alpha$.
This proves \ref{item:bqp-msp-1}.

The proof of part~\ref{item:bqp-msp-2} is analogous.
Let \(z\) be a vertex of \(\bqp_{(V,F)}\), and define $U=\{v\in V\mid z_v=1\}$.
Let \(x=x^U \in \feasible\).
Then \(1-x_{uv}\leq z_{uv}\) for every \(uv\in F\), with equality for \(uv\in E\). 
Since \(a_{uv}\leq 0\) for all \(uv\in F\setminus E\), it follows that $az\leq a(1-x)$.
By validity of \(a(1-x)\leq \alpha\) for \(\polytope\), we obtain $az\leq \alpha$.
Thus \(az\leq \alpha\) is valid for \(\bqp_{(V,F)}\).

It remains to prove part~\ref{item:bqp-msp-3}. 
Assume that $\{f\in F\mid a_f\neq 0\} \subseteq E \subseteq F$ and that \(a(1-x)\leq \alpha\) is facet-defining for \(\polytope\).
Note that the inequality \(a|_{V \cup E}z\leq \alpha\) is valid for \(\bqp_G\).
Thus we only need to show that the face defined by \(a|_{V \cup E}z\leq \alpha\) has dimension \(|V| + |E| - 1\).

Let $\sigmab = \{x\in \feasible \mid a(1-x)=\alpha\}$ be the set of feasible multi-separator vectors in the corresponding face. 
Since \(\polytope\) is full-dimensional and \(a(1-x)\leq\alpha\) defines a facet, we have $\dim \lin{\sigmab} = |V| + |F| - 1$.
Moreover, \(a_f=0\) for every \(f\in F\setminus E\). 
Hence $\mathbbm{1}_{\{f\}} \in \lin{\sigmab}$ for $f \in F \setminus E$.

Let $Z' = \{z\in \operatorname{vert}(\bqp_G)\mid a|_{V \cup E}z=\alpha\}$.
Since \(E\subseteq F\), the affine map defined by $z_v=1-x_v$ for $v \in V$ and $z_e=1-x_e$ for $e\in E$ identifies the projection of \(\sigmab\) onto the subsapce \(\mathbb{R}^{V\cup E}\) with \(Z'\).
Therefore, $\lin{\sigmab} \cong \lin{Z'} \oplus \mathbb{R}^{F \setminus E}$, and hence
\[
    \dim \conv Z'
    = \dim \lin{Z'}
    = \dim \lin{\sigmab} - \lvert F \setminus E \rvert
    = \lvert V \rvert + \lvert E \rvert - 1
    \,,
\]
Since \(\bqp_G\) is full-dimensional in \(\mathbb R^{V\cup E}\), the inequality \(a|_{V \cup E}z\leq\alpha\) defines a facet of \(\bqp_G\).
\qed

\subsection{Proof of Proposition~\ref{prop:lmc-to-ms}}
\label{section:proof:proposition:lmc-to-ms}

Let
\[
    X^{\prime}_{\widehat{G}F} =
    \{
        x \in X_{\widehat{G}F} \mid
        \forall v \in V \colon x_v = 0
        \land
        \forall e \in E \colon x_e = x_{v_{e}}
    \}
\]
denote the set of all feasible vectors in $X_{\widehat{G}F}$ that vanish on $V$ and whose $e$- and $v_{e}$-components take the same value for each $e \in E$.
Since all inequalities $x_v \geq 0$ with $v \in V$ are valid for $x \in X_{\widehat{G}F}$ and $x_e \leq x_{v_{e}}$ are valid for $x \in X_{\widehat{G}F}$ with $x|_V = 0$, the convex hull $\Xi^{\prime}_{\widehat{G}F} \coloneqq \conv X^{\prime}_{\widehat{G}F}$ is a face of the multi-separator polytope $\Xi_{\widehat{G}F}$.

We claim that the lifted multicut polytope with respect to $G$ and $G^{\prime}$ is the image of the face $\Xi^{\prime}_{\widehat{G}F}$ under the projection $\proj_{F} \colon \mathbb{R}^{\widehat{V} \cup F} \to \mathbb{R}^{F}$ which maps $(x_1, x_2) \in \mathbb{R}^{\widehat{V}} \times \mathbb{R}^{F}$ to $x_2$, i.e.,
\[
    \lmc_{GG^{\prime}} = \proj_{F} \Xi^{\prime}_{\widehat{G}F}\,.
\]
For this purpose, it is enough to show that $\lmcfeasible_{GG^{\prime}} = \proj_{F} X^{\prime}_{\widehat{G}F}$, since the operations of constructing the convex hull and taking the projection commute with each other.

To establish $\proj_F X^{\prime}_{\widehat{G}F} \subseteq \lmcfeasible_{GG^{\prime}}$, let $x \in X^{\prime}_{\widehat{G}F}$ and let $f \in F$, we show that those of the inequalities~\eqref{eq:lmc-cycle}~--~\eqref{eq:lmc-cut} associated with $f$ for $\lmcfeasible_{GG^{\prime}}$ can be obtained from the connector and separator inequalities associated with $f$ for $X_{\widehat{G}F}$.
Suppose $H$ is a cycle containing $f$ in $G$ if $f \in E$ or an $f$-path in $G$ if $f \in F \setminus E$.
By \Cref{remark:lmc-ilp}, $H$ can be assumed to be chordless in $G$.
No matter which is the case, the vertex set of $\widehat{H}$ is always a minimal $f$-connector of $\widehat{G}$ by \ref{item:cor-connector-minimal} of \Cref{cor:connector-separator-minimal}.
In this way, each chordless cycle or path inequality associated with $f$ for $\lmcfeasible_{GG'}$ corresponds to a minimal connector inequality associated with $f$ for $X_{\widehat{G}F}$.
Using this fact and the definition of $X^{\prime}_{\widehat{G}F}$, it is straightforward to see that $x$ satisfies the cycle inequalities~\eqref{eq:lmc-cycle} and the path inequalities~\eqref{eq:lmc-path} for $\lmcfeasible_{GG'}$.
Furthermore, since every minimal $f$-connector of $\widehat{G}$, except $f \cup \{v_f\}$ when $f \in E$, arises from such a chordless cycle or path in $G$, it follows from the separator inequalities for $X_{\widehat{G}F}$ that $y$ also satisfies the cut inequalities~\eqref{eq:lmc-cut} and hence its image under $\proj_F$ belongs to $\lmcfeasible_{GG^{\prime}}$.
To see the other direction, i.e.~$\lmcfeasible_{GG^{\prime}} \subseteq \proj_F X^{\prime}_{\widehat{G}F}$, for any $y \in \lmcfeasible_{GG^{\prime}}$ we construct another vector $x \in \{0, 1\}^{\widehat{V} \cup F}$, whose projection to $\mathbb{R}^{F}$ is $y$, as follows:
\[
    \forall v \in V \colon x_v = 0
    \land
    \forall e \in E \colon x_{v_{e}} = y_{e}
    \land
    \forall f \in F \colon x_{f} = y_{f}
    \,.
\]
It is easy to verify that $x$ is feasible and hence $x \in X^{\prime}_{\widehat{G}F}$.
This completes the proof.
\qed

\subsection{Proof of Proposition~\ref{prop:ms-to-lmc}}
\label{section:proof:proposition:ms-to-lmc}

Let
\[
    \lmcfeasible^{\prime}_{\bar{G}\bar{G}^{\prime}} =
    \{
        y \in \lmcfeasible_{\bar{G}\bar{G}^{\prime}} \mid
        \forall v \in e \in E \colon
        y_{\{\bar{v}, v_e\}} = 1
        \textnormal{ and }
        y_{\{v, \bar{v}\}} + y_{\{v, v_e\}} = 1
    \}
\]
denote the set of all those feasible vectors in $\lmcfeasible_{\bar{G}\bar{G}^{\prime}}$ whose induced lifted multicut $M$ satisfy the following condition:
All edges $\{\bar{v}, v_e\}$ are in the lifted multicut $M$ for $v \in e \in E$, and either $\{v, \bar{v}\} \in M$ and $\{v, v_e\} \not\in M$ for all $e \in E$ incident with $v$, or $\{v, \bar{v}\} \not\in M$ and $\{v, v_e\} \in M$ for all $e \in E$ incident with $v$.
The convex hull of $\lmcfeasible^{\prime}_{\bar{G}\bar{G}^{\prime}}$ defines a face of $\lmc_{\bar{G}\bar{G}^{\prime}}$, since all inequalities $y_{e'} \leq 1$ with $e' \in E'$ are valid for $y \in \lmcfeasible_{\bar{G}\bar{G}^{\prime}}$ and $y_{\{v, \bar{v}\}} + y_{\{v, v_e\}} \geq 1$ are valid for those $y \in \lmcfeasible_{\bar{G}\bar{G}^{\prime}}$ satisfying $y|_{E'} = 1$, where $E'$ denotes the set of all edges $\{\bar{v}, v_e\}$ with $v \in e \in E$.
Below we show that $\polytope$ can be obtained as a projection of the face $\conv \lmcfeasible^{\prime}_{\bar{G}\bar{G}^{\prime}}$.
Since the convex hull operator and the projection operator commute, it suffices to show that $\feasible$ is a projection of $\lmcfeasible^{\prime}_{\bar{G}\bar{G}^{\prime}}$.

Let $\tilde{E}$ be an edge subset of the subdivision $\widehat{G}$ such that there is exactly one edge in $\tilde{E}$ incident with each vertex of $G$.
So there is a bijection $\lambda \colon V \cup F \to \tilde{E} \cup F$ mapping $v \in V$ to its unique adjacent edge in $\tilde{E}$ and leaving $f \in F$ unchanged.
This gives rise to a linear isomorphism $\lambda^{*} \colon \mathbb{R}^{\tilde{E} \cup F} \to \mathbb{R}^{V \cup F}$ sending $x \colon \tilde{E} \cup F \to \mathbb{R}$ to the composite function $x \circ \lambda$.
We claim that the image of $\lmcfeasible^{\prime}_{\bar{G}\bar{G}^{\prime}}$ under $\lambda^* \circ \proj_{\tilde{E} \cup F}$ is exactly $\feasible$.
Similar to the proof of \Cref{prop:lmc-to-ms}, for any $f \in F$, every minimal $f$-connector of $G$ is the set of all the vertices in $V$ of some chordless path in $\bar{G}$ and vice versa.
Thus $\lambda^{*}(\proj_{\tilde{E} \cup F} \lmcfeasible^{\prime}_{\bar{G}\bar{G}^{\prime}})$ is contained in $X_{GF}$.
The other direction follows from the fact that any $x \in X_{GF}$ is the image of $y$ under $\lambda^* \circ \proj_{\tilde{E} \cup F}$, where $y$ is the vector in $\lmcfeasible^{\prime}_{\bar{G}\bar{G}^{\prime}}$ defined by
\[
    \forall v \in e \in E \colon y_{\{v, \bar{v}\}} = 1 - x_v \ \textnormal{and} \ y_{\{v, v_e\}} = x_v\,.
\]
Altogether, we have
\[
    X_{GF} =
    \lambda^{*}
    \left(
        \proj_{\tilde{E} \cup F} \lmcfeasible^{\prime}_{\bar{G}\bar{G}^{\prime}}
    \right)
    \,.
\]
Therefore, up to relabeling, $\Xi_{GF}$ is the projection of the face $\conv \lmcfeasible^{\prime}_{\bar{G}\bar{G}^{\prime}}$ of $\lmc_{\bar{G}\bar{G}^{\prime}}$, as desired.
\qed

\subsection{Proof of Proposition~\ref{prop:ms-lmc-path}}
\label{section:proof:proposition:ms-lmc-path}

We can identity $V_n \cup F$ with $E^\prime_{n+1}$ by the bijection $i \mapsto \{i,i+1\}$ for $i \in V_n$ and $\{i,j\} \mapsto \{i,j+1\}$ for $\{i,j\} \in F$ with $i < j$.
By renaming the indices of $x$ in inequalities \eqref{eq:path-connector} and \eqref{eq:path-separator} according to this bijection we obtain precisely the path and cut inequalities for the lifted multicut polytope for paths from \citet{lange2021lifted}.
Therefore there is a one-to-one correspondence between the multi-separator polytope for a path of length $n$ the lifted multicut polytope for a path of length $n+1$.

\subsection{Proof of Proposition~\ref{prop:ms-lmc-path-tdi}}
\label{section:proof:proposition:ms-lmc-path-tdi}

Because of the one-to-one correspondence between system~\eqref{eq:path-box}~--~\eqref{eq:path-intersection} and system (11)~--~(15) from \citet{lange2021lifted}, the desired result follows immediately from Theorem~1 of \citet{lange2021lifted}, which says that system (11)~--~(15) therein is totally dual integral.

\end{document}